\documentclass[10pt]{article}

\usepackage[a4paper,total={176mm,260mm},footnotesep=5mm]{geometry}
\usepackage{blindtext}
\usepackage{color}
\usepackage{amsfonts,amssymb,amsthm,bbm,graphicx,mathtools,amsmath,mathrsfs,dsfont,booktabs,dsfont,dsfont,dsfont,titlesec}
\usepackage[hang]{footmisc}
\usepackage[colorlinks=true,linkcolor=blue,citecolor=magenta,linktocpage=true]{hyperref}
\usepackage{graphicx,float}
\usepackage[font=normalsize]{subfig} 
\usepackage[export]{adjustbox}
\usepackage{comment}
\usepackage{indentfirst}

\titleformat*{\section}{\large\bfseries}
\titleformat*{\subsection}{\normalsize\bfseries}
\titleformat*{\subsubsection}{\normalsize\bfseries}

\usepackage{cite}
\usepackage[most,breakable]{tcolorbox}
\usepackage{soul}
\setstcolor{red}
\usepackage{array}   
\newcolumntype{C}{>{$}c<{$}}

\makeatletter
\renewcommand{\@dotsep}{10000}
\makeatother

\numberwithin{equation}{section}

\newcommand{\C}{{\mathbb C}}
\newcommand{\N}{{\mathbb N}}
\newcommand{\R}{{\mathbb R}}

\newcommand{\cS}{{\mathcal S}}

\newcommand{\SU}{\mathrm{SU}}

\newcommand{\f}{\frac}

\def\be#1\ee{\begin{equation}#1\end{equation}}
\def\beq#1\eeq{\begin{eqnarray}#1\end{eqnarray}}
\def\bea#1\eea{\begin{align}#1\end{align}}

\def\de{\mathrm{d}}

\newtheorem{Definition}{Definition}[section]
\newtheorem{Theorem}[Definition]{Theorem}
\newtheorem{Lemma}[Definition]{Lemma}
\newtheorem{Proposition}[Definition]{Proposition}
\newtheorem{Corollary}[Definition]{Corollary}
\newtheorem{Remarks}[Definition]{Remarks}
\newtheorem{Remark}[Definition]{Remark}
\newtheorem{Example}[Definition]{Example}

\graphicspath{{Figures/}}
\usepackage{lipsum}
\let\oldbibliography\thebibliography
\renewcommand{\thebibliography}[1]{\oldbibliography{#1}
\setlength{\itemsep}{0\baselineskip}}

\usepackage{tocbasic}
\DeclareTOCStyleEntry[
  beforeskip=.4em plus 1pt,
  pagenumberformat=\textbf
]{tocline}{section}
\begin{document}

\title{\Large{\textbf{\sffamily Transition Amplitudes in 3D Quantum Gravity:\\ Boundaries and Holography in the Coloured Boulatov Model}}}

\author{\sffamily
Christophe Goeller$^{1,}$\footnote{Electronic address: \href{mailto:c.goeller@physik.uni-muenchen.de}{c.goeller@physik.uni-muenchen.de}} \,,
Daniele Oriti$^{1,}$\footnote{Electronic address: \href{mailto:daniele.oriti@physik.lmu.de}{daniele.oriti@physik.lmu.de}} \,
and
Gabriel Schmid$^{1,}$\footnote{Electronic address: \href{mailto:gabriel.schmid@physik.uni-muenchen.de}{gabriel.schmid@physik.uni-muenchen.de}}
}
\date{\small{\textit{
$^1$Arnold Sommerfeld Center for Theoretical Physics,\\ Ludwig-Maximilians-Universität München, Theresienstraße 37, 80333 München, Germany\\}}}


\newgeometry{total={176mm,280mm}}

\maketitle

\vspace{-5mm}

\begin{abstract}
    We consider transition amplitudes in the coloured simplicial Boulatov model for three-dimensional Riemannian quantum gravity. First, we discuss aspects of the topology of coloured graphs with non-empty boundaries. Using a modification of the standard rooting procedure of coloured tensor models, we then write transition amplitudes systematically as topological expansions. We analyse the transition amplitudes for the simplest boundary topology, the $2$-sphere, and prove that they factorize into a sum entirely given by the combinatorics of the boundary spin network state and that the leading order is given by graphs representing the closed $3$-ball in the large $N$ limit. This is the first step towards a more detailed study of the holographic nature of coloured Boulatov-type GFT models for topological field theories and quantum gravity.
\end{abstract}

\thispagestyle{empty}
\setcounter{page}{1}

\tableofcontents

\restoregeometry
\newpage
\section*{Introduction}
\addcontentsline{toc}{section}{\hspace{15pt}Introduction}

\vspace*{-0.1cm}
The holographic principle, the study of boundary symmetries, boundary conditions and boundary states have become one of the main points of interest over the last years for many approaches to quantum gravity. The holographic principle, historically motivated from the study of the entropy of black holes \cite{BekensteinBH,HawkingsBH}, in particular from the discovery of the area law, and formulated in its original form by L.~Susskind \cite{Susskind} and G.~'t Hooft \cite{Hooft}, refers to the idea of fully describing a theory in a region of spacetime in terms of a dual theory solely living on its boundary. One of the prime examples is the famous AdS/CFT correspondence \cite{AdSCFT}, which conjectures a duality between (quantum) gravity on $d$-dimensional (asymptotically) anti-de Sitter (AdS) space and a conformal field theory (CFT) on its $(d-1)$-dimensional flat boundary at spatial infinity.
\bigskip

Quantum gravity in three dimensions turns out to be particularly useful when studying holographic dualities. It is an example of a topological field theory (classically as well as quantum mechanically, it only deals with constant curvature geometries, in absence of matter) and it is well-known that it can be formulated as a Chern-Simons theory \cite{ChernSimonsGravity}, or equivalently, as a BF-theory \cite{HorowitzBFTheory}. Due to the absence of local degrees of freedom, it provides us with a simple set-up for studying the interplay between the choice of boundary states and holographic dualities. Recently, there have been many works regarding quasi-local holographic dualities in the context of the Ponzano-Regge spin foam model for three-dimensional quantum gravity \cite{PonzanoReggeModel,BarrettPonzanoRegge,FreidelPonzanoRegge1,FreidelPonzanoRegge2,FreidelPonzanoRegge3}. The term quasi-local means that one is looking at a finite, bounded region of spacetime instead of an asymptotic one, as in the standard AdS/CFT correspondence. Spin foam models are background-independent approaches to quantum gravity, formulated as state sum models, in which one assigns local weights to discrete building blocks of spacetime. The Ponzano-Regge model mentioned above is a particular instance of a spin foam model \cite{BaezBFTheory,PerezSM} for three-dimensional Riemannian quantum gravity without a cosmological constant and can be understood as being the discretization of the quantum partition function of three-dimensional gravity formulated as a BF-theory \cite{FreidelPRDisc}. The model was in fact the first spin foam model ever proposed and has also been related to other approaches to 3d quantum gravity, such as loop quantum gravity (LQG) \cite{NouiPerez} and Chern-Simons theory \cite{OoguriSasakura,FreidelPonzanoRegge2}. Furthermore, it corresponds to the limit of the Turaev-Viro model \cite{TuraevViroModel} for vanishing cosmological constant \cite{BarrettPonzanoRegge}. The Turaev-Viro model, in turn, computes the Reshetikhin-Turaev invariant \cite{TuraevReshetikhin,RobertsTuraevViro}, which reflects the relation between three-dimensional quantum gravity and Chern-Simons theory \cite{Witten1,Witten2,Witten3}. With respect to holographic dualities, it has been shown that the Ponzano-Regge model on a 3-ball is dual to two copies of the two-dimensional Ising model on its boundary 2-sphere, in the sense that the partition function of the Ponzano-Regge model is proportional to the square of the boundary Ising partition function\cite{PRBallIsing1,PRBallIsing2}. In a recent series of paper \cite{TorusPR1,TorusPR2,TorusPR3,TorusPR4,ChristopheThesis}, the Ponzano-Regge model on the solid torus with boundary given by the $2$-torus was systematically studied and related to the BMS group \cite{BMS1,BMS2} --the asymptotic symmetry group of continuum three-dimensional asymptotically flat gravity-- for a boundary state encoding the intrinsic geometry of a solid torus. Both these works provide us with clear insights into the holographic nature of the Ponzano-Regge model.
\bigskip 

When discussing transition amplitudes in quantum gravity models, which are the physical scalar products between two spatial boundary topologies, it is natural to ask whether one should also include a sum over all topologies in addition to a sum over geometries, in order to treat also the topology as a dynamical variable. There are several arguments for the necessity of doing so \cite{TopChangeHorowitz,StringTheoryTopChange,Geons}. The next question however is how to do so in a systematic and controllable manner, in a given quantum gravity framework. In the context of spin foam models, initially defined on a given cellular complex, such a sum over (bulk) topologies can be defined by introducing the corresponding \textit{Group Field Theory} (GFT)  \cite{FreidelGFTOverview,OritiGFTApproach,OritiMicroscopicDynamics}. From the physical point of view, a GFT can be understood as the completion of a given spin foam model in the sense that it gives us a prescription on how to systematically organize the spin foam amplitudes corresponding to different complexes, for different topologies, but also for given topology, since in dimensions higher than three, where gravity is not topological, a restriction to a given complex implies a truncation to a subset of quantum gravity degrees of freedom, that has to be removed to define the full theory. GFTs are quantum field theories \textit{of} spacetime, instead of \textit{on} spacetime. In more technical terms, GFTs are generically non-local field theories defined on (copies of) a Lie group (or quantum group, homogeneous space, etc.) and can be viewed as generalizations of matrix models \cite{MatrixModels1,MatrixModels2} of (pure) two-dimensional quantum gravity to higher dimensions. They can also be understood as generalizations of random tensor models \cite{TensorModels1,TensorModels2,TensorModels3}, enriched with group theoretic data, which allows for imposing additional symmetry properties of their fields and for richer dynamical amplitudes.\footnote{Tensor models and GFTs can be seen as specific examples of models within the common framework of {\it Tensorial Group Field Theories} (TGFTs), defined to encompass all models with a tensorial field (regardless of the domain) and combinatorially non-local interactions (regardless of the specific action), such that the perturbative expansion produces a sum over cellular complexes as Feynman diagrams.}  Furthermore, the quantum states of GFT models are in fact (generalised) tensor networks, thus GFTs can be understood as defining a dynamics (probability distributions) for tensor networks, which in turn have proven themselves very useful to study holographic properties of quantum gravity models \cite{Chirco:2017wgl,Chirco:2019dlx,Colafranceschi:2021acz,Colafranceschi:2022dig,Colafranceschi:2022ual}.
Last but not least, GFT can also be seen as a second quantized formulation of LQG \cite{GFTandLQG1,GFTandLQG2}.
\bigskip

In this paper, we aim at setting up a formalism for studying holographic dualities in Boulatov-Ooguri type GFT models \cite{BoulatovModel,OoguriModel}. Focusing on the Boulatov model --the completion of the Ponzano-Regge model-- describing three-dimensional gravity, we will construct and classify amplitudes for boundary states describing the trivial topology --the sphere. It will allow us to exhibit a clear holographic behaviour of the model, in the sense that the amplitudes will only depend on boundary data. This is an important step in the context of discrete models for quantum gravity with spacetime emerging from more fundamental degrees of freedom. It expands insights from LQG and spin foam models into a broader framework, opening the road towards a better understanding of dualities in GFTs and tensor network models.
\bigskip

A first step towards a study of holographic properties of such models is to define boundary observables and transition amplitudes. For doing so, the {\it coloured} version of the Boulatov model \cite{GurauColouredGFT,BosonicGFT} is most convenient. A colouring of tensor models and GFTs has been proven to be useful for two main reasons. First, the colouring allows full control over the topology of (complexes dual to) the Feynman diagrams of the models. Second, these Feynman diagrams are then dual to manifolds or normal pseudomanifolds (topologies which contain at most isolated and point-like singularities). In other words, coloured GFTs do not produce more singular topologies, which are generically present in uncoloured models and which tend to dominate in power counting \cite{GurauColouredGFTPseudo}. These features also permit the definition of the large $N$ limit \cite{GurauLargeN1,GurauLargeN2,GurauLargeN3} of all such GFT (and tensor) models, the analytic study of the critical behaviour and continuum limit \cite{CritTM}, as well as to derive key universality results showing that the tensors are distributed by a Gaussian in the large $N$ limit \cite{CritTM2}. It has also been observed that colouring might be a crucial ingredient in order to define a suitable notion of a discrete counterpart (or, better, remnant) of diffeomorphism invariance \cite{GFTDiff1,GFTDiff2} in GFT, as a field theoretic counterpart of what has been done in simplicial gravity, e.g.~\cite{DiffeoSM,DiffeoReview}. 

While there is an extensive literature about the topology of \textit{closed} coloured graphs in the context of tensor models and GFTs \cite{GurauBook,GurauColouredTensorModelsReview,GurauColouredTensorModelsReview2,Bonzom:2010zh}, much less is known about \textit{open} coloured graphs, i.e.~graphs admitting external legs. In \cite{GurauBoundaryGraph}, the notion of a boundary graph and its corresponding complex was introduced. A further analysis of open coloured graphs and their degree of divergence can be found for example in \cite{DegreeOpenGraphs,GFTRen,SenseTM} and other works on renormalization in group field theory. However, it turns out that the topology of coloured graphs is not only studied in the context of quantum gravity, but also in \textit{Crystallization Theory} \cite{GagliardiBoundaryGraph,CTReview,Review2018}, a branch of geometric topology. Many result have been obtained in the crystallization theory literature, pioneered by M.~Pezzana, C.~Gagliardi, M.~Ferri and others in the late 1960s and 1970s. In order to find suitable tools for defining transition amplitudes, we will also give a detailed review of techniques developed in crystallization theory for the particular case of general open coloured graphs representing pseudomanifolds with non-empty boundaries, which can be viewed as generalizations of the well-known techniques used in coloured tensor models and GFTs to graphs with external legs.
\bigskip

This paper is organized as follows: In Section \ref{SectionI:Model}, we introduce the coloured (bosonic, simplicial) Boulatov model for three-dimensional quantum gravity and briefly review the different representations of its Feynman graphs. In particular, we systematically define both closed and open coloured graphs and explain their simplicial interpretation. We discuss the Feynman amplitudes corresponding to closed (vacuum) diagrams and briefly review their relation to the Ponzano-Regge spin foam model. This section can also be skipped by readers familiar with general notions of coloured graphs and coloured tensor models/GFTs.

In Section \ref{SectionII:CryTheo}, we turn our attention to open coloured graphs, i.e.~Feynman graphs of the coloured Boulatov model with external legs. We mainly discuss aspects of the topology of coloured graphs with non-empty boundary, based on the literature on crystallization theory. More precisely, we look at the bubble structure of these graphs, the relation between the boundary graph and the boundary complex, as well as moves allowing for transforming one graph into another (in a topology preserving way).

Next, we discuss transition amplitudes of the coloured Boulatov model in Section \ref{SectionIII:TransAmpl}. First of all, we define suitable boundary observables out of spin network states living on some fixed boundary graph representing a fixed topology. Using these observables, we then define transition amplitudes, which are given by a sum over all bulk topologies with respect to the fixed boundary graph. Afterwards, we rewrite this sum as a topological expansion, using a similar rooting procedure as introduced by R.~Gurau to study the large $N$ limit of the free energy.

In Section \ref{SectionIV:Sphere}, we apply the formalism to the simplest boundary topology, the $2$-sphere. We show that the transition amplitude factorizes into a sum entirely given by the combinatorics of the boundary spin network state. More precisely, we see that every manifold with spherical boundary has a contribution proportional to the spin network evaluation. We end this section by quickly discussing the case of the boundary $2$-torus to illustrate why the previous result is not just a consequence of the topological nature of the theory (which would diminish its general interest), but it is due to the simple topology of the chosen boundary, so that one can expect a similar holographic behaviour, but more intricate details of the map, for more involved topologies.

Finally, in Section \ref{SectionV:LeadingOrder}, we show that the leading order contribution to the transition amplitude of some spherical boundary graph, when restricted to manifolds, is given by certain graphs representing the closed $3$-balls. We show that these graphs generalize the melonic graphs from the large $N$ limit of coloured tensor models, in the sense that they are exactly those graphs for which a suitable generalization of the Gurau degree to open graphs vanishes.

In Appendix \ref{AppendixA:PseudoManifolds}, the reader can find a short discussion of pseudomanifolds and an overview of the terminology used for simplicial complexes. Furthermore, we give some further details on the topology of coloured graphs with non-empty boundaries by reviewing general existence theorems of crystallization theory and by discussing a connected sum operation in Appendix \ref{AppendixB:CryTheo}. Appendix \ref{AppendixC:SolidTorusGraphs} contains instead a derivation of a family of open coloured graphs representing the solid torus.
%
%
%
%
%
\section{The Coloured Boulatov Model}
\label{SectionI:Model}

\textit{This section mainly introduces notations, definitions and standard properties of coloured GFT model and their Feynman graphs, and can be safely skipped for readers familiar with the subject. For the notation of a particular set of coloured graphs, which we will use throughout the present paper, see Definition \ref{def:notation_graph}.}
\bigskip

The Boulatov model \cite{BoulatovModel} is defined using a single ($\R$-valued) bosonic scalar field. The colour extension of the model  \cite{GurauColouredGFT,BosonicGFT} were shown to be very useful for studying, for example, exact power counting \cite{GurauColouredGFTPseudo}, the large $N$ limit \cite{GurauLargeN1,GurauLargeN2,GurauLargeN3} and the critical behaviour and continuum limit \cite{CritTM}. In this paper, we consider the bosonic version of the model \cite{BosonicGFT,GFTDiff2,GFTVertex}. The bosonic model lacks an $\mathrm{SU}(4)$ colour symmetry of the fermionic one \cite{GurauColouredGFT,BosonicGFT} but this does not change the combinatorial structure of the Feynman diagrams, nor their amplitudes.
\bigskip

In this section, we start with the definition of the model, then we discuss the structure of its Feynman diagrams with and without external legs and discuss the amplitudes of closed (vacuum) diagrams. Furthermore, we review briefly the relation to the Ponzano-Regge spin foam model \cite{PonzanoReggeModel,BarrettPonzanoRegge,FreidelPonzanoRegge1,FreidelPonzanoRegge2,FreidelPonzanoRegge3}.

\subsection{Definition of the Model}
Let $\{\varphi_{l}\}_{l=0}^{3}\subset L^{2}(\SU(2)^{3},\de g;\C)$, with $\de g$ the normalized $\SU(2)$ Haar measure, be four bosonic and $\C$-valued scalar fields defined on three copies of $\SU(2)$. They are labelled by a ``\textit{colour index}'' $l\in\{0,\dots,3\}$ and we assume that they are $\SU(2)$ gauge invariant, i.e.
\begin{align}
    \forall h\in\SU(2) \quad \varphi_{l}(hg_{1},hg_{2},hg_{3})=\varphi_{l}(g_{1},g_{2},g_{3})
\end{align}
for all $g_{1},g_{2},g_{3}\in\SU(2)$ and $l\in\{0,1,2,3\}$.\footnote{The choice of imposing either right or left translation is just a convention as one can always redefine $\varphi_{l}\to\widetilde{\varphi}_{l}$, where $\widetilde{\varphi}_{l}(g_{1},g_{2},g_{3}):=\varphi_{l}(g_{1}^{-1},g_{2}^{-1},g_{3}^{-1})$ are now right-invariant fields. The action does not change under this transformation, i.e.~$\cS_{\lambda}[\varphi_{l},\overline{\varphi}_{l}]=\cS_{\lambda}[\widetilde{\varphi}_{l},\overline{\widetilde{\varphi}}_{l}]$ by unimodularity of the Haar measure of compact Lie groups.} Note that we do not assume any supplementary invariance of the fields. In particular, we do not assume any action of the permutation group (or any of its subgroups) leaving them invariant. Such assumption often appears in the uncoloured case to guarantee that only \textit{orientable} simplicial complexes are produced \cite{DePietriPetronio}. In the coloured case however, this is already guaranteed by takings the fields to be complex. Additionally, the colouring allows to describe the Feynman diagrams as bipartite edge-coloured graphs.
We define the $\mathrm{SU}(2)$ delta function at some cutoff\footnote{See \cite{Freidel:2009hd,Carrozza:2016vsq} and references therein for a discussion about this cut-off and its application for integration purposes.} $N \in \N/2$ using the Plancherel decomposition following \cite{OoguriSasakura}
\begin{align}
    \SU(2)\ni g\mapsto\delta^{N}(g):=\sum_{j\in\mathbb{N}/2, j\leq N}(2j+1)\chi^{j}(g) \; ,
\end{align}
where $\chi^{j}$ denote the characters of the unitary and irreducible representations of $\mathrm{SU}(2)$, labelled by spins $j\in\N/2$. The action of the coloured Boulatov model is then defined by
\begin{equation}
    \begin{aligned}
        \label{BoulatovAction}
        \cS_{\lambda}[\varphi_{l},\overline{\varphi}_{l}]
        :=
        & \sum_{l=0}^{3}\int_{\mathrm{SU}(2)^{3}}\bigg (\prod_{i=1}^{3}\mathrm{d}g_{i}\bigg )\,\vert\varphi_{l}(g_{1},g_{2},g_{3})\vert^{2} \\
        & -\frac{\lambda}{\sqrt{\delta^{N}(\mathds{1})}}\int_{\mathrm{SU}(2)^{6}}\bigg (\prod_{i=1}^{6}\mathrm{d}g_{i}\bigg )\,\varphi_{0}(g_{1},g_{2},g_{3})\varphi_{1}(g_{3},g_{4},g_{5})\varphi_{2}(g_{5},g_{2},g_{6})\varphi_{3}(g_{6},g_{4},g_{1}) + c.c.\; ,
    \end{aligned}
\end{equation}
where $\mathds{1}$ denotes the identity of $\SU(2)$. The scaling in the action coincides with \cite{GurauLargeN1,GurauLargeN2,GurauLargeN3,GFTVertex,Caravelli,Caravelli2} and is chosen in order for maximally divergent graphs to have a uniform degree of divergence at all orders. Indeed, providing this scaling, the degree of divergence of Feynman graphs is independent under a certain type of transformation, called ``internal proper $1$-dipole moves'', as we will discuss later on (see Subsection \ref{Subse:Rooting}).
\bigskip


The geometric interpretation of the action \eqref{BoulatovAction} is shown in figure \ref{GeoInt}. First, note that each field $\varphi_{l}(g_{1},g_{2},g_{3})$ encodes the kinematics of a quantum triangle described by three dual edges labelled by $g_{1},g_{2},g_{3}$ \cite{OritiMicroscopicDynamics,Barbieri}. In other words, the GFT field $\varphi_{l}$ lives on the space of possible geometries of the triangle. Having four distinct fields, we have four different triangles, labelled by the field colour index $l$. The four kinetic terms represent the gluing of two triangles of the same colour while the two interaction terms describe the gluing of four triangles along their edges such that they form a tetrahedron ($3$-simplex). We therefore have two different types of tetrahedra, one for the $\varphi_{l}$-fields and one for the $\overline{\varphi}_{l}$-fields, corresponding to the two different choices of orientation of a tetrahedron.

\begin{figure}[H]
    \captionsetup[subfigure]{labelformat=empty,position=top}
    \centering
    \subfloat{\includegraphics[clip,trim=0cm -0.5cm 0 0cm,scale=1.2]{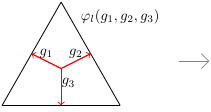}}\hspace{0.5cm}
    \subfloat[$\mathcal{L}_{\lambda,\mathrm{int}}\text{[}\varphi_{l}\text{]}\propto\varphi_{l}^{4}$]{\includegraphics[clip,trim=1cm 1.8cm 0 0,scale=0.3]{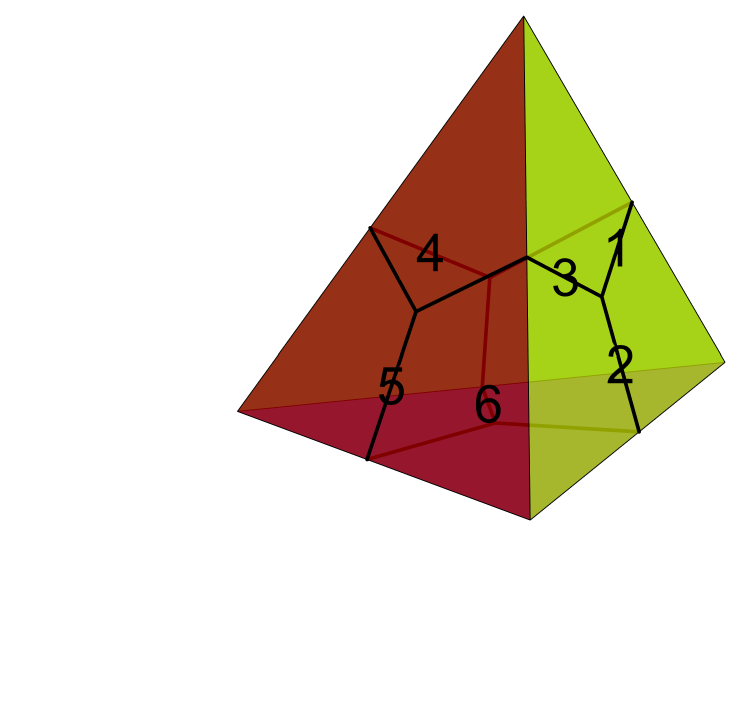}}\hspace{-0.5cm}
    \subfloat[$\mathcal{L}_{\lambda,\mathrm{int}}\text{[}\overline{\varphi}_{l}\text{]}\propto\overline{\varphi}_{l}^{4}$]{\includegraphics[clip,trim=1cm 1.8cm 0 0cm,scale=0.3]{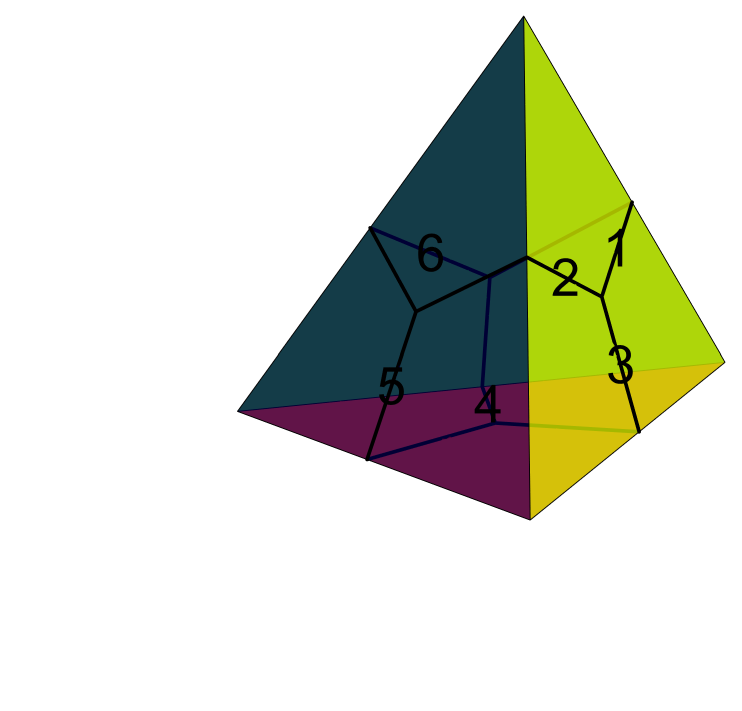}}
    \vspace{-0.5cm}
    \caption{The fields $\varphi_{l}$ describe triangles, equipped with a corresponding colour index $l$, and the interaction terms produce tetrahedra with opposite orientation.\label{GeoInt}}
\end{figure}

\subsection{Feynman Graphs: Closed and Open Coloured Graphs}
As usual in GFT, Feynman graphs can be represented as ``stranded diagrams'' \cite{FreidelGFTOverview,OritiGFTApproach,OritiMicroscopicDynamics}. Figure \ref{fig:StrandedDiagrams} shows the two interaction vertices together with their geometrical interpretation.

\vspace*{-0.5cm}
\begin{figure}[H]
    \captionsetup[subfigure]{labelformat=empty}
    \centering
    \subfloat[$\mathcal{L}_{\lambda,\mathrm{int}}\text{[}\varphi_{l}\text{]}\propto\varphi_{l}^{4}$]{\includegraphics[width=0.15\textwidth]{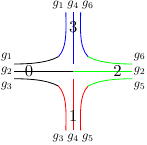}}
    \subfloat{\includegraphics[trim=1cm 2cm 0 0,width=0.2\textwidth]{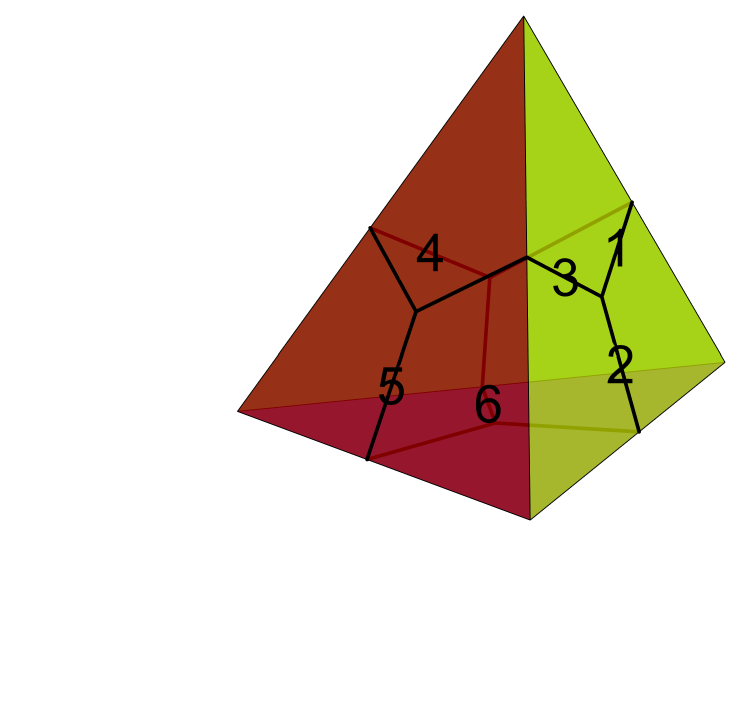}}\hspace{0.5cm}
    \subfloat[$\mathcal{L}_{\lambda,\mathrm{int}}\text{[}\overline{\varphi}_{l}\text{]}\propto\overline{\varphi}_{l}^{4}$]{\includegraphics[width=0.15\textwidth]{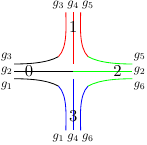}}
    \subfloat{\includegraphics[trim=1cm 2cm 0 0,width=0.2\textwidth]{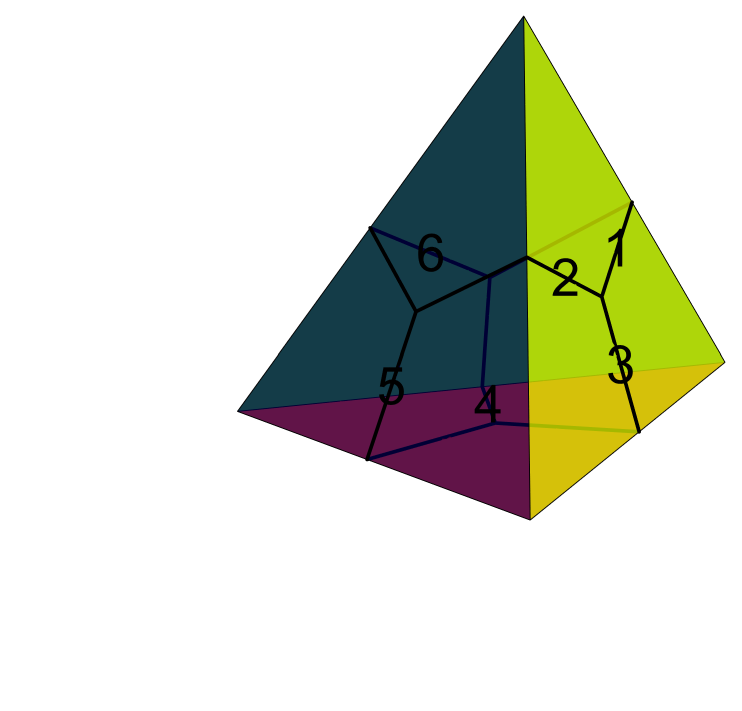}}\hspace{0.5cm}
    \caption{Interaction vertices of the coloured Boulatov model drawn in their stranded diagram representation and their corresponding geometric interpretation.}
    \label{fig:StrandedDiagrams}
\end{figure}

Each strand of colour $i$ represents a triangle of colour $i$ and a free line of colour $ij$ represents an edge, which connects the triangles of colours $i$ and $j$. Since we have not assumed any additional symmetry properties of the field arguments, the structure of the kinetic term tells us that we can glue two faces of the same colour belonging to two different tetrahedra only in a \textit{unique} way: in the stranded picture, a free line with colours $i,j\in\{0,1,2,3\}$ is always glued to a free line with the same pair of colours. Geometrically, it means that the colouring of faces of a tetrahedron induces a colouring of its vertices, obtained by labelling each vertex with the colour of the opposite triangle in the tetrahedron. The gluing of two faces is then such that all the colours of vertices agree. The stranded structure of the Feynman diagrams is therefore rigid and there are no twists within the strands such that we can collapse each strand to a single thin edge and represent Feynman graphs equivalently as edge-coloured graphs, see figure \ref{fig:edge_coloured_graphs}.

\begin{figure}[H]
    \captionsetup[subfigure]{labelformat=empty}
    \centering{\includegraphics[width=0.7\textwidth]{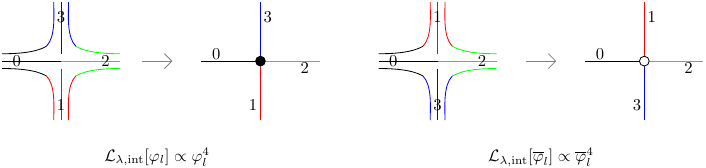}}
    \caption{Feynman graphs of the coloured Boulatov model can equivalently be viewed as coloured graphs.}
    \label{fig:edge_coloured_graphs}
\end{figure}

In this graphical representation, tetrahedra are represented as vertices and the coloured edges of the graph represent the corresponding coloured triangles. Whenever two vertices are connected by an edge of colour $i$, the corresponding tetrahedra are glued together on their faces of colour $i$ in the unique way explained above. Let us discuss the structure of these graphs in a more systematic way. To start with, let us briefly set up the following terminology from graph theory, which we will use throughout the paper:

\begin{itemize}
    \item A ``\textit{graph}'' is always meant to be a multigraph without loops. More precisely, this means that a graph is defined as a pair $\mathcal{G}=(\mathcal{V}_{\mathcal{G}},\mathcal{E}_{\mathcal{G}})$, where $\mathcal{V}_{\mathcal{G}}$ is a set called the ``\textit{vertex set}'' and where $\mathcal{E}_{\mathcal{G}}$ is a multiset containing sets of the form $\{v,w\}\in\mathcal{V}_{\mathcal{G}}\times\mathcal{V}_{\mathcal{G}}$, called the ``\textit{edge set}''. Allowing $\mathcal{E}_\mathcal{G}$ to be a multiset means that two vertices can be connected by several edges. However, note that an edge is by definition a proper set, which means that we do not allow for tadpole lines, i.e.~edges starting and ending at the same vertex.
    \item A graph $\mathcal{G}$ is called ``\textit{bipartite}'' if there is a partition $\mathcal{V}_{\mathcal{G}}=V_{\mathcal{G}}\cup\overline{V}_{\mathcal{G}}$ such that every edge connects a vertex in $V_{\mathcal{G}}$ with a vertex in $\overline{V}_{\mathcal{G}}$. If in addition $\vert V_{\mathcal{G}}\vert=\vert \overline{V}_{\mathcal{G}}\vert$, the graph is called ``\textit{balanced}''.
    \item A ``\textit{$(d+1)$-edge-colouring}'' is a map $\gamma:\mathcal{E}_{\mathcal{G}}\to\mathcal{C}_{d}$, where $\mathcal{C}_{d}$ is some set with cardinality $\vert\mathcal{C}_{d}\vert=d+1$, called the ``\textit{colour set}''. In the following, we will choose $\mathcal{C}_{d}:=\{0,\dots,d\}$ for definiteness. An edge-colouring is called ``\textit{proper}'' if $\gamma(e_{1})\neq\gamma(e_{2})$ for all edges $e_{1},e_{2}\in\mathcal{E}_{\mathcal{G}}$ incident to the same vertex $v\in\mathcal{V}_{\mathcal{G}}$
\end{itemize}

For the sake of generality, in the remaining of this section, we will consider the general $d$-dimensional case unless specified otherwise. The following discussion also applies to higher-dimensional Boulatov-Ooguri type models. Closed (vacuum) Feynman diagrams of the coloured Boulatov model are ``closed coloured graphs''.

\begin{Definition}[Closed Coloured Graphs]
    \label{ClosedColouredGraph} A ``closed $(d+1)$-coloured graph'' is a pair $(\mathcal{G},\gamma)$, where $\mathcal{G}$ is a $(d+1)$-valent and bipartite graph $\mathcal{G}=(\mathcal{V}_{\mathcal{G}},\mathcal{E}_{\mathcal{G}})$ and where $\gamma:\mathcal{E}_{\mathcal{G}}\to\mathcal{C}_{d}$ is a proper $(d+1)$-edge colouring of $\mathcal{G}$.
\end{Definition}

\begin{Remarks}
    \label{Remarks:Balanced}
    \begin{itemize}
        \item[]\item[(a)]In the following, we usually omit writing the colouring map $\gamma$ explicitly and we simply call $\mathcal{G}$ a closed $(d+1)$-coloured graph.
        \item[(b)]A closed $(d+1)$-coloured graph $\mathcal{G}$ is always balanced, i.e.~$\vert V_{\mathcal{G}}\vert=\vert \overline{V}_{\mathcal{G}}\vert$. To see this, observe that the graph obtained by deleting all the edges of colours $i\neq 0$ results into a disconnected graph containing pairs of vertices, which are connected by an edge of colour $0$. In other words, vertices always come in pairs.
    \end{itemize}
\end{Remarks}

The following figure shows four examples of closed $(3+1)$-coloured graphs representing $3$-manifolds \cite{GagliardiBoundaryGraph,LinsCrystallization}.

\begin{figure}[H]
    \captionsetup[subfigure]{labelformat=empty}
    \centering
    \subfloat[$S^{3}$]{\includegraphics[width=0.1\textwidth]{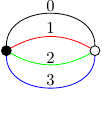}}\hspace{1cm}
    \subfloat[$S^{1}\times S^{2}$]{\includegraphics[width=0.13\textwidth]{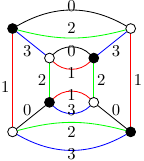}}\hspace{1cm}
    \subfloat[$\mathbb{R}P^{3}$]{\includegraphics[width=0.15\textwidth]{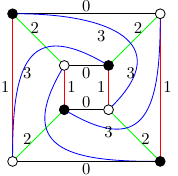}}\hspace{1cm}
    \subfloat[$L(3,1)$]{\includegraphics[width=0.15\textwidth]{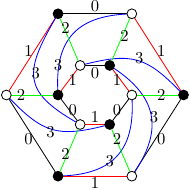}}
    \caption{Four closed $(3+1)$-coloured graphs representing the manifolds $S^{3}$, $S^{1}\times S^{2}$, $\mathbb{R}P^{3}$ and $L(3,1)$. \label{Closed3Manifolds}}
\end{figure}

In order to define transition amplitudes, we also have to discuss open (non-vacuum) Feynman graphs, i.e.~Feynman graphs, which admit external legs.

\begin{Definition}[Open Coloured Graphs]
    \label{OpenColouredGraph}
    An open $(d+1)$-coloured graph is a finite, bipartite and proper $(d+1)$-edge-coloured graph $\mathcal{G}=(\mathcal{V}_{\mathcal{G}},\mathcal{E}_{\mathcal{G}})$ with the following extra property: the vertex set admits a decomposition $\mathcal{V}_{\mathcal{G}}=\mathcal{V}_{\mathcal{G},\mathrm{int}}\cup\mathcal{V}_{\mathcal{G},\partial}$, where $\mathcal{V}_{\mathcal{G},\mathrm{int}}$ consists of $(d+1)$-valent vertices, called ``internal vertices'', and where $\mathcal{V}_{\mathcal{G},\partial}$ consists of $1$-valent vertices, which we call ``boundary vertices''.
\end{Definition}

As a consequence, the edge set of an open $(d+1)$-coloured graph $\mathcal{G}$ can be decomposed as $\mathcal{E}_{\mathcal{G}}=\mathcal{E}_{\mathcal{G},\mathrm{int}}\cup\mathcal{E}_{\mathcal{G},\partial}$, where edges in $\mathcal{E}_{\mathcal{G},\mathrm{int}}$, called ``\textit{internal edges}'', connect two internal vertices and an edge in $\mathcal{E}_{\mathcal{G},\partial}$ --an \textit{external leg}-- connects an internal vertex with a boundary vertex.

\begin{Remarks}
    \begin{itemize}
        \item[]
        \item[(a)]An open coloured graph is in general not balanced. As an example, take the open $(d+1)$-coloured graph consisting of a single $(d+1)$-valent vertex with $(d+1)$ external legs, which represents a single $d$-simplex.
        \item[(b)]There are other conventions for open graphs in the literature. Some authors define open graphs to be ``pregraphs'', in which external legs are defined to be half-edges, i.e.~they do not end at a $1$-valent vertex (e.g.~in \cite{SenseTM}). Furthermore, open graphs in crystallization theory are usually defined without external legs at all, i.e.~they define graphs with two types of vertices: ``Internal'' $(d+1)$-valent vertices and vertices with valency $\leq d$, which they then call ``boundary vertices'' \cite{GagliardiBoundaryGraph,Gagliardi87}.
    \end{itemize}
\end{Remarks}

Having defined the notion of coloured graphs, we can now define the corresponding simplicial complex. For the terminology and notation used for complexes and PL-manifolds, see Appendix \ref{AppendixA:PseudoManifolds}. For completeness, we summarize the construction in the following definition\footnote{Strictly speaking, complexes dual to coloured graphs are \textit{pseudo(simplicial)-complexes} \cite{SeifertTopology}, since two $d$-simplices can share more than one face. As usual in the GFT literature, we won’t make such a distinction and just speak about “simplicial complexes”}:

\begin{Definition}
    Let $\mathcal{G}=(\mathcal{V}_{\mathcal{G}},\mathcal{E}_{\mathcal{G}})$ be some open or closed $(d+1)$-coloured graph. Then we define its dual simplicial complex $\Delta_{\mathcal{G}}$ in the following way:
    \begin{itemize}
        \item[(1)]Assign a $d$-simplex $\sigma_{v}$ to each vertex $v\in\mathcal{V}_{\mathcal{G}}$ and colour the $(d-1)$-faces of $\sigma_{v}$ by $d+1$ colours. This induces a vertex colouring, where each vertex is labelled by the colour of the $(d-1)$-face on the opposite.
        \item[(2)]If two vertices $v,w$ in $\mathcal{G}$ are connected by an edge of colour $i\in\mathcal{C}_{d}$, we glue the two $d$-simplices together along their $(d-1)$-face of colour $i$ in the unique ways such that all the colours of vertices agree.
    \end{itemize}
\end{Definition}

The underlying graph of some open $(d+1)$-coloured graph is nothing else than the internal dual $1$-skeleton of the simplicial complex $\Delta_{\mathcal{G}}$. The boundary dual $1$-skeleton can be read off as follows \cite{GagliardiBoundaryGraph,GurauBoundaryGraph}:

\begin{Definition}[Boundary Graph]
    Let $\mathcal{G}$ be an open $(d+1)$-coloured graph. Then we define the ``boundary graph'' $\partial\mathcal{G}$ as follows: there is a vertex in $\partial\mathcal{G}$ for each external leg in $\mathcal{G}$ and each vertex has a colour coming from the colour of the corresponding external leg. Two vertices of $\partial\mathcal{G}$ are connected by a bicoloured edge of colour $ij$ whenever there is a bicoloured path in $\mathcal{G}$ with colours $i,j$ starting and ending at the corresponding external legs.
\end{Definition}

The following figure shows an example of an open $(3+1)$-coloured graph together with its boundary graph $\partial\mathcal{G}$ and its simplicial complex $\Delta_{\mathcal{G}}$.\footnote{We will omit drawing $1$-valent boundary vertices in open graphs in order to make the concept of external legs more visible.}
\vspace{-0.8cm}

\begin{figure}[H]
    \captionsetup[subfigure]{labelformat=empty}
    \centering
    \subfloat[$\mathcal{G}$]{\includegraphics[width=0.18\textwidth]{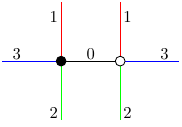}}\hspace{0.5cm}
    \subfloat[$\partial\mathcal{G}$]{\includegraphics[width=0.23\textwidth]{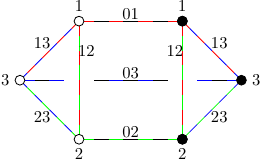}}\hspace{0.5cm}
    \subfloat[$\Delta_{\mathcal{G}}$]{\includegraphics[width=0.23\textwidth]{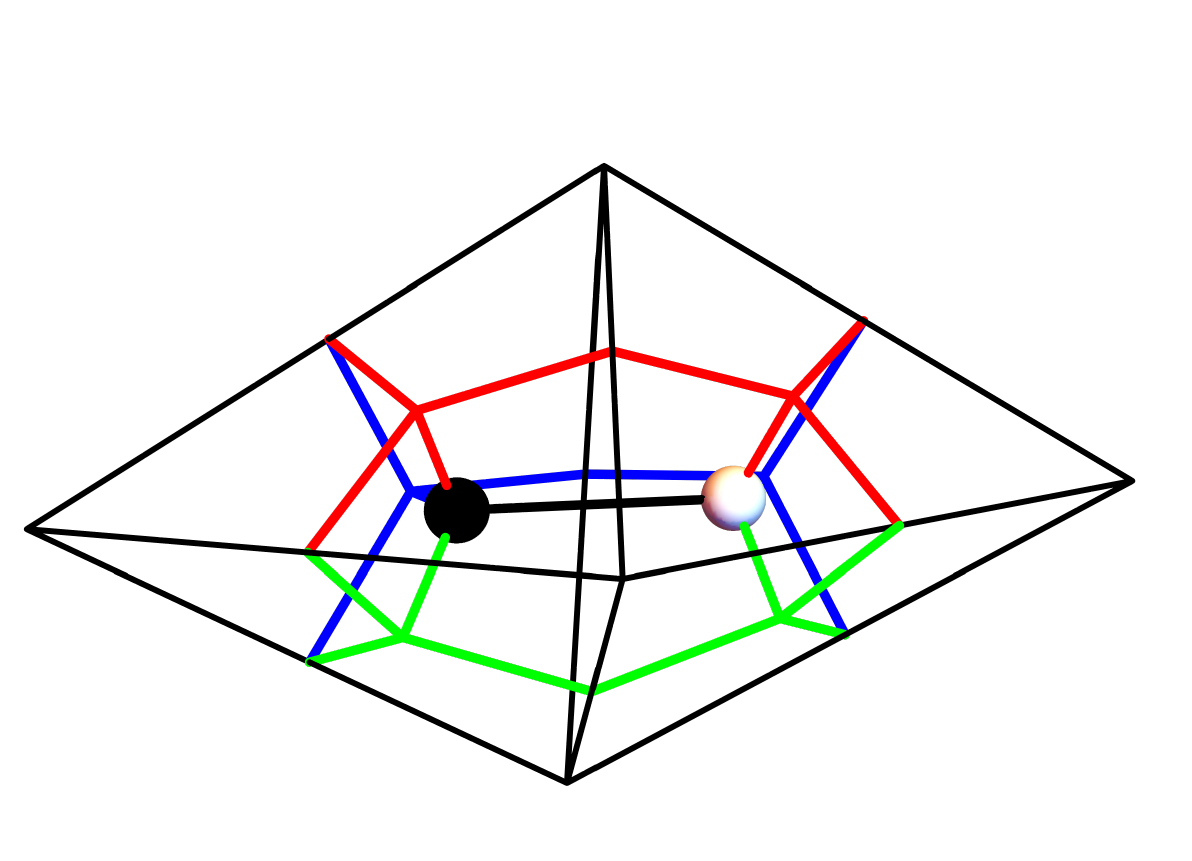}}
    \caption{An open $(3+1)$-coloured graph $\mathcal{G}$ with its boundary graph $\partial\mathcal{G}$ and its corresponding simplicial complex $\Delta_{\mathcal{G}}$ (drawn with its dual $1$-skeleton).\label{OpenGraphFig}}
\end{figure}

\begin{Remark}
    \label{RemarkBoundaryGraph}
    A boundary graph of some open $(d+1)$-coloured graph is always $d$-valent but is in general neither proper edge-coloured nor bipartite (see the example in figure \ref{OpenGraphFig}). However, every boundary vertex has a colour $i\in\mathcal{C}_{d}$ and its $d$ adjacent edges have colours $\{ij\mid j\in\mathcal{C}_{d}\backslash\{i\}\}$. Note that this implies that an edge of colour $ij$ can only connect vertices of colours $\{i,j\}$, $\{i,i\}$ or $\{j,j\}$.
\end{Remark}

As already mentioned, the advantages of working with coloured models is the fact that we only produce pseudomanifolds and no other types of topological singularities. This is summarized in the following theorem:

\begin{Theorem}
    Let $\mathcal{G}$ be an open $(d+1)$-coloured graph. Then $\vert\Delta_{\mathcal{G}}\vert$ is an orientable and normal pseudomanifold with boundary.
\end{Theorem}

\begin{proof}
    The proof that a graph represents a normal pseudomanifold  for the closed case can be found in \cite{GurauColouredGFTPseudo}. A generalization for the open case is straightforward. For orientability, see for example, \cite{GagliardiCorbodant, GagliardiBoundaryGraph} and \cite{Caravelli}.
\end{proof}

\begin{Remark}
    In the case of \textit{real} coloured GFTs, we are also producing non-orientable manifolds since for coloured graphs orientability is equivalent to bipartiteness \cite{GagliardiCorbodant, GagliardiBoundaryGraph}. In that sense, working with complex models seems to be more natural from a physical point of view.
\end{Remark}

In the following, it will be more convenient to restrict to those open coloured graphs for which the boundary graph becomes again a closed coloured graph as defined in Definition \ref{ClosedColouredGraph}. This condition can be imposed using the following proposition:

\begin{Proposition}
    Let $\mathcal{G}$ be an open $(d+1)$-coloured graph with the property that all external legs have the same colour. Then the boundary graph $\partial\mathcal{G}$ is a closed $d$-coloured graph as defined in Definition \ref{ClosedColouredGraph} and $\mathcal{G}$ is bipartite and balanced.
\end{Proposition}

\begin{proof}
    If all external legs of $\mathcal{G}$ have the same colour, say $0$, then there is no information encoded in the vertex colouring of $\partial\mathcal{G}$ and we can ignore it. Furthermore, all the edges of $\partial\mathcal{G}$ are coloured by $0i$ for some $i\in\mathcal{C}_{d}\backslash\{0\}$ and hence, we can just colour them by $i$. This shows that $\partial\mathcal{G}$ admits an obvious proper $d$-edge colouring $\gamma_{\partial}:\mathcal{E}_{\partial\mathcal{G}}\to\mathcal{C}_{d-1}^{\ast}$ induced by the colouring $\gamma$ of $\mathcal{G}$, where $\mathcal{C}_{d-1}^{\ast}:=\{1,\dots,d\}$. To see that $\partial\mathcal{G}$ is bipartite, observe that every edge in $\partial\mathcal{G}$ comes from a bicoloured path of $\mathcal{G}$, which starts and ends at an external leg of the same colour. The number of edges contained in this path is odd, which means that the number of vertices contained in this path is even. Therefore, the source and target vertex of an edge of $\partial\mathcal{G}$ are of different kind. For the second claim, note that the graph $\mathcal{G}^{\prime}$ obtained from $\mathcal{G}$ by deleting all the edges of colour $0$ is in this case a (possibly disconnected) $d$-valent and proper $d$-edge coloured graph and such a graph is always balanced (by similar arguments as in Remark \ref{Remarks:Balanced}(b)).
\end{proof}

\begin{Remark}
    Note that in crystallization theory, open graphs are usually defined directly with the property that all their external legs have the same colour \cite{GagliardiBoundaryGraph,Gagliardi87}. Furthermore, also in tensor models using a single, uncoloured, tensor with bubble interactions, Feynman graphs are (open) coloured graphs of this type \cite{CritTM3,GurauBook}.
\end{Remark}

From now on we will mainly work with this restricted class of graphs and so we introduce the following notation:

\begin{Definition}
    \label{def:notation_graph}
    We will denote by $\mathfrak{G}_{d}$ the set of all open $(d+1)$-coloured graphs in which all external legs have colour $0$. The subset of closed $(d+1)$-coloured graphs is denoted by $\overline{\mathfrak{G}}_{d}\subset\mathfrak{G}_{d}$.
\end{Definition}

An immediate consequence of the definition is

\begin{Lemma}
    If $\mathcal{G}\in\mathfrak{G}_{d}$, then $\partial\mathcal{G}\in\overline{\mathfrak{G}}_{d-1}$. Furthermore, $\partial\mathcal{G}$ is the empty graph if and only if $\mathcal{G}\in\overline{\mathfrak{G}}_{d}$. In particular, this means that $\partial(\partial\mathcal{G})$ is the empty graph for every $\mathcal{G}\in\mathfrak{G}_{d}$.
\end{Lemma}


\subsection{Feynman Amplitudes of Closed Graphs and Ponzano-Regge Model}
The generating functional of the coloured Boulatov model is given by the path integral \cite{BoulatovModel,OritiGFTApproach,OritiMicroscopicDynamics}
\begin{align}
    \mathcal{Z}_{\mathrm{cBM}}=\int\,\bigg(\prod_{l=0}^{3}\,\mathcal{D}\varphi_{l}\mathcal{D}\overline{\varphi}_{l}\bigg)\,e^{-\mathcal{S}_{\lambda}[\varphi_{l},\overline{\varphi}_{l}]}=\sum_{\mathcal{G}\in\overline{\mathfrak{G}}_{3}}\frac{1}{\mathrm{sym}(\mathcal{G})}\mathcal{A}^{\lambda}_{\mathcal{G}} \; ,
\end{align}
where $\mathrm{sym}(\mathcal{G})$ denotes the symmetry factor of the graph $\mathcal{G}$. The Feynman amplitude $\mathcal{A}_{\mathcal{G}}^{\lambda}$ corresponding to some closed $(3+1)$-coloured graph $\mathcal{G}\in\overline{\mathfrak{G}}_{3}$ can be derived by convoluting the propagators and interaction kernels, which can be read off the action \eqref{BoulatovAction} and are given by
\begin{align}
    \mathcal{P}(\{g_{i}\},\{\widetilde{g}_{i}\})=\int_{\mathrm{SU}(2)}\,\mathrm{d}h\,\prod_{i=1}^{3}\delta(g^{-1}_{i}h\widetilde{g}_{i})\hspace{0.5cm}\text{and}\hspace{0.5cm}\mathcal{V}(\{g_{ij}\})=\int_{\mathrm{SU}(2)^{4}}\,\bigg(\prod_{i=1}^{4}\mathrm{d}h_{i}\bigg )\,\prod_{i\neq j}\delta(g_{ij}^{-1}h_{i}^{-1}h_{j}g_{ji}) \; ,
\end{align}
where $g_{ij}$ is the group element assigned to the dual edge living on the triangle $i$ of colour $ij$. The amplitude $\mathcal{A}_{\mathcal{G}}^{\lambda}$ is then precisely the partition function of the Ponzano-Regge spin foam model \cite{PonzanoReggeModel,BarrettPonzanoRegge,FreidelPonzanoRegge1,FreidelPonzanoRegge2,FreidelPonzanoRegge3}  multiplied by a prefactor depending on $N$ and $\lambda$ coming from the interaction term:
\begin{align}
    \mathcal{A}^{\lambda}_{\mathcal{G}}=\bigg (\frac{\lambda\overline{\lambda}}{\delta^{N}(\mathds{1})}\bigg )^{\frac{\vert\mathcal{V}_{\mathcal{G}}\vert}{2}}\int_{\mathrm{SU}(2)^{\vert\mathcal{E}_{\mathcal{G}}\vert}}\,\bigg(\prod_{e\in\mathcal{E}_{\mathcal{G}}}\mathrm{d}h_{e}\bigg )\,\prod_{f\in\mathcal{F}_{\mathcal{G}}}\delta^{N}\bigg (\overrightarrow{\prod_{e\in f}}h_{e}^{\varepsilon(e,f)}\bigg ) \; ,
\end{align}
where $\mathcal{F}_{\mathcal{G}}$ denotes the ``\textit{set of faces}'' of the graph $\mathcal{G}$, i.e.~the bicoloured paths within $\mathcal{G}$, where we write $e\in f$ for an edge belonging to the face $f$ and where $\varepsilon(e,f)$ is equal to $1$ if the orientation of $e$ and $f$ agrees and $-1$ otherwise\footnote{Note that a coloured graph can always be assigned a canonical orientation of edges, i.e.~by orienting each edge from a black vertex to a white vertex. Furthermore, we have implicitly chosen a starting point in the product for each face. The amplitude does not depend on these choices, by the properties of the Haar measure and delta function \cite{FreidelPonzanoRegge2}.}. The amplitudes above take the standard spin foam expression in terms of irreducible representations of the rotation group, once expanded using the Peter-Weyl decomposition of functions on the group \cite{PonzanoReggeModel,BarrettPonzanoRegge,FreidelPonzanoRegge1,FreidelPonzanoRegge2,FreidelPonzanoRegge3}. The ``\textit{free energy}'' of the model is given by
\begin{align}
    \mathcal{F}_{\mathrm{cBM}}[\lambda]=\mathrm{ln}(\mathcal{Z}_{\mathrm{cBM}})=\sum_{\mathcal{G}\in\overline{\mathfrak{G}}_{3}\text{ connected }}\frac{1}{\mathrm{sym}(\mathcal{G})}\mathcal{A}^{\lambda}_{\mathcal{G}}.
\end{align}
As shown in \cite{GurauLargeN1}, the leading order graphs of this expansion in the large $N$ limit are so-called ``\textit{melonic diagrams}'', which are certain coloured graphs dual to the $3$-sphere $S^{3}$. This result generalizes the well-known fact that planar graphs form the leading order in matrix models for pure two-dimensional quantum gravity \cite{Planar}. A similar result has been obtained for higher-dimensional Ooguri-Boulatov type models \cite{GurauLargeN2,GurauLargeN3}. See also \cite{CritTM,GurauColouredTensorModelsReview,GurauColouredTensorModelsReview2} for an extended discussion in the setting of simplicial coloured tensor models and \cite{CritTM2,GurauBook} for a discussion in the setting of coloured tensor models with bubble interactions.
\section{Topology of Coloured Graphs with Non-Empty Boundaries}
\label{SectionII:CryTheo}

As seen above, the Feynman diagrams of coloured tensor models and GFTs are certain types of edge-coloured graphs. The topology of these graphs is not only studied in quantum gravity, but also in crystallization theory --a branch of geometric topology. In this section, we discuss some general concepts and important results from the topology of coloured graphs, combining notions which are used both in quantum gravity and crystallization theory. We will mainly focus on the general notion of coloured graphs representing pseudomanifolds with non-empty boundaries. For a general review of the topology of coloured graphs in the context of coloured tensor models and GFTs see for example \cite{GurauBook,GurauColouredTensorModelsReview,GurauColouredTensorModelsReview2}. For surveys on crystallization theory see \cite{GagliardiBoundaryGraph,CTReview,Review2018} and references therein. Further details on the topology of coloured graphs with non-empty boundary can be found in Appendix \ref{AppendixB:CryTheo}.

\subsection{Bubbles and their Multiplicities}
\label{SecBubbles}

The underlying graph of some closed (resp. open) $(d+1)$-coloured graph $\mathcal{G}$ is the dual $1$-skeleton (resp. \textit{internal} dual $1$-skeleton) of the corresponding simplicial complex $\Delta_{\mathcal{G}}$. However, as discussed previously, the simplicial complex assigned to $\mathcal{G}$ is unique and hence we expect that also the higher-dimensional dual cells and their nested structure are encoded in the graph $\mathcal{G}$. This leads to the notion of ``bubbles'' \cite{GurauColouredTensorModelsReview}, or equivalently, ``residues'' \cite{GagliardiBoundaryGraph} in the mathematical literature on crystallization theory:

\begin{Definition}[Bubbles]
    Let $\mathcal{G}\in\mathfrak{G}_{d}$ (see Definition \ref{def:notation_graph}) be an open $(d+1)$-coloured graph and $i_{1},\dots,i_{k}\in\mathcal{C}_{d}$ with $i_{1}<\dots<i_{k}$, $k\in\{0,\dots,d\}$. We call a connected component of the graph obtained by deleting all the edges of colours $\mathcal{C}_{d}\backslash\{i_{1},\dots,i_{k}\}$ a ``$k$-bubble of colours $i_{1},\dots,i_{k}$''. We denote such a bubble by $\mathcal{B}^{i_{1}\dots i_{k}}_{(\rho)}$, where $\rho$ labels the various bubbles of the same colours. The total number of $k$-bubbles of arbitrary colours is denoted by $\mathcal{B}^{[k]}$.
\end{Definition}

Figure \ref{BubblesFig} below shows an open $(3+1)$-coloured graph $\mathcal{G}\in\mathfrak{G}_{3}$, called the ``\textit{elementary melonic $3$-ball}'' \cite{GurauColouredTensorModelsReview2}, together with all its $3$-bubbles:

\begin{figure}[H]
    \centering
    \includegraphics[scale=1]{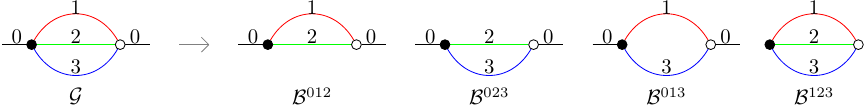}
    \caption{The elementary melonic 3-ball $\mathcal{G}$ (l.h.s.) and its four $3$-bubbles.\label{BubblesFig}}
\end{figure}

Note that the set of $0$-bubbles is precisely the vertex set $\mathcal{V}_{\mathcal{G}}$ of $\mathcal{G}$. In principle, this also includes the $1$-valent boundary vertices. However, we consider in the following the convention where only the $(d+1)$-valent internal vertices are considered $0$-bubbles so that $0$-bubbles correspond to the $d$-simplices of the simplicial complex. It is immediate to see that $1$-bubbles are edges and so correspond to the $(d-1)$-simplices of the complex. Similarly, $2$-bubbles are called the ``\textit{faces of the graph}'' and they correspond to the $(d-2)$-simplices of the complex. This correspondence can be extended to all dimensions:

\begin{Proposition}
    There is a one-to-one correspondence between the $k$-bubbles of some open $(d+1)$-coloured graph $\mathcal{G}\in\mathfrak{G}_{d}$ and the $(d-k)$-simplices of the corresponding simplicial complex $\Delta_{\mathcal{G}}$.
\end{Proposition}

\begin{proof}
    It is not too hard to see that a $k$-bubble $\mathcal{B}$ is exactly the graph, which is dual to the (disjoint) link (see Appendix \ref{AppendixA:PseudoManifolds}) of a $(d-k)$-simplex $\sigma$ of $\Delta_{\mathcal{G}}$, i.e.
    \begin{align}
        \Delta_{\mathcal{B}}=\mathrm{Lk}_{\Delta_{\mathcal{G}}}(\sigma).
    \end{align}
    More precisely, recall that the colouring of the $d+1$ faces of each $d$-simplex in the complex induces a colouring of vertices. Now, a $k$-simplex $\sigma$ has $(k+1)$ vertices, which have some colours, lets say $\{i_{1},\dots,i_{k+1}\}\subset\mathcal{C}_{d}$. The link of $\sigma$ is by definition a $(d-1-k)$-dimensional complex, which is dual to a $(d-k)$-coloured graph. This $(d-k)$-coloured graph is exactly a $(d-k)$-bubble in $\mathcal{G}$ with colours $\mathcal{C}_{d}\backslash\{i_{1},\dots,i_{k+1}\}$.
\end{proof}

In particular, this means that there is the following correspondence in the case of dimension $d=3$:

\begin{align*}
    \text{$0$-bubbles (internal vertices) of } \mathcal{G}\hspace{0.5cm} &\Leftrightarrow\hspace{0.5cm} \text{$3$-simplices (tetrahedra) of  } \Delta_{\mathcal{G}}\\
    \text{$1$-bubbles (edges) of } \mathcal{G}\hspace{0.5cm} &\Leftrightarrow\hspace{0.5cm} \text{$2$-simplices (triangles) of } \Delta_{\mathcal{G}}\\
    \text{$2$-bubbles (faces) of } \mathcal{G}\hspace{0.5cm} &\Leftrightarrow\hspace{0.5cm} \text{$1$-simplices (edges) of } \Delta_{\mathcal{G}}\\
    \text{$3$-bubbles of } \mathcal{G}\hspace{0.5cm} &\Leftrightarrow\hspace{0.5cm} \text{$0$-simplices (vertices) of } \Delta_{\mathcal{G}}
\end{align*}

\begin{Remarks}
    \begin{itemize}
        \item[]
        \item[(a)]A $k$-bubble is by itself a $k$-coloured graph that can either be open or closed. If a $k$-bubble $\mathcal{B}$ is open, then the corresponding $(d-k)$-simplex lives purely on the boundary of the simplicial complex $\Delta_{\mathcal{G}}$. Instead, if $\mathcal{B}$ is closed, then the corresponding simplex lives in the interior of $\Delta_{\mathcal{G}}$ (possibly touching the boundary). As an example, the complex dual to the graph in figure \ref{BubblesFig} has three boundary vertices and only one internal vertex (the vertex dual to $\mathcal{B}^{123}$).
        \item[(b)]The proposition above tells us that there is a family of bijective maps of the form $\varphi_{k}:\Delta_{\mathcal{G},k}\to\mathcal{B}^{[d-k]}$, where $\Delta_{\mathcal{G},k}$ denotes the set of $k$-simplices of the complex $\Delta_{\mathcal{G}}$. Note also that these maps are inclusion reversing: Consider a $k$-simplex $\sigma$ and let $\tau$ be some $l$-face of $\sigma$. Then $\varphi_{k}(\sigma)$ is a $(d-k)$-bubble within the $(d-l)$-bubble $\varphi_{l}(\tau)$. Hence, the colouring does not only include information about higher-dimensional dual cells but also about their nested structure.
    \end{itemize}
\end{Remarks}

The topology of bubbles can be used to determine whether a coloured graph describes a manifold or a pseudomanifold:

\begin{Proposition}
    \label{ManifoldsGraphs}
    Let $\mathcal{G}\in\mathfrak{G}_{d}$ be an open $(d+1)$-coloured graph. Then $\vert\Delta_{\mathcal{G}}\vert$ is a manifold if and only if all the $d$-bubbles of $\mathcal{G}$ represent either $(d-1)$-spheres or $(d-1)$-balls.
\end{Proposition}

\begin{proof}
    Every triangulation with the property that all the links of its vertices (=the $d$-bubbles of the graph) represent spheres or balls (a so-called ``combinatorial triangulation'', see Appendix \ref{AppendixA:PseudoManifolds}) is a manifold (in fact, a PL-manifold), see \cite{Hudson}. For the reverse, see \cite{GagliardiBoundaryGraph} and references therein.
\end{proof}

Previously, we have defined the boundary graph $\partial\mathcal{G}$ of some open $(d+1)$-coloured graph $\mathcal{G}\in\mathfrak{G}_{d}$ and said that the underlying graph is exactly the boundary dual $1$-skeleton of the complex $\Delta_{\mathcal{G}}$. Since $\partial\mathcal{G}$ is a closed $d$-coloured graph, we can construct the corresponding simplicial complex $\Delta_{\partial\mathcal{G}}$. Naively, we would guess that this simplicial complex is exactly the boundary of the simplicial complex dual to $\mathcal{G}$, i.e.~$\Delta_{\partial\mathcal{G}}=\partial\Delta_{\mathcal{G}}$. However, it turns out that $\partial\Delta_{\mathcal{G}}$ is in general just a quotient of the simplicial complex $\Delta_{\partial\mathcal{G}}$ obtained by identifying some of its simplices. This is actually well known in crystallization theory and goes under the name ``\textit{multiple residues}'' \cite{GagliardiExistence,Gagliardi87,GagliardiMultRes}. Let us discuss this point in more details using an explicit example. Consider the following closed $(2+1)$-coloured graph $\gamma\in\overline{\mathfrak{G}}_{2}$, called the ``\textit{pillow graph}'', as boundary graph:

\begin{figure}[H]
    \centering
    \includegraphics[scale=1.2]{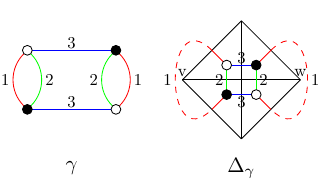}
    \caption{A closed $(2+1)$-coloured graph $\gamma$ (l.h.s.) together with its simplicial complex $\Delta_{\gamma}$ (r.h.s.).\label{MultRes}}
\end{figure}

The graph represents a $2$-sphere, as can be seen by looking at the simplicial complex $\Delta_{\gamma}$ dual to $\gamma$. Now, consider the two open $(3+1)$-coloured graphs $\mathcal{G}_{1},\mathcal{G}_{2}\in\mathfrak{G}_{3}$ of figure \ref{G1G2MultRes}:

\begin{figure}[H]
    \centering
    \includegraphics[scale=1.2]{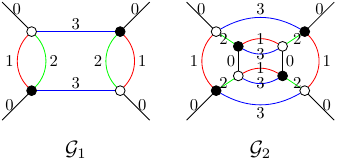}
    \caption{Two open $(3+1)$-coloured graphs $\mathcal{G}_{1,2}\in\mathfrak{G}_{3}$ with boundary graph given by the graph $\gamma$.\label{G1G2MultRes}}
\end{figure}

Both of these graphs satisfy $\partial\mathcal{G}_{1}=\partial\mathcal{G}_{2}=\gamma$. One can easily see that the boundary of the simplicial complex $\Delta_{\mathcal{G}_{1}}$, which describes a $3$-ball, is given by the complex $\Delta_{\gamma}$, i.e.
\begin{align}
    \partial\Delta_{\mathcal{G}_{1}}=\Delta_{\partial\mathcal{G}}=\Delta_{\gamma}.
\end{align}
However this is not the case for the simplicial complex dual to $\mathcal{G}_{2}$. Indeed, note that the graph $\mathcal{G}_{2}$ has in total four $3$-bubbles, from which three are open graphs. One of them, the $3$-bubble of colour $012$, has two disconnected boundary components, see figure \ref{fig:BubbleMultRes}.

\begin{figure}[H]
    \centering
    \includegraphics[scale=1.2]{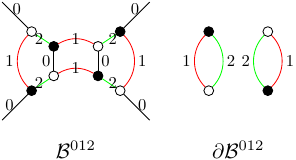}
    \caption{The unique $3$-bubble $\mathcal{B}^{012}$ of colour $012$ of the graph $\mathcal{G}_{2}$ together with its boundary graph $\partial\mathcal{B}^{012}$.}
    \label{fig:BubbleMultRes}
\end{figure}

As explained above, the $3$-bubbles of some open $(3+1)$-coloured graph $\mathcal{G}$ correspond to the vertices of the simplicial complex $\Delta_{\mathcal{G}}$, whereas the $2$-bubbles of its closed $(2+1)$-coloured boundary graph $\partial\mathcal{G}$ correspond to the vertices of the complex $\Delta_{\partial\mathcal{G}}$. Hence, we see that in the above example, the two vertices dual to the two $2$-bubbles of colour $12$ of $\gamma$ are identified in the simplicial complex $\Delta_{\mathcal{G}_{2}}$, since they both correspond to the same $3$-bubble in $\mathcal{G}_{2}$. In other words, the boundary of the simplicial complex $\Delta_{\mathcal{G}_{2}}$ is the complex obtained by identifying the two vertices $v$ and $w$ of the complex $\Delta_{\gamma}$ drawn on the right-hand side in figure \ref{MultRes}, i.e.~we can write
\begin{align}
    \partial\Delta_{\mathcal{G}_{2}}=\Delta_{\gamma}/_{v\sim w}.
\end{align}
The geometric realization of this complex is the ``\textit{pinched torus}'', i.e.~the pseudomanifold obtained by identifying two distinct points on a $2$-sphere. This discussion leads to the following definition:

\begin{Definition}[Multiplicity of Bubbles \cite{GagliardiMultRes}]
    \label{DefSimpleBubbles}
    Let $\mathcal{G}$ be an open $(d+1)$-coloured graph. We call the number of boundary components of some bubble $\mathcal{B}$ the ``multiplicity of $\mathcal{B}$'' and denote it by $\mathrm{mult}(\mathcal{B})$. If $\mathrm{mult}(\mathcal{B})\in\{0,1\}$, then we call the bubble ``simple''.
\end{Definition}

If $\mathcal{G}\in\mathfrak{G}_{d}$ only has simple bubbles, then we clearly have that $\Delta_{\partial\mathcal{G}}=\partial\Delta_{\mathcal{G}}$. This is in particular the case if $\mathcal{G}$ represents a manifold, since all its $d$-bubbles are spheres and balls. Furthermore, this is also clearly true for pseudomanifolds without boundary singularities, i.e.~pseudomanifolds for which all the open $3$-bubbles represents $(d-1)$-balls. Using the discussion of the example above, one can easily see that there is the following general relationship between the complex of the boundary graph and the boundary of the simplicial complex of the corresponding open graph:

\begin{Proposition}[Boundary Complex of a General Open Graph]
    Let $\mathcal{G}$ be an open $(d+1)$-coloured graph with boundary graph $\mathcal{G}$. Then
    \begin{align*}
        \partial\Delta_{\mathcal{G}}=\Delta_{\partial\mathcal{G}}/\sim,
    \end{align*}
    where $\sim$ identifies for each non-simple $k$-bubble $\mathcal{B}$ of $\mathcal{G}$ with $k\in\{3,\dots d\}$ the corresponding $(d-k)$-simplices belonging to the various boundary components of $\mathcal{B}$.
\end{Proposition}

The appearance of this additional pinching effect on the boundary could have been expected since the boundary graph only takes the $1$-skeleton of the complex $\partial\Delta_{\mathcal{G}}$ into account. While it does encode a full simplicial complex, it does not contain any information about these possible identifications of $k$-simplices with $k\leq d-3$, which are coming from the bulk graph. In other words, the boundary graph only describes the ``\textit{desingularized}'' boundary of the complex $\Delta_{\mathcal{G}}$.

\subsection{Combinatorial and Topological Equivalence}\label{SecDipoles}
Every manifold admits a coloured graph representing it (see Appendix \ref{SecCryTheo}), however there are in general infinitely many inequivalent graphs representing the same topology. In order to properly describe a manifold of a given topology, we need transformations changing the graph but leaving the topology of the associated manifold invariant. For PL-manifolds Pachner's theorem \cite{Pachner} states that two PL-manifolds are PL-homeomorphic if and only if they are related by a finite sequence of so-called Pachner moves. In three dimensions, there are only two different types of Pachner moves, the $(1-4)$- and the $(2-3)$-move. For our purpose, these moves do not work since they are in general not respecting the underlying structure of the coloured graph. For example, applying a $(1-4)$-move to some tetrahedron results into a complex which is not bipartite anymore. It turns out that a suitable set of moves is given by so-called ``dipole moves'', which were introduced in \cite{GagliardiFerri}:

\begin{Definition}[Dipoles and Dipole Contraction]
    Let $\mathcal{G}\in\mathfrak{G}_{d}$ be an open $(d+1)$-coloured graph, such that $\vert\mathcal{V}_{\mathcal{G},\mathrm{int}}\vert>2$. We call a subgraph $d_{k}$ consisting of two internal vertices $v,w\in\mathcal{V}_{\mathcal{G},\mathrm{int}}$, which are connected by $k$ edges of colours $i_{1},\dots,i_{k}\in\mathcal{C}_{d}$, ``$k$-dipole of colours $i_{1},\dots,i_{k}$'', if the two $(d+1-k)$-bubbles of colour $\mathcal{C}_{d}\backslash\{i_{1},\dots,i_{k}\}$ containing $v$ and $w$, respectively, are distinct.
\end{Definition}

If some coloured graphs admits a dipole, then we define another graph by ``contracting the dipole'' \cite{GagliardiFerri,Gagliardi87}:

\begin{Definition}[Dipole Contraction]
    Let $\mathcal{G}\in\mathfrak{G}_{d}$ be an open $(d+1)$-coloured graph and $d_{k}$ a $k$-dipole within $\mathcal{G}$ with vertices $v,w$. Then we define the graph $\mathcal{G}/d_{k}\in\mathfrak{G}_{d}$ by deleting the two vertices $v$ and $w$ of $\mathcal{G}$ and by connecting the ``hanging pairs'' of edges respecting their colouring. We say that ``$\mathcal{G}/d_{k}$ is obtained by contracting the $k$-dipole $d_{k}$ in $\mathcal{G}$''. The inverse process is called ``creating a dipole''. See figure \ref{FigDipoleMove} for examples in dimension $d=3$.
\end{Definition}

\begin{Remarks}
    \label{DipoleConSum}
    \begin{itemize}
        \item[]
        \item[(a)]If both vertices $v$ and $w$ admit an adjacent external leg, then the procedure would produce a disconnected part containing a single edge of colour $0$ connecting two boundary edges. In this case, we do not include this additional disconnected piece in the definition of $\mathcal{G}/d_{k}$, as a convention (e.g.~see figure \ref{FigDipoleMove}(b) below).
        \item[(b)]Note that performing a $k$-dipole move in some open $(d+1)$-coloured graph $\mathcal{G}$ is geometrically one and the same as performing the graph connected sum (see Appendix \ref{Subsec:ConSum}) of two $(d+1-k)$-bubbles within the graph $\mathcal{G}$.
    \end{itemize}
\end{Remarks}

Figure \ref{FigDipoleMove} shows three examples of $1$-dipole contraction in open $(3+1)$-coloured graphs.

\begin{figure}[H]
    \centering
    \includegraphics[scale=1]{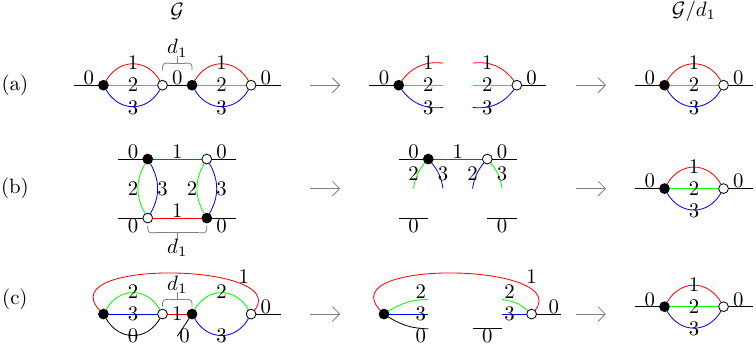}
    \caption{Three examples of $1$-dipole contractions in open $(3+1)$-coloured graphs.\label{FigDipoleMove}}
\end{figure}

Note that the boundary graph only changes in example (b). The reason for this is that the two vertices involved in the dipole admit adjacent external legs and therefore, after contracting the dipole, the number of boundary triangles is reduced by two. On the other hand the boundary graph is left untouched whenever one of the separated $(d+1-k)$-bubbles is closed, as in example (a) and (c) above:

\begin{Proposition}[Boundary Complex and Dipole Moves]
    \label{PropDipoleBound}
    Let $\mathcal{G}\in\mathfrak{G}_{d}$ be an open $(d+1)$-coloured graph and $d_{k}$ a $k$-dipole within $\mathcal{G}$. If at least one of the two $(d+1-k)$-bubbles separated by the dipole is closed, then $\partial\mathcal{G}=\partial(\mathcal{G}/d_{k})$ and also $\partial\Delta_{\mathcal{G}}=\partial\Delta_{\mathcal{G}/d_{k}}$. We call such a dipole ``internal''.
\end{Proposition}

\begin{proof}
    Let us assume without loss of generality that the colours involved in the $k$-dipole are $1,\dots,k$, because if $0$ is involved in the dipole, the two $(d+1-k)$-bubbles separated by $d_{k}$ are both closed and the claim is trivially true in this case. The general situation is sketched in the figure below

    \begin{figure}[H]
        \centering
        \includegraphics[scale=1.1]{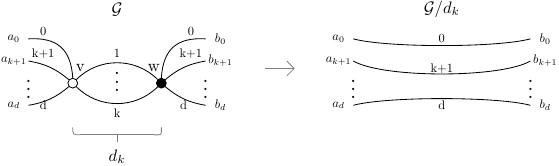}
        \caption{Contraction of a generic $k$-dipole of colours $1,\dots,k$.}
    \end{figure}

    By assumption, one of the $(d+1-k)$-bubbles separated by $d_{k}$ is closed and we choose without loss of generality the bubble $\mathcal{B}_{v}^{d+1-k}$ containing $v$. Note that the vertices $a_{i}$ do not necessarily have to be distinct and similarly for the $b_{i}$'s. Furthermore, $b_{0}$ could in principle be a $1$-valent boundary vertex.  Clearly all the bicoloured paths starting and ending at an external leg of colour $0i$ with $i\in\{1,\dots,k\}$, which are going through the dipole, necessarily contain the vertices $a_{0}$ and $b_{0}$ and still exist after contracting $d_{k}$. Next, consider a bicoloured path containing the vertex $w$ of colour $0j$ with $j\in\{k+1,\dots,d\}$. Such a path connects the vertex $b_{0}$ with $w$ and the vertex $w$ with $b_{j}$. Now, if we contract the dipole $d_{k}$, then this bicoloured path still exists precisely because we have assumed that $\mathcal{B}_{v}^{d+1-k}$ is closed: The path in $\mathcal{G}/d_{k}$ connects the vertex $b_{0}$ with $a_{0}$, the vertex $a_{0}$ with $a_{j}$ by a bicoloured path of colours $0j$ and the vertex $a_{j}$ with $b_{j}$. This shows that all the non-cyclic faces of $\mathcal{G}$ are still contained in $\mathcal{G}/d_{k}$. Furthermore, it is also clear that we do not produce new non-cyclic faces, since the number of external legs is left untouched in this case. Therefore, we conclude that $\partial\mathcal{G}=\partial(\mathcal{G}/d_{k})$. Since there is a natural inclusion of bubbles of $\mathcal{G}/d_{k}$ into bubbles of $\mathcal{G}$, the multiplicities of bubbles do not change, which implies that also the boundary complexes are the same.
\end{proof}

Note that a dipole move does not always preserve the topology. This can easily be seen by the fact that performing a dipole move is the same as performing the connected sum of two submanifolds, as mentioned in Remark \ref{DipoleConSum}(b). For example, whenever both of these two submanifolds are neither spheres nor balls, a dipole move will change the topology of the manifold. Let us introduce the following terminology:

\begin{Definition}[Proper Dipole Moves]
    Let $\mathcal{G}\in\mathfrak{G}_{d}$ be an open $(d+1)$-coloured graph and $d_{k}$ a $k$-dipole within $\mathcal{G}$. We say that $d_{k}$ is ``proper'', if $\vert\Delta_{\mathcal{G}}\vert$ and $\vert\Delta_{\mathcal{G}/d_{k}}\vert$ represent the same manifold (up to PL-homeomorphism).
\end{Definition}

As an example, all the dipole moves drawn in figure \ref{FigDipoleMove} are proper, because all the graphs represent $3$-balls. More generally, as proven in \cite{GagliardiFerri} (for closed graphs) and in \cite{Gagliardi87} (for open graphs), one can define two classes of dipole moves, which preserve the topology:

\begin{Theorem}[Gagliardi-Ferri]
    \label{DipoleProper}
    Let $\mathcal{G}\in\mathfrak{G}_{d}$ be an open $(d+1)$-coloured graph and $d_{k}$ a $k$-dipole involving vertices $v,w\in\mathcal{V}_{\mathcal{G},\mathrm{int}}$.
    \begin{itemize}
        \item[(1)]If at least one of the $(d+1-k)$-bubbles separated by the dipole represents a $(d-k)$-sphere, then $d_{k}$ is proper. We call such dipoles ``internal proper dipoles''.
        \item[(2)]If both $v$ and $w$ admit an adjacent external legs and at least one of the $(d+1-k)$-bubbles separated by the dipole represents a $(d-k)$-ball, then $d_{k}$ is proper. We call such dipoles ``non-internal proper dipoles''.
    \end{itemize}
\end{Theorem}

\begin{proof}
    For a complete geometrical proof see Proposition 5.3.~in \cite{Gagliardi87}. Using the graph-connected sum discussed in Appendix \ref{Subsec:ConSum}, one can actually give an alternative proof of the statement, as already observed in \cite{CasaliPL}: In case (1), we basically just perform the (internal) connected sum of a spherical $(d+1-k)$-bubble with some other topology (possibly with boundary), which is trivial and hence leaves the topology invariant (see Corollary \ref{CorConSum}). In case (2), we perform the boundary connected sum of some $(d+1-k)$-bubble representing a $(d-k)$-ball with some other topology and hence we do not change the topology either (see Theorem \ref{ThmConSum} (1)).
\end{proof}

The three examples of figure \ref{FigDipoleMove} do have these properties. More precisely, the dipoles in the graphs (a) and (c) are internal proper $1$-dipoles and the dipole in (b) is a non-internal proper $1$-dipole.

\begin{Remarks}
    \begin{itemize}
        \item[]
        \item[(a)]Note that by Proposition \ref{PropDipoleBound}, an internal proper dipole leaves the boundary complex invariant whereas a non-internal proper dipole changes the boundary complex explicitly, since it does reduce the number of boundary $(d-1)$-simplices by two.
        \item[(b)]Every non-internal proper $k$-dipole move induces an internal proper $k$-dipole move on its boundary graph. The reverse is in general not true. However, it turns out that every proper dipole on the boundary graph corresponds to a ``wound move'', another set of moves discussed in \cite{Gagliardi87}, in the open graph.
        \item[(c)]Note that the Theorem of Gagliardi-Ferri is not an ``if-and-only-if'' statement and there are also proper dipoles, which do not fall into the two classes defined above. See Appendix A.4 in \cite{GabrielThesis} for a short discussion and examples. However, as it turns out, the two classes of proper dipole moves discussed in the Theorem of Gagliardi-Ferri are sufficient to characterize topological invariance, as stated in Theorem \ref{Casali} below.
    \end{itemize}
\end{Remarks}

Let us mention the following immediate consequences of the above theorem:

\begin{Corollary}
    \label{Cor:Dipoles}
    Let $\mathcal{G}\in\mathfrak{G}_{d}$ be some $(d+1)$-coloured graph.
    \begin{itemize}
        \item[(1)]Every $d$-dipole is proper. If $\mathcal{G}$ is closed, then also every $(d-1)$-dipole is proper.
        \item[(2)]If $\mathcal{G}$ is closed and represents a manifold, then every dipole is proper.
        \item[(3)]If $\mathcal{G}$ represents a manifold --possibly with boundary-- then every $k$-dipole involving the colour $0$ is proper.
        \item[(4)]If $\mathcal{G}$ is open and represents a manifold, then every $k$-dipole in which both vertices admit adjacent external legs is a non-internal proper one.
    \end{itemize}
\end{Corollary}

\begin{proof}
    This follows from the previous theorem as well as Proposition \ref{ManifoldsGraphs}, i.e.~the fact that for manifolds all $d$-bubbles represent spheres or balls.
\end{proof}

Up to now, we have introduced a set of moves for general coloured graphs leaving the topology invariant. However, it is not yet clear if this set of moves are enough to relate any two coloured graphs describing the same topology to each other. It turns out to be the case:

\begin{Theorem}[Equivalence Criterion of Casali \cite{Casali,CasaliPL}]
    \label{Casali}
    Let $\mathcal{G}_{1},\mathcal{G}_{2}\in\mathfrak{G}_{d}$ be two open $(d+1)$-coloured graphs representing manifolds. Then the manifolds $\mathcal{M}:=\vert\Delta_{\mathcal{G}_{1}}\vert$ and $\mathcal{M}_{2}:=\vert\Delta_{\mathcal{G}_{2}}\vert$ are PL-homeomorphic if and only if $\mathcal{G}_{1}$ and $\mathcal{G}_{2}$ are related by a finite sequence of proper dipole moves of the two types defined in Theorem \ref{DipoleProper}. Moreover, if $\partial\mathcal{G}_{1}\cong\partial\mathcal{G}_{2}$, then $\mathcal{M}_{1}$ and $\mathcal{M}_{2}$ are PL-homeomorphic if and only if $\mathcal{G}_{1}$ and $\mathcal{G}_{2}$ are related by a finite sequence of internal proper dipole moves.
\end{Theorem}

Therefore, we are free to use proper dipole moves in order to study the different graphs associated to a manifold of a given topology.
\section{Transition Amplitudes}\label{SectionIII:TransAmpl}

Having discussed the graph theoretical and topological properties of the Feynman graphs emerging from the coloured Boulatov model, we now move on to the transition amplitudes. The purpose of this construction, from a canonical quantum gravity point of view, is in fact to define a physical scalar product between two boundary states.\footnote{Therefore, the label \lq transition amplitudes\rq~should be understood loosely. For the difference between spin foam amplitudes defining the canonical scalar product and those encoding quantum gravity \lq transitions\rq, see \cite{Livine:2002rh, Oriti:2004mu, Oriti:2006wq, Bianchi:2021ric}.} In the context of the Boulatov model, the boundary states are spin networks states \cite{PenroseSpinNetwork,BaezGaugeTheory} living on some fixed boundary graph, which are dual to some fixed boundary topology. The transition amplitudes should then provide us with a sum over all topologies with boundary given by our fixed boundary graphs, each weighted by their corresponding spin foam amplitude. In this section, we will start by defining suitable GFT boundary observables, which can then be used to define transition amplitudes. Afterwards, we will apply the techniques from crystallization theory discussed in the previous section, in order to rewrite the amplitudes as topological expansions similar in spirit to the topological expansion of the free energy in the large $N$ limit proposed by R.~Gurau \cite{GurauLargeN1,GurauLargeN2,GurauLargeN3}. The results of this section are based on the Master’s thesis of one of the authors (GS) \cite{GabrielThesis}.

\subsection{Boundary Observables and Transition Amplitudes}
GFT boundary states are described by spin networks \cite{PenroseSpinNetwork,BaezGaugeTheory}. To start with, let us recall that a ``\textit{$\mathrm{SU}(2)$ spin network}'' is defined to be a triple $(\gamma,\rho,i)$, where $\gamma=(\mathcal{V}_{\gamma},\mathcal{E}_{\gamma})$ is a directed and finite graph, $\rho=(\rho_{e},\mathcal{H}_{e})_{e\in\mathcal{E}_{\gamma}}$ is an assignment of irreducible and unitary representations of $\mathrm{SU}(2)$ to edges of the graph $\gamma$ and $i=(i_{v})_{v\in\mathcal{V}_{\gamma}}$ is an assignment of intertwiners of the type
\begin{align}
    i_{v}:\bigotimes_{e\in\mathcal{T}(v)}\mathcal{H}_{e}\to\bigotimes_{e \in\mathcal{S}(v)}\mathcal{H}_{e},
\end{align}
where $\mathcal{T}(v)$ denotes the collection of edges incoming to $v$ and $\mathcal{S}(v)$ the collection of edges outgoing from $v$. To every spin network $\Psi:=(\gamma,\rho,i)$, one can associate a corresponding ``\textit{spin network function}'', which is a map $\psi\in L^{2}(\mathrm{SU}^{\vert\mathcal{E}_{\gamma}\vert},\mathrm{d}g;\mathbb{C})$ satisfying
\begin{align}
    \psi(\{g_{e}\}_{e\in\mathcal{E}_{\gamma}})=\psi(\{k_{t(e)}^{-1}g_{e}k_{s(e)}\}_{e\in\mathcal{E}_{\gamma}})
\end{align}
for all $\{k_{e}\}_{e\in\mathcal{V}_{\gamma}}\in\mathrm{SU}(2)^{\vert\mathcal{V}_{\gamma}\vert}$, defined by
\begin{align}
    \psi(\{g_{e}\}_{e\in\mathcal{E}}):=\bigg (\bigotimes_{v\in\mathcal{V}_{\gamma}}i_{v}\bigg )\bullet_{\gamma}\bigg(\bigotimes_{e\in\mathcal{E}_{\gamma}}\rho_{e}(g_{e})\bigg ),
\end{align}
where $\bullet_{\gamma}$ means contracting at each vertex $v\in\mathcal{V}_{\gamma}$ the upper indices of the matrices corresponding to the incoming edges in $v$, the lower indexes of the matrices assigned to the outgoing edges in $v$ and the corresponding upper and lower indices of the intertwiners $i_{v}$. The Hilbert space $L^{2}(\mathrm{SU}(2)^{\vert\mathcal{E}_{\gamma}\vert}/\mathrm{SU}(2)^{\vert\mathcal{V}_{\gamma}\vert},\mathrm{d}g;\mathbb{C})$ is spanned by spin network states \cite{BaezGaugeTheory}. Furthermore, from the physical point of view, spin network states are kinematic states representing quantum $3$-geometries \cite{Barbieri,BaezSpinFoams}.
\bigskip

In order to describe transition amplitudes between spin network states defined on the boundary, we have to introduce suitable boundary observables, which are endowed with the corresponding quantum geometric data. Since we are working in the language of field theory, these observables should be functionals of the fundamental fields and compatible with the $\mathrm{SU}(2)$ gauge symmetry of the model. Following the general idea of \cite{FreidelGFTOverview}, we define GFT boundary observables in the following way:

\begin{Definition}[Boundary Observables of the Coloured Boulatov Model]
    Consider a closed $(2+1)$-coloured graph $\gamma\in\overline{\mathfrak{G}}_{2}$, which we fix to be our boundary graph and which we equip with source and target maps $s,t:\mathcal{E}_{\gamma}\to\mathcal{V}_{\gamma}$. Furthermore, let us choose a spin network $\Psi=(\gamma,\rho,i)$ on $\gamma$ with corresponding spin network function $\psi\in L^{2}(\mathrm{SU}(2)^{\vert\mathcal{E}_{\gamma}\vert}/\mathrm{SU}(2)^{\vert\mathcal{V}_{\gamma}\vert})$. Then, we define the Boulatov boundary observable to be the functional
    \begin{align*}
        \mathcal{O}_{\Psi}[\varphi_{0},\overline{\varphi}_{0}]:=\int_{\mathrm{SU}(2)^{3\vert\mathcal{V}_{\gamma}\vert}}\bigg (\prod_{v\in\mathcal{V}_{\gamma}}\prod_{i=1}^{3}&\mathrm{d}g_{vi}\bigg )\psi(\{g_{s(e)\widetilde{\gamma}(e)}g_{t(e)\widetilde{\gamma}(e)}^{-1}\}_{e\in\mathcal{E}_{\gamma}})\times\\&\times\bigg (\prod_{v\in V_{\gamma}} \varphi_{0}(g_{v3},g_{v2},g_{v1})\bigg )\bigg (\prod_{v\in \overline{V}_{\gamma}}\overline{\varphi}_{0}(g_{v3},g_{v2},g_{v1})\bigg ),
    \end{align*}
    where $g_{vi}$ are the three group elements of colours $i=1,2,3$ assigned to the three half-edge adjacent to the vertex $v\in\mathcal{V}_{\gamma}$ and where $\widetilde{\gamma}:\mathcal{E}_{\gamma}\to\mathcal{C}_{2}^{\ast}=\{1,2,3\}$ is the proper edge-colouring of the graph $\gamma$.
\end{Definition}

It is important to stress that we restrict only to a certain class of boundary states, namely to boundary states living on \textit{closed $(2+1)$-coloured graphs}. Hence, all the open graph appearing in the expansion of the transition amplitudes will be such that all external legs have the same colour $0$. This is done for purely technical reasons. Note that the boundary observables can straightforwardly be generalized to arbitrary admissible bicoloured boundary graphs. In this case, the corresponding observables are then functionals of the fields of \textit{all} colours.

\begin{Remark}
    Note that we do not only fix a boundary graph, but already a boundary graph with a fixed orientation and colouring and hence with a fixed topology. This is an important difference to the general definition in \cite{FreidelGFTOverview}, since in the uncoloured version, we only fix a graph and the (dual) $1$-skeleton alone is not enough to determine a topology.
\end{Remark}

With the observables defined above, it is straightforward to define the corresponding transition amplitudes of the coloured Boulatov model:

\begin{Definition}[Transition Amplitudes for the Coloured Boulatov Model]
    \label{DefTransAmpl}
    Let $\gamma\in\overline{\mathfrak{G}}_{2}$ be a closed $(2+1)$-coloured graph and $\Psi=(\gamma,\rho,i)$ be a spin network living on $\gamma$. Then the transition amplitude is defined by \begin{align*}
        \langle\mathcal{Z}_{\mathrm{cBM}}\vert \Psi\rangle:=\int\,\bigg(\prod_{l=0}^{3}\,\mathcal{D}\varphi_{l}\mathcal{D}\overline{\varphi}_{l}\bigg)\,e^{-\mathcal{S}_{\lambda}[\varphi_{l},\overline{\varphi}_{l}]}\mathcal{O}_{\Psi}[\varphi_{0},\overline{\varphi}_{0}].
    \end{align*}
\end{Definition}

For the following discussion, let us briefly recall and set up the following terminologies and notations which we use for open $(3+1)$-coloured graph $\mathcal{G}\in\mathfrak{G}_{3}$ with boundary graph $\gamma:=\partial\mathcal{G}$:

\begin{itemize}
    \item[(1)]Recall that the vertex set can be decomposed as $\mathcal{V}_{\mathcal{G}}=\mathcal{V}_{\mathcal{G},\mathrm{int}}\cup \mathcal{V}_{\mathcal{G},\partial}$, where vertices in $\mathcal{V}_{\mathcal{G},\mathrm{int}}$ are $4$-valent internal vertices and vertices in $\mathcal{V}_{\mathcal{G},\partial}$ are $1$-valent boundary vertices and are in one-to-one correspondence with the vertices of the boundary graph $\mathcal{V}_{\gamma}$.
    \item[(2)]Similarly, we decompose the edge set as $\mathcal{E}_{\mathcal{G}}=\mathcal{E}_{\mathcal{G},\mathrm{int}}\cup \mathcal{E}_{\mathcal{G},\partial}$, where edges in $\mathcal{E}_{\mathcal{G},\mathrm{int}}$ connect two internal vertices and edges in $\mathcal{E}_{\mathcal{G},\partial}$ are external legs, i.e.~edges connecting a vertex in $\mathcal{V}_{\mathcal{G},\mathrm{int}}$ with a vertex in $\mathcal{V}_{\mathcal{G},\partial}$. Note that the set $\mathcal{E}_{\mathcal{G},\partial}$ is also in one-to-one correspondence with $\mathcal{V}_{\gamma}$.
    \item[(3)]The set of faces (=2-bubbles) of $\mathcal{G}$ is denoted by $\mathcal{F}_{\mathcal{G}}$. This set can also be decomposed as $\mathcal{F}_{\mathcal{G}}=\mathcal{F}_{\mathcal{G},\mathrm{int}}\cup\mathcal{F}_{\mathcal{G},\partial}$, where $\mathcal{F}_{\mathcal{G},\mathrm{int}}$ is the set of ``\textit{internal faces}'', i.e.~the set of closed $2$-bubbles of $\mathcal{G}$ (they correspond to the internal edges of the simplicial complex $\Delta_{\mathcal{G}}$) and where $\mathcal{F}_{\mathcal{G},\partial}$ is the set of open $2$-bubbles of $\mathcal{G}$, i.e.~faces starting and ending at an external leg (they correspond to the edges on the boundary of $\Delta_{\mathcal{G}}$). There is a one-to-one correspondence between the sets of edges of the boundary graph $\mathcal{E}_{\gamma}$ and the set $\mathcal{F}_{\mathcal{G},\partial}$. We denote this bijection by
    \begin{equation}
        \begin{aligned}
            e:\partial\mathcal{F}_{\mathcal{G}}&\to\mathcal{E}_{\gamma}\\f&\mapsto e(f).
        \end{aligned}
    \end{equation}
    From a geometrical point of view, an open $2$-bubble $f$ is the interior part of a face of the dual complex touching the boundary and the edge $e(f)$ is the corresponding edge on the boundary dual complex, ``closing'' the face.
\end{itemize}

Expanding the interaction term of the action in the coupling, we can write the formal path integral of Definition \ref{DefTransAmpl} as a sum over Gaussian integrals, which will lead to a sum over all pair-wise contractions of fields in the product of interaction Lagrangians and the given fields within the boundary observable. Renaming $g_{e}:=g_{s_{e}i}g_{t_{e}i}^{-1}$ for each edge $e\in\mathcal{E}_{\gamma}$ of colour $i\in\{1,2,3\}$, we are left with an integration over all boundary edges, where the integrand is given by the spin network $\psi$ weighted by the corresponding spin foam amplitude for each Feynman diagram. More precisely, we can write
\begin{align}
    \label{transampl}
    \langle\mathcal{Z}_{\mathrm{cBM}}\vert \Psi\rangle=\sum_{\mathcal{G}\in\mathfrak{G}_{3}\text{ with }\partial\mathcal{G}=\gamma}\frac{1}{\mathrm{sym}(\mathcal{G})}\mathcal{A}_{\mathcal{G}}^{\lambda}[\Psi],
\end{align}
where the sum is over all open $(3+1)$-coloured graphs in $\mathfrak{G}_{3}$ with $\partial\mathcal{G}=\gamma$ and where the amplitude $\mathcal{A}_{\mathcal{G}}^{\lambda}[\Psi]$ for a given open $(3+1)$-coloured graph $\mathcal{G}$ is given by the Ponzano-Regge transition function together with a prefactor depending on $N$ and $\lambda$. More precisely, the amplitude is the $L^{2}(\mathrm{SU}(2)^{\vert\mathcal{E}_{\gamma}\vert},\mathrm{d}g)$-inner product
\begin{align}
    \mathcal{A}_{\mathcal{G}}^{\lambda}[\Psi]:=\langle\mathcal{A}_{\mathcal{G}}^{\lambda}\vert\psi\rangle_{L^{2}}=\int_{\mathrm{SU}(2)^{\vert\mathcal{E}_{\gamma}\vert}}\,\bigg (\prod_{e\in\mathcal{E}_{\gamma}}\mathrm{d}g_{e}\bigg )\,\mathcal{A}_{\mathcal{G}}^{\lambda}[\{g_{e}\}_{e\in\mathcal{E}_{\gamma}}]\cdot\psi(\{g_{e}\}_{e\in\mathcal{E}_{\gamma}}),
\end{align}
where $\psi\in L^{2}(\mathrm{SU}(2)^{\vert\mathcal{E}_{\gamma}\vert}/\mathrm{SU}(2)^{\vert\mathcal{V}_{\gamma}\vert})$ is the corresponding spin network function of $\Psi$. The functionals $\mathcal{A}_{\mathcal{G}}^{\lambda}[\{g_{e}\}_{e\in\mathcal{E}_{\gamma}}]$ are defined by
\begin{align}
    \mathcal{A}_{\mathcal{G}}^{\lambda}[\{g_{e}\}_{e\in\mathcal{E}_{\gamma}}]:=\bigg (\frac{\lambda\overline{\lambda}}{\delta^{N}(\mathds{1})}\bigg )^{\frac{\vert\mathcal{V}_{\mathcal{G},\mathrm{int}}\vert}{2}}\mathcal{Z}_{\mathrm{PR}}^{\mathcal{G}}[\{g_{e}\}_{e\in\mathcal{E}_{\gamma}}],
\end{align}
where the ``\textit{Ponzano-Regge functional}'' $\mathcal{Z}_{\mathrm{PR}}^{\mathcal{G}}[\{g_{e}\}_{e\in\mathcal{E}_{\gamma}}]$ is the well-known spin foam amplitude given by
\begin{equation}
    \begin{aligned}
        \mathcal{Z}_{\mathrm{PR}}^{\mathcal{G}}[\{g_{e}\}_{e\in\mathcal{E}_{\gamma}}]=\int_{\mathrm{SU}(2)^{\vert\mathcal{E}_{\mathcal{G}}\vert}}\bigg (\prod_{e\in\mathcal{E}_{\mathcal{G}}}\mathrm{d}h_{e}\bigg )\,\bigg \{\prod_{f\in\mathcal{F}_{\mathcal{G},\mathrm{int}}}&\delta^{N}\bigg (\overrightarrow{\prod_{e\in f}}h_{e}^{\varepsilon(e,f)}\bigg )\bigg \}\times\\&\times\bigg \{\prod_{f\in\mathcal{F}_{\mathcal{G},\partial}}\delta^{N}\bigg (g_{e(f)}^{\varepsilon(e(f),f)}\cdot\overrightarrow{\prod_{e\in f}}h_{e}^{\varepsilon(e,f)}\bigg )\bigg \}.
    \end{aligned}
\end{equation}
The starting points of the products within the delta functions corresponding to the non-cyclic faces (second line) are fixed to be one of the corresponding boundary vertices.
\bigskip

The interpretation of the quantity $\langle\mathcal{Z}_{\mathrm{cBM}}\vert \Psi\rangle$ is the following. If $\gamma$ has two boundary components, then it computes the probability amplitude (overlap) between these two states, where we sum over all topologies matching the given boundary topologies, each weighted by the Ponzano-Regge partition function. If $\gamma$ has a single boundary component, then $\langle\mathcal{Z}_{\mathrm{cBM}}\vert \Psi\rangle$ can be interpreted as the probability for the transition of the state from the vacuum.

\begin{Remark}
    More precisely, we should take the logarithm in the definition of $\langle\mathcal{Z}_{\mathrm{cBM}}\vert \Psi\rangle$, since then we only produce \textit{connected} Feynman graphs. However, we will mainly work with connected boundary graphs in the following and hence, all the disconnected parts produced in the amplitude are closed graphs and these additional vacuum diagrams are anyway cancelled by the normalization one usually puts in front of the path integral.
\end{Remark}

\subsection{Bubble Rooting and Core Graphs}
\label{Subse:Rooting}
The guiding idea of the following section is to collect different coloured graphs with the same amplitude, the same boundary and the same topology together. This essentially generalizes the bubble rooting procedure for closed graphs introduced in \cite{GurauLargeN1,GurauLargeN2,GurauLargeN3} to open graphs. We will restrict our attention to the three-dimensional case, although everything can easily be generalized to higher dimensions.
\bigskip

A suitable way to relate graphs in a topology- and boundary-preserving way is given by performing internal proper dipole moves, as discussed in Section \ref{SecDipoles}. Hence, we should have a look how amplitudes change when performing such a transformation. Before stating the result, let us prove the following preliminary lemma:

\begin{Lemma}
    \label{LemmaSphere}
    Consider a closed $(2+1)$-coloured graph $\gamma$ representing the $2$-sphere equipped with group elements on its edges. Furthermore, let $\mathcal{P}$ be a closed $3$-coloured path within the graph $\gamma$. Then
    \begin{align*}
        \delta\bigg(\overrightarrow{\prod_{e\in \mathcal{P}}}h_{e}^{\varepsilon(e,f)}\bigg )\prod_{f\in\mathcal{F}_{\gamma}}\delta^{N}\bigg (\overrightarrow{\prod_{e\in f}}h_{e}^{\varepsilon(e,f)}\bigg )=\delta^{N}(\mathds{1})\prod_{f\in\mathcal{F}_{\gamma}}\delta^{N}\bigg (\overrightarrow{\prod_{e\in f}}h_{e}^{\varepsilon(e,f)}\bigg ).
    \end{align*}
    The same is true if $\gamma$ is an open $(2+1)$-coloured graph representing the $2$-ball (=disk) and $\mathcal{P}$ is a closed $3$-coloured path in the interior, i.e.~not including external legs and edges of the boundary graph, if we replace $\mathcal{F}_{\gamma}$ by $\mathcal{F}_{\gamma,\mathrm{int}}$.
\end{Lemma}

\begin{proof}
    Since $\gamma$ represents the $2$-sphere, it is in particular a ``\textit{planar}'' graph, which means that it can be drawn in such a way that all the faces of the underlying graph (=regions bounded by a closed set of vertices and edges) are actually also faces in the coloured sense, i.e.~they are bicoloured. In other words, if we represent $\gamma$ in the stranded diagram picture, which in the $2$-dimensional case is just a ribbon graph, it can be drawn in such way that there are no crossing of lines. As a consequence, every closed path within $\gamma$ enclosed a set of faces of the graph and using all the corresponding delta functions allows to contract the path to a point. As an example, consider the graph drawn in figure \ref{fig:LemmaSphere}.
    \begin{figure}[H]
        \centering
        \includegraphics[scale=1.2]{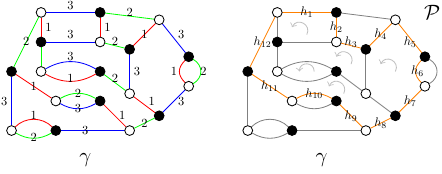}
        \caption{A planar $(3+1)$-coloured graph representing a $2$-sphere and a closed $3$-coloured path $\mathcal{P}$ (orange) equipped with group elements on its edges. Using all the faces enclosed in the path $\mathcal{P}$, the delta function associated to $\mathcal{P}$ can be replaced by $\delta^{N}(\mathds{1})$.\label{fig:LemmaSphere}}
    \end{figure}
    The example shows a closed $(2+1)$-coloured graph $\gamma$ representing a $2$-sphere, drawn in a planar representation, and the right-hand side shows a closed path, denoted by $\mathcal{P}$, within $\gamma$. Due to planarity, we can use all the delta functions corresponding to the faces enclosed by $\mathcal{P}$ in order to shrink $\mathcal{P}$ until it becomes a face of the graph itself and can hence be replaced by $\mathds{1}$:
    \begin{equation}
        \begin{aligned}
            &\delta^{N}(h_{1}^{-1}h_{2}h_{3}^{-1}h_{4}h_{5}^{-1}h_{6}h_{7}^{-1}h_{8}h_{9}^{-1}h_{10}h_{11}^{-1}h_{12})\xrightarrow{\text{enclosed faces}}\delta^{N}(\mathds{1})
        \end{aligned}
    \end{equation}
    In other words, due to planarity, the path can always be shrunk to identity by using all the delta functions, which are enclosed. The same is true if $\gamma$ is an open graph representing a disk as long as the closed path lies in the interior and is not touching the boundary. Note that from the topological point of view, the result is a consequence of the simply-connectedness of the $2$-sphere and $2$-disk, since every closed path can be contracted to a point. A more technical and rigorous proof of a similar statement can be found in \cite{GurauLargeN1,GurauLargeN3}.
\end{proof}

Using the above lemma, we can now show how amplitudes change under performing internal proper dipole moves, which essentially generalizes Lemma 6 in \cite{GurauLargeN3} to the case of graphs with boundary:

\begin{Lemma}
    \label{AmplChangeInt}
    Let $\mathcal{G}\in\mathfrak{G}_{3}$ with $\gamma:=\partial\mathcal{G}$ and $d_{k}$ with $k\in\{1,2,3\}$ be an internal proper $k$-dipole in $\mathcal{G}$. Then the amplitudes of $\mathcal{G}$ and $\mathcal{G}/d_{k}$ satisfy
    \begin{align*}
        &k=1:\hspace{0.5cm}\mathcal{A}_{\mathcal{G}}^{\lambda}[\{g_{e}\}_{e\in\mathcal{E}_{\gamma}}]=(\lambda\overline{\lambda})\mathcal{A}_{\mathcal{G}/d_{1}}^{\lambda}[\{g_{e}\}_{e\in\mathcal{E}_{\gamma}}]\\
        &k=2:\hspace{0.5cm}\mathcal{A}_{\mathcal{G}}^{\lambda}[\{g_{e}\}_{e\in\mathcal{E}_{\gamma}}]=(\lambda\overline{\lambda})\delta^{N}(\mathds{1})^{-1}\mathcal{A}_{\mathcal{G}/d_{2}}^{\lambda}[\{g_{e}\}_{e\in\mathcal{E}_{\gamma}}]\\
        &k=3:\hspace{0.5cm}\mathcal{A}_{\mathcal{G}}^{\lambda}[\{g_{e}\}_{e\in\mathcal{E}_{\gamma}}]=(\lambda\overline{\lambda})\mathcal{A}_{\mathcal{G}/d_{3}}^{\lambda}[\{g_{e}\}_{e\in\mathcal{E}_{\gamma}}].
    \end{align*}
\end{Lemma}

\begin{proof}
    We only prove the case of $k=1$ since the proofs for the other two cases are analogues.
    \bigskip

    We need to distinguish between the cases where the dipole edge has colour $i\neq 0$ or colour $i=0$. In the first case, the general situation is sketched in figure \ref{fig:LemmaAmplitudesDipole} below.
    \begin{figure}[H]
        \centering
        \includegraphics[scale=1.2]{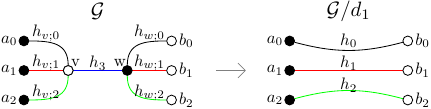}
        \caption{A $1$-dipole contraction involving an edge of colour $3$ and group elements assigned to all the edges.\label{fig:LemmaAmplitudesDipole}}
    \end{figure}
    We consider a $1$-dipole consisting of an edge, which without loss of generality is taken of colour $3$, connecting two internal vertices $v,w\in\mathcal{V}_{\mathcal{G},\mathrm{int}}$. Furthermore, we assume that the $3$-bubble $\mathcal{B}_{va_{0}a_{1}a_{2}}^{012}$ of colour $012$ containing the vertex $v$ represents a $2$-sphere, whereas the $3$-bubble $\mathcal{B}_{wa_{0}a_{1}a_{2}}^{012}$ of colour $012$ containing $w$ is allowed to be open and to have arbitrary topology. Note that the vertices $a_{i}$ do not have to be distinct and similar for the $b_{i}$'s. Furthermore, the vertex $b_{0}$ could in principle also be a $1$-valent boundary vertex. Now, let us denote the group elements living on the edges $\overline{va_{i}}$ by $h_{v;i}$, the group elements living on $\overline{b_{i}w}$ by $h_{w;i}$ and the group element assigned to the dipole edge $\overline{vw}$ by $h_{3}$. The contribution of all these edges to the Ponzano-Regge transition function is given by the following integrals:
    \begin{equation}
        \begin{aligned}
            \int_{\mathrm{SU}(2)^{7}}&\,\mathrm{d}h_{3}\bigg (\prod_{i=0}^{2}\mathrm{d}h_{v;i}\,\mathrm{d}h_{w;i}\bigg )\\
            &\delta^{N}(h_{v;0}h_{3}^{-1}h_{w;0}H^{03}[g])\delta^{N}(h_{v;1}h_{3}^{-1}h_{w;1}H^{13})\delta^{N}(h_{v;2}h_{3}^{-1}h_{w;2}H^{23})\\
            & \delta^{N}(h_{v;0}h_{v;1}^{-1}H_{v}^{01})\delta^{N}(h_{v;2}h_{v;0}^{-1}H_{v}^{02})\delta^{N}(h_{v;1}h_{v;2}^{-1}H_{v}^{12})\\
            &\delta^{N}(h_{w;1}^{-1}h_{w;0}H_{w}^{01}[g])\delta^{N}(h_{w;0}^{-1}h_{w;2}H_{w}^{02}[g])\delta^{N}(h_{w;2}^{-1}h_{w;1}H_{w}^{12}).
        \end{aligned}
    \end{equation}
    The group elements $H^{i3}$ for $i\in\{0,1,2\}$ denote the products of group elements assigned to the bicoloured path of colour $i3$ starting at $b_{i}$ and ending at $a_{i}$. The product $H^{03}$ could in principle contain a boundary group element, which is indicated by the notation $[g]$, since the corresponding face could be non-cyclic. The group elements $H^{ij}_{v}$ with $i,j\in\{0,1,2\}$ and $i<j$ are the product of the remaining group elements of the edges belonging to the faces of colour $ij$ containing the vertex $v$. Since the $3$-bubble $\mathcal{B}_{v}^{012}$ is closed, all these faces are cyclic and hence these products do not contain boundary group elements. Lastly, $H^{ij}_{w}$ with $i,j\in\{0,1,2\}$ and $i<j$ are the product of the remaining group elements of the edges belonging to the faces of colour $ij$ containing the vertex $w$. The faces of colour $01$ and $02$ could in principle be non-cyclic and hence $H^{01}_{w}$ and $H^{02}_{w}$ could again contain one of the boundary group elements, which we again indicate by $[g]$. To start with, let us change the variables $h_{w;i}$ to $h_{i}:=h_{v;i}h_{3}^{-1}h_{w;i}$ for all $i\in\{0,1,2\}$. Under this transformation, we see that the integrand is no longer dependent on $h_{3}$ and hence, we can integrate trivially over it. Using $\mathrm{d}h^{\prime}_{w;i}=\mathrm{d}h_{w;i}$, the contribution from the dipole becomes
    \begin{equation}
        \begin{aligned}
            \int_{\mathrm{SU}(2)^{6}}&\,\bigg (\prod_{i=0}^{2}\mathrm{d}h_{v;i}\,\mathrm{d}h_{i}\bigg )\\
            &\delta^{N}(h_{0}H^{03}[g])\delta^{N}(h_{1}H^{13})\delta^{N}(h_{2}H^{23})\\
            & \delta^{N}(h_{v;0}h_{v;1}^{-1}H_{v}^{01})\delta^{N}(h_{v;2}h_{v;0}^{-1}H_{v}^{02})\delta^{N}(h_{v;1}h_{v;2}^{-1}H_{v}^{12})\\
            &\delta^{N}(h_{1}^{-1}h_{v;1}h^{-1}_{v;0}h_{0}H_{w}^{01}[g])\delta^{N}(h_{0}^{-1}h_{v;0}h^{-1}_{v;2}h_{2}H_{w}^{02}[g])\delta^{N}(h_{2}^{-1}h_{v;2}h^{-1}_{v;1}h_{1}H_{w}^{12}).
        \end{aligned}
    \end{equation}
    Next, we can use the two delta functions $\delta^{N}(h_{v;0}h_{v;1}^{-1}H_{v}^{01})$ and $\delta^{N}(h_{v;2}h_{v;0}^{-1}H_{v}^{02})$ to integrate over the group elements $h_{v;1}$ and $h_{v;2}$, which results into the replacements $h_{v;1}:=H^{01}_{v}h_{v;0}$ and $h^{-1}_{v;2}:=h_{v;0}^{-1}H^{02}_{v}$. We are left with
    \begin{equation}
        \begin{aligned}
            \int_{\mathrm{SU}(2)^{4}}&\,\mathrm{d}h_{0;v}\mathrm{d}h_{0}\mathrm{d}h_{1}\mathrm{d}h_{2}\\
            &\delta^{N}(h_{0}H^{03}[g])\delta^{N}(h_{1}H^{13})\delta^{N}(h_{2}H^{23})\\
            &\delta^{N}(H^{01}_{v}H^{02}_{v}H^{12}_{v})\\
            &\delta^{N}(h_{1}^{-1}H^{01}_{v}h_{0}H^{01}_{w}[g])\delta^{N}(h_{0}^{-1}H^{02}_{v}h_{2}H^{02}_{w}[g])\delta^{N}(h_{2}^{-1}H^{12}_{v}h_{1}H^{12}_{w}).
        \end{aligned}
    \end{equation}
    We see that the integration over $h_{v;0}$ is now trivial and so it can be taken out thanks to the Haar measure normalisation.

    \bigskip
    The interpretation of this result is as follows. We integrate over three group elements $h_{0}$, $h_{1}$, $h_{2}$ which are the group elements living on the three edges $\overline{a_{i}b_{i}}$ in the graph $\mathcal{G}/d_{1}$. The first row of delta functions corresponds to the bicoloured paths $i3$ for $i\in\{0,1,2\}$ containing one of the three edges $\overline{a_{i}b_{i}}$. For the third line, before contracting the dipole, we had for each pair $ij$ with $i,j\in\{0,1,2\}$ and $i<j$ precisely two bicoloured faces in our integration, one containing $v$ and one containing $w$. After contracting the dipole, we get rid of the colour $3$ edge and connect all the lines with colours $i\in\{0,1,2\}$ to each other. As a consequence, we combine for each $i,j$ the two bicoloured paths, which before contracting the dipole were disconnected by the colour $3$ edge. To sum up, the third line of delta functions corresponds to all the faces with colour $i,j\in\{0,1,2\}$ of the graph containing two of the edges $\overline{a_{i}b_{i}}$. At the end of the day, we see that the first and third line of our result above precisely corresponds to the contribution of the three edges $\overline{a_{i}b_{i}}$ of the contracted graph $\mathcal{G}/d_{1}$. Hence, we have related the amplitude of $\mathcal{G}$ with the amplitude of $\mathcal{G}/d_{1}$ up to the additional factor of $\delta^{N}(H^{01}_{v}H^{02}_{v}H^{12}_{v})$. To get rid of this term, we make use of the assumption that the bubble $\mathcal{B}^{012}_{va_{0}a_{1}a_{2}}$ is spherical. The product $H^{01}_{v}H^{02}_{v}H^{12}_{v}$ corresponds to a closed $3$-coloured path, which completely lies within the graph obtained by cutting the vertex $v$ from the spherical $3$-bubble $\mathcal{B}_{va_{0}a_{1}a_{2}}^{012}$. The latter is a graph representing the $2$-disk and since all the delta function corresponding to closed faces of this planar graph are also contained in the amplitude of $\mathcal{G}/d_{1}$, we can replace $\delta^{N}(H^{01}_{v}H^{02}_{v}H^{12}_{v})$ by $\delta^{N}(\mathds{1})$, according to Lemma \ref{LemmaSphere}. As a consequence, taking into account that we reduce the number of internal vertices by two, we have that
    \begin{align}
        \mathcal{A}_{\mathcal{G}}^{\lambda}[\{g_{e}\}_{e\in\mathcal{E}_{\gamma}}]=\frac{\lambda\overline{\lambda}}{\delta^{N}(\mathds{1})}\delta^{N}(\mathds{1})\mathcal{A}_{\mathcal{G}/d_{1}}^{\lambda}[\{g_{e}\}_{e\in\mathcal{E}_{\gamma}}]=(\lambda\overline{\lambda})\cdot\mathcal{A}_{\mathcal{G}/d_{1}}^{\lambda}[\{g_{e}\}_{e\in\mathcal{E}_{\gamma}}],
    \end{align}
    which concludes the proof. In the second case, i.e.~the case where the dipole edge is of colour $0$, the proof is exactly the same with the difference that now all the faces containing the dipole edge could contain a boundary group element and none of the faces containing the vertex $w$.
\end{proof}

\begin{Remark}
    The reason for the additional factor of $\delta^{N}(\mathds{1})$ which cancels the factor $1/\delta^{N}(\mathds{1})$ is not the same for $1$ and $3$-dipoles move. Instead of obtaining a redundant delta as in the $1$-dipole move, for $3$-dipoles, the amplitudes of $\mathcal{G}$ and $\mathcal{G}/d_{3}$ can be directly related, but there is by definition one redundant face within in the $3$-dipoles giving the factor of $\delta^{N}(\mathds{1})$.
\end{Remark}

Next, let us generalize the bubble rooting procedure introduced in \cite{GurauLargeN1,GurauLargeN2,GurauLargeN3} to the case of open graphs. Let $\gamma\in\overline{\mathfrak{G}}_{2}$ be a closed $(2+1)$-coloured graph --our boundary graph-- and let $\mathcal{G}\in\mathfrak{G}_{3}$ be a connected and open $(d+1)$-coloured graph with $\partial\mathcal{G}=\gamma$. For every colour $i\in\mathcal{C}_{d}$, we have two possibilities:
\begin{itemize}
    \item[(1)]All $d$-bubbles without colour $i$ are closed and represent $d$-spheres.
    \item[(2)]There exists at least one $d$-bubble without colour $i$, which is not spherical. Note that this includes both the case of open and closed but not spherical $d$-bubble.
\end{itemize}

\begin{Remark}
    For any graph $\mathcal{G}$ in $\mathfrak{G}_{3}$ with $\partial\mathcal{G}$ non-empty, property (1) can only be satisfied in the case $i=0$, since for any $2$-bubble $\mathcal{B}$ in $\partial\mathcal{G}$ of colour $ij$ for $i,j\in\{1,2,3\}$, there exists a $3$-bubble in $\mathcal{G}$ of colour $0ij$, whose boundary contains $\mathcal{B}$ as a connected component. If $\mathcal{G}$ represents a manifold, then property (1) is necessarily satisfied in the case $i=0$, since all its internal $3$-bubbles represent $3$-spheres.
\end{Remark}

In case (1), we choose one of the spherical $3$-bubbles without colour $i$ as ``\textit{principal root}'' and denote it by $\mathcal{R}^{i}_{(1)}$. In case (2), we choose one of the non-spherical $3$-bubbles without colour $i$ as ``\textit{principal root}'' $\mathcal{R}^{i}_{(1)}$ and all other non-spherical $3$-bubbles without colour $i$ as ``\textit{branching roots}'', which we denote by $\mathcal{R}^{i}_{(\mu)}$ with some labelling parameter $\mu$. Next, we need the ``connectivity graph of colour $i$'', which is defined as follows:

\begin{Definition}[Connectivity Graphs]
    Let $\mathcal{G}\in\mathfrak{G}_{3}$ be some open $(3+1)$-coloured graph. Then the ``connectivity graph of colour $i\in\mathcal{C}_{3}$'' is the graph $\mathcal{C}^{i}[\mathcal{G}]$ defined as follows:
    \begin{itemize}
        \item[(1)]There is a vertex in $\mathcal{C}^{i}[\mathcal{G}]$ for each $3$-bubbles without colour $i$ in $\mathcal{G}$.
        \item[(2)]Two vertices in $\mathcal{C}^{i}[\mathcal{G}]$ are connected by an edge if and only if there is an edge of colour $i$ in $\mathcal{G}$ connecting the two $3$-bubbles corresponding to the two vertices.
    \end{itemize}
\end{Definition}

\begin{Remark}
    The connectivity graphs corresponding to some coloured graph $\mathcal{G}$ are in general pseudographs, i.e.~multigraphs in which also tadpole lines (=edges starting and ending at the same vertex) are allowed.
\end{Remark}

The bubble rooting procedure is now defined via the following algorithm:

\begin{itemize}
    \item[(1)]Take the graph $\mathcal{C}^{0}[\mathcal{G}]$ and choose a maximal tree $\mathcal{T}^{0}$ in it. There are two different types of vertices in this graph, namely the roots and all the other vertices representing spherical $3$-bubbles. Now, every vertex is connected to the principal root by a unique path contained in the maximal tree. For each branch root, let us draw the incident edge belonging to the tree, which is contained in this path, as a dashed line. All the other edges we draw as solid lines.
    \item[(2)]The solid lines in the tree $\mathcal{T}^{0}$ are internal proper $1$-dipoles and we contract them. Repeating this procedure for all solid lines, we are left with either a unique $\hat{0}$-bubble, which is spherical, or with a bunch of non-spherical $\hat{0}$-bubbles.
    \item[(3)]Next, we choose a maximal tree $\mathcal{T}^{1}$ in the $1$-connectivity graph in the graph obtained \textit{after} step (2). Note that this tree in general depends on the tree $\mathcal{T}^{0}$. Then, we repeat step (2), i.e.~we contract the internal proper $1$-dipole corresponding to the solid lines.
    \item[(4)]We repeat this procedure for all other colours by choosing trees $\mathcal{T}^{j}$ for all $j\in\{0,\dots,3\}$, which depend on the chosen trees $\mathcal{T}^{j-1},\dots,\mathcal{T}^{0}$.
\end{itemize}

The procedure described in the algorithm above for some colour $j\in\mathcal{C}_{3}$ is sketched in the figure \ref{fig:Rooting} below.

\begin{figure}[H]
    \centering
    \includegraphics[scale=1.1]{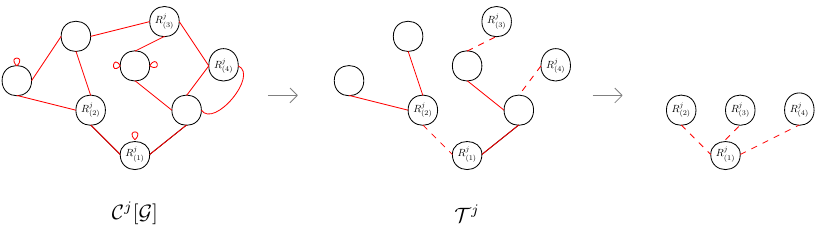}
    \caption{The rooting procedure per iteration $j$.\label{fig:Rooting}}
\end{figure}

The graph obtained after some iteration $j$ clearly depends on the choice of tree $\mathcal{T}^{j}$. However, the obtained graph is independent of the order of proper $1$-dipoles contracted within some fixed tree. Furthermore, note that the procedure is well defined, since when we contract the tree of colour $j$, we do not change the number of internal proper $1$-dipoles of colours $\neq j$. This is a consequence of the following lemma:

\begin{Lemma}
    \label{LemmaCoreGraphs}
    Let $\mathcal{G}$ be an open $(d+1)$-coloured graph. An internal proper $1$-dipole move of colour $i\in\mathcal{C}_{d}$ does not change the number as well as the topologies of all $d$-bubbles of $\mathcal{G}$ involving colour $i$.
\end{Lemma}

\begin{proof}
    It is clear that we do not touch the number of all these bubbles. For the second claim, observe that every internal proper $1$-dipole move in $\mathcal{G}$ consisting of an edge $e\in\mathcal{E}_{\mathcal{G}}$ of colour $i\in\mathcal{C}_{d}$ also corresponds to an internal proper $1$-dipole move within all the $d$-bubbles of $\mathcal{G}$ involving the edge $e$, because every $(d-1)$-bubble contained in the spherical bubble separated by the dipole is itself a closed bubble representing a $(d-2)$-sphere.
\end{proof}

In other words, contracting the connectivity graph of some colour $j$ might change the connectivity graphs of colour $i\neq j$, but there is still either a unique spherical $\hat{i}$-bubble for $i<j$, or all the $\hat{i}$-bubbles for $i<j$ are non-spherical and hence, we do not produce any new internal proper $1$-dipole of colour $i<j$. Furthermore, the number of dipoles we can contract in the connectivity graphs with $i>j$ stays the same, since the number of spherical bubbles and roots stays the same. After this procedure, we are left with a graph in which we cannot perform any more internal proper $1$-dipole contractions. In accordance to \cite{GurauLargeN1,GurauLargeN2,GurauLargeN3}, we call these objects ``core graphs'':

\begin{Definition}[Core Graphs with Boundary]
    A ``core graph with boundary $\gamma$ of order $p$'' is an open $(d+1)$-coloured graph $\mathcal{G}_{p}\in\mathfrak{G}_{3}$ with $2p$ internal vertices, such that $\partial\mathcal{G}_{p}=\gamma$ and such that for every colour $i\in\mathcal{C}_{3}$ one of the following applies:
    \begin{itemize}
        \item[(1)]There is unique closed and spherical $3$-bubble without colour $i$.
        \item[(2)]All $3$-bubbles without colour $i$ are non-spherical. Note that this includes both non-spherical closed $3$-bubbles as well as open $3$-bubbles.
    \end{itemize}
\end{Definition}

\begin{Remark}
    A closed core graph representing a manifold is nothing else than a crystallization as defined in Section \ref{SecCryTheo}. For open graphs, this is in general not true. While every crystallization of a manifold with boundary is clearly a core graph, the reverse is only true if we choose a boundary graph, which is by itself a core graph.
\end{Remark}

The core graph obtained from a coloured graph by the rooting procedure introduced above does in general depend on the chosen trees. However, their amplitudes are independent of these choices:

\begin{Proposition}
    \label{PropCoreGraph}
    Let $\mathcal{G}\in\mathfrak{G}_{3}$ be an open $(3+1)$-coloured graph with boundary graph $\gamma:=\partial\mathcal{G}$ and $\mathcal{G}_{c}$ some core graph obtained by rooting $\mathcal{G}$. Then the order of $\mathcal{G}_{c}$ is given by
    \begin{align*}
        p_{c}=\frac{\vert\mathcal{V}_{\mathcal{G}_{c},\mathrm{int}}\vert}{2}=\frac{\vert\mathcal{V}_{\mathcal{G},\mathrm{int}}\vert}{2}-(\mathcal{B}^{[d]}-\mathcal{R}^{[d]}),
    \end{align*}
    where $\mathcal{R}^{[d]}$ denotes the total number of roots in $\mathcal{G}$. Their associated Boulatov amplitudes are related by
    \begin{align*}
        \mathcal{A}^{\lambda}_{\mathcal{G}}[\{g_{e}\}_{e\in\mathcal{E}_{\gamma}}]=(\lambda\overline{\lambda})^{\mathcal{B}^{[d]}-\mathcal{R}^{[d]}}\mathcal{A}^{\lambda}_{\mathcal{G}_{c}}[\{g_{e}\}_{e\in\mathcal{E}_{\gamma}}].
    \end{align*}
\end{Proposition}

\begin{proof}
    The number of contracted $1$-dipoles does not depend on any choices by Lemma \ref{LemmaCoreGraphs} and one can easily convince oneself that it is given by $\mathcal{B}^{[d]}-\mathcal{R}^{[d]}$. The second claim follows from Lemma \ref{AmplChangeInt}.
\end{proof}

Therefore, it makes sense to introduce the following notion of equivalence:

\begin{Definition}[Core Equivalence Classes]
    We shall call two core graphs with the same boundary, the same amplitude, the same topology and the same order ``core equivalent''. This defines an equivalence relation $\sim_{\mathrm{c}}$ and we denote the set of all equivalences classes for some given boundary graph $\gamma\in\overline{\mathfrak{G}}_{2}$ by
    \begin{align*}
        \mathfrak{G}^{\mathrm{core}}_{\gamma}:=\{\mathcal{G}\in\mathfrak{G}_{3}\mid\mathcal{G}\text{ is core graph and }\partial\mathcal{G}=\gamma\}/\sim_{\mathrm{c}}.
    \end{align*}
    Furthermore, let us decompose this set as $\mathfrak{G}^{\mathrm{core}}_{\gamma}=\bigcup_{p=\vert\mathcal{V}_{\gamma}\vert/2}^{\infty}\mathfrak{G}^{\mathrm{core}}_{p,\gamma}$, where $\mathfrak{G}^{\mathrm{core}}_{p,\gamma}$ denotes the subsets containing all the core equivalence classes with boundary $\gamma$ for some fixed order $p$.
\end{Definition}

\begin{Remark}
    \label{SMG}
    The smallest order core equivalence class for some given boundary graph $\gamma\in\overline{\mathfrak{G}}_{2}$ has order $p=\vert\mathcal{V}_{\gamma}\vert/2$ and contains only one representative, namely the open graph obtained by adding an external leg of colour $0$ to each vertex in $\gamma$. We call these graphs the ``smallest matching graphs''. See Section \ref{SectionV:LeadingOrder} for more details.
\end{Remark}

\subsection{Topological Expansion of the Transition Amplitude}
Let $\gamma\in\overline{\mathfrak{G}}_{2}$ be some closed $(2+1)$-coloured graph of arbitrary topology and $\Psi=(\gamma,\rho,i)$ be some spin network state living on it. As explained in the last section, every open $(3+1)$-coloured graph $\mathcal{G}\in\mathfrak{G}_{3}$ can be rooted to some core graph $\mathcal{G}_{c}$. As previously stated, the resulting core graph depends on the choice of trees in the rooting procedure, however, following Proposition \ref{PropCoreGraph}, every other open $(3+1)$-open coloured core graph $\widetilde{\mathcal{G}}_{c}$ obtained from $\mathcal{G}$ is core equivalent to $\mathcal{G}_{c}$: $\mathcal{G}_{c}\sim_{c}\widetilde{\mathcal{G}}_{c}$. In other words, every open $(3+1)$-coloured graph $\mathcal{G}$ can be rooted into a \textit{unique} core equivalence class $[\mathcal{G}_{c}]$. This motivates the expansion of the transition amplitude \eqref{transampl} in terms of the core equivalence classes
\begin{align}
    \langle\mathcal{Z}_{\mathrm{cBM}}\vert\Psi\rangle=\sum_{p=\vert\mathcal{V}_{\gamma}\vert/2}^{\infty}\,\sum_{[\mathcal{G}]\in \mathfrak{G}^{\mathrm{core}}_{p,\gamma}}C^{[\mathcal{G}]}(\lambda,\overline{\lambda})\mathcal{A}^{\lambda}_{[\mathcal{G}]}[\Psi],
\end{align}
where $C^{[\mathcal{G}]}(\lambda,\overline{\lambda})$ is a combinatorial factor counting all the factors of $\lambda\overline{\lambda}$ coming from graphs, which root back to some graph in the equivalence class $[\mathcal{G}]$ as well as their symmetry factors. More precisely, the factor of some core equivalence class $[\mathcal{G}]$ of order $p$ can be written as

\begin{align}
    C^{[\mathcal{G}]}(\lambda,\overline{\lambda}):=\sum_{\mathcal{G}\in\mathfrak{G}_{3}\text{ with }\partial\mathcal{G}=\gamma\text{ and }\mathcal{G}\to [\mathcal{G}]}\frac{(\lambda\overline{\lambda})^{\frac{\vert\mathcal{V}_{\mathcal{G},\mathrm{int}}\vert}{2}-p}}{\mathrm{sym}(\mathcal{G})},
\end{align}
where the sum is over all open $(3+1)$-coloured graphs with boundary $\gamma$ and which can be rooted to one of the members of the core equivalence class $[\mathcal{G}]$. Note that this combinatorial factors do not contain the cutoff parameter $N$, since all the divergences are contained in the amplitude of the corresponding core equivalence class. This is also the main reason for our choice of scaling. Indeed, with this choice, internal proper dipole $1$-moves do not change the degree of divergence and all the graphs rooting back to some given core equivalence class have the same power of $\delta^{N}(\mathds{1})$.
\bigskip

The expansion written above is a \textit{topological expansion}, in the sense that each term in the sum corresponds to some fixed bulk topology. Note that
\begin{itemize}
    \item[(1)]Two core graphs at the same order $p$ might have the same amplitude, but might not be topological equivalent.
    \item[(2)]Conversely, two core graphs at the same order $p$ might be topological equivalent but still have different amplitudes.
\end{itemize}

To sum up, every core equivalence class represents a fixed topology but there are in general an infinite number of distinct equivalence classes representing the same (pseudo)manifold. In fact, for every topology there is a smallest order $p\in\mathbb{N}$ for which there is a core graph representing it and it exists core graphs for \textit{all} higher orders. Examples can be obtained by performing internal proper $2$-dipole moves.
\section{Spherical Boundary and Factorization}\label{SectionIV:Sphere}
In this section, we apply the above formalism to the simplest possible boundary topology, the $2$-sphere. We start by considering the simplest possible boundary graph representing the $2$-sphere, the elementary melonic $2$-sphere, and show that the transition amplitude, restricted to topologies without singularities touching the boundary, is proportional to the spin network evaluation. Afterwards, we extend this result to the next-to-simplest boundary graph representing the $2$-sphere, the pillow graph, and to arbitrary spherical boundary graphs. In other words, we show that the transition amplitude with respect to some spherical boundary graph factorizes and only depends on boundary data. Afterwards, we argue by briefly discussing the next-to-simplest boundary topology, namely the $2$-torus, that the transition amplitudes contains non-trivial information about the admissible bulk topology, a fact careful hidden in the ball case due to the simple result obtained.

\subsection{Simplest Boundary Graph Representing the 2-Sphere}
As an example of the formalism developed so far and to fix ideas, let us discuss the simplest possible example: a spherical boundary topology with the ``\textit{elementary melonic $2$-sphere}\footnote{Elementary melonic spheres are also known as ``\textit{dipoles}'' in the literature. However, we have already used this name for the concept of dipole moves.}'' as boundary graph $\gamma$. It is represented, together with its triangulation, in figure \ref{fig:MelonicSphere}.

\begin{figure}[H]
    \centering
    \includegraphics[scale=1.2]{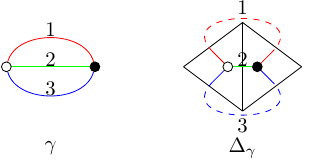}
    \caption{The elementary melonic 2-sphere $\gamma$ (l.h.s.) and its corresponding simplicial complex (r.h.s.), where the gluing of edges is as indicated by the dashed lines.\label{fig:MelonicSphere}}
\end{figure}

Figure \ref{SimpleExampleFig} below shows five open $(3+1)$-coloured graphs with boundary given by $\gamma$. Each of them is a core graph and defines a distinct core equivalence class. They are in fact all inequivalent core graphs up to order $p=3$:\footnote{Note that these graphs exactly correspond to the radiative corrections of the propagator from the group field theoretic point of view.}

\begin{figure}[H]
    \centering
    \includegraphics[scale=0.9]{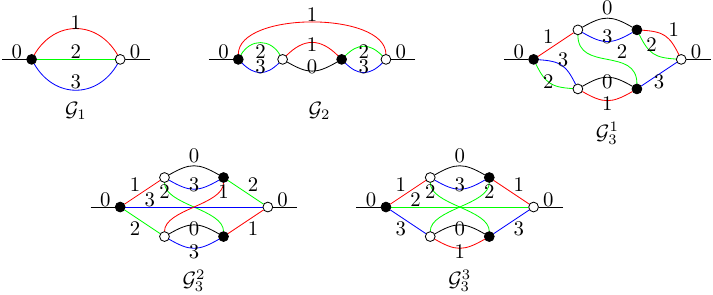}
    \caption{Representative of all inequivalent core equivalence classes with boundary $\gamma$ up to order $p=3$.\label{SimpleExampleFig}}
\end{figure}

The core graphs drawn in the first line --$\mathcal{G}_{1}$, $\mathcal{G}_{2}$ and $\mathcal{G}_{3}^{1}$-- represent $3$-balls: The graph $\mathcal{G}_{1}$ is usually called the ``\textit{elementary melonic $3$-ball}'' \cite{GurauColouredTensorModelsReview2} and the graphs $\mathcal{G}_{2}$ and $\mathcal{G}_{3}$ can be reduced to $\mathcal{G}_{1}$ by performing internal proper $2$-dipole moves. The two core graphs $\mathcal{G}_{3}^{2}$ and $\mathcal{G}^{3}_{3}$ both represent pseudomanifolds, which can be seen from the fact that they both contain a $3$-bubble of toroidal topology. They are however not homeomorphic (and not even homotopy equivalent): the Euler characteristic of $\Delta_{\mathcal{G}_{3}^{2}}$ is $\chi(\Delta_{\mathcal{G}_{3}^{2}})=3$ whereas the Euler characteristic of $\Delta_{\mathcal{G}_{3}^{3}}$ is $\chi(\Delta_{\mathcal{G}_{3}^{3}})=2$.\footnote{The Euler characteristic is a homotopy invariant of general CW-complexes and therefore in particulat also of pseudomanifolds \cite{Hatcher}. }
\bigskip

A straightforward calculation gives the Ponzano-Regge amplitudes together with their prefactor coming from the interaction term corresponding to the five core equivalence classes represented above:
\begin{subequations}
    \begin{align}
        \mathcal{A}_{[\mathcal{G}_{1}]}^{\lambda}[\{g_{1},g_{2},g_{3}\}]&=(\lambda\overline{\lambda})\delta^{N}(g_{1}g_{3}^{-1})\delta^{N}(g_{2}g_{3}^{-1})\\
        \mathcal{A}_{[\mathcal{G}_{2}]}^{\lambda}[\{g_{1},g_{2},g_{3}\}]&=(\lambda\overline{\lambda})^{2}[\delta^{N}(\mathds{1})]^{-1}\delta^{N}(g_{1}g_{3}^{-1})\delta^{N}(g_{2}g_{3}^{-1})\\
        \mathcal{A}_{[\mathcal{G}_{3}^{1}]}^{\lambda}[\{g_{1},g_{2},g_{3}\}]&=(\lambda\overline{\lambda})^{3}[\delta^{N}(\mathds{1})]^{-2}\delta^{N}(g_{1}g_{3}^{-1})\delta^{N}(g_{2}g_{3}^{-1})\\
        \mathcal{A}_{[\mathcal{G}_{3}^{2}]}^{\lambda}[\{g_{1},g_{2},g_{3}\}]&=(\lambda\overline{\lambda})^{3}[\delta^{N}(\mathds{1})]^{-3}\delta^{N}(g_{1}g_{3}^{-1})\delta^{N}(g_{2}g_{3}^{-1})\int_{\mathrm{SU}(2)^{3}}\mathrm{d}x\mathrm{d}y\mathrm{d}z\,\delta^{N}(xyz(zyx)^{-1})\\
        \mathcal{A}_{[\mathcal{G}_{3}^{3}]}^{\lambda}[\{g_{1},g_{2},g_{3}\}]&=(\lambda\overline{\lambda})^{3}[\delta^{N}(\mathds{1})]^{-2}\delta^{N}(g_{1}g_{3}^{-1})\delta^{N}(g_{2}g_{3}^{-1})
    \end{align}
\end{subequations}
The group element $g_{i}$ is assigned to the boundary edge of colour $i\in\{1,2,3\}$. In all three cases, the amplitudes encode the flatness of the boundary connections, as expected from the Bianchi identity. The remaining integral in the amplitude of graph $\mathcal{G}_{3}^{2}$ comes from the non-trivial bulk topology. Note also that the amplitudes of $\mathcal{G}_{2}$ and $\mathcal{G}_{3}^{1}$ can be obtained from the amplitude of $\mathcal{G}_{1}$, by applying Lemma \ref{AmplChangeInt}.
\bigskip

By the Theorem (\ref{ThmCavGag}), we know that, at the very least, all manifolds with spherical boundary appear in the transition amplitude. Let us discuss some explicit examples of other manifolds appearing in the transition amplitudes. Note that every compact, orientable and connected $3$-manifold $\mathcal{M}$ with boundary $\partial\mathcal{M}\cong S^{2}$ can be obtained by cutting out the interior of a (sufficiently nicely) embedded ball inside some closed, orientable and connected $3$-manifold $\mathcal{N}$. Hence, every open $(3+1)$-coloured graph representing a manifold, whose boundary graph is given by the elementary melonic $2$-sphere $\gamma$, can be obtained by cutting an edge of colour $0$ in some closed $(3+1)$-coloured graph representing a closed $3$-manifold. As an example, consider the three non-trivial graphs of figure \ref{Closed3Manifolds} representing $S^{2}\times S^{1}$, $\mathbb{R}P^{3}\cong L(2,1)$ and $L(3,1)$. The graphs obtained by cutting an edge of colour $0$ are drawn figure \ref{fig:ManifoldsSphereBoundary}:\footnote{Note that the manifold obtained by cutting out the interior of a ball from some closed manifold does not depend on the chosen ball. This follows essentially from the annulus theorem \cite{KirbyAT,Quinn}.}

\begin{figure}[H]
    \captionsetup[subfigure]{labelformat=empty}
    \centering
    \subfloat[$(S^{2}\times S^{1})\backslash\mathring{B}^{3}$]{\includegraphics[width=0.15\textwidth]{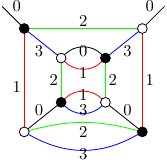}}\hspace{1cm}
    \subfloat[$\mathbb{R}P^{3}\backslash\mathring{B}^{3}$]{\includegraphics[width=0.15\textwidth]{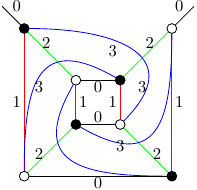}}\hspace{1cm}
    \subfloat[$L(3,1)\backslash\mathring{B}^{3}$]{\includegraphics[width=0.15\textwidth]{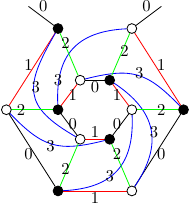}}
    \caption{Three open $(3+1)$-coloured graphs representing manifolds with spherical boundary.\label{fig:ManifoldsSphereBoundary}}
\end{figure}

A straightforward calculation gives the following amplitudes for each manifold:
\begin{subequations}
    \begin{align}
        \mathcal{A}_{S^{2}\times S^{1}\backslash\mathring{B}^{3}}^{\lambda}[\{g_{1},g_{2},g_{3}\}]&=(\lambda\overline{\lambda})^{4}[\delta^{N}(\mathds{1})]^{-2}\delta^{N}(g_{1}g_{3}^{-1})\delta^{N}(g_{2}g_{3}^{-1})\\\mathcal{A}_{\mathbb{R}P^{3}\backslash\mathring{B}^{3}}^{\lambda}[\{g_{1},g_{2},g_{3}\}]&=(\lambda\overline{\lambda})^{4}[\delta^{N}(\mathds{1})]^{-3}\delta^{N}(g_{1}g_{3}^{-1})\delta^{N}(g_{2}g_{3}^{-1})\,\int_{\mathrm{SU}(2)}\,\mathrm{d}k\,\delta^{N}(k^{2})\\\mathcal{A}_{L(3,1)\backslash\mathring{B}^{3}}^{\lambda}[\{g_{1},g_{2},g_{3}\}]&=(\lambda\overline{\lambda})^{6}[\delta^{N}(\mathds{1})]^{-5}\delta^{N}(g_{1}g_{3}^{-1})\delta^{N}(g_{2}g_{3}^{-1})\,\int_{\mathrm{SU}(2)}\,\mathrm{d}k\,\delta^{N}(k^{3})
    \end{align}
\end{subequations}

We see that all of them are proportional to the spin network evaluation when applied to some boundary spin network state. Let us now prove that this is true more generally.

\begin{Proposition}
    If $\mathcal{G}$ is an open $(3+1)$-coloured graph with $\partial\mathcal{G}=\gamma$ representing a manifold, then its amplitude satisfies
    \begin{align*}
        \mathcal{A}^{\lambda}_{\mathcal{G}}[\{g_{e}\}_{e\in\mathcal{E}_{\gamma}}]\propto\delta^{N}(g_{1}g_{2}^{-1})\delta^{N}(g_{1}g_{3}^{-1}).
    \end{align*}
    The same holds true for pseudomanifolds, for which all the singularities are in the bulk. As a consequence, we get that
    \begin{align*}
        \langle\mathcal{Z}_{\mathrm{cBM}}\vert\Psi\rangle\vert_{\substack{\mathrm{manifolds+pseudomanifolds}\\\mathrm{without}\hspace*{0.1cm}\mathrm{boundary}\hspace*{0.1cm}\mathrm{singularities}}}=\mathcal{C}[N,\lambda,\overline{\lambda}]\cdot\psi(\{g_{i}=\mathds{1}\}_{i=1,2,3}).
    \end{align*}
\end{Proposition}

\begin{proof}
    Figure \ref{SimpleExampleProof} below shows the boundary graph $\gamma$ equipped with group elements $g_{1,2,3}\in\mathrm{SU}(2)$, as well as the general structure of an open $(3+1)$-coloured graphs $\mathcal{G}\in\mathfrak{G}_{3}$ with $\partial\mathcal{G}=\gamma$.
    \begin{figure}[H]
        \centering
        \includegraphics[scale=1.1]{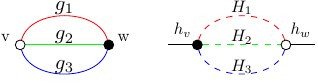}
        \caption{The graph $\gamma$ equipped with group elements $g_{1,2,3}$ as well as a sketch of an open $(3+1)$-coloured graph $\mathcal{G}$ with $\partial\mathcal{G}$. The dotted lines represent the non-cyclic faces, which lead to the corresponding boundary edges.\label{SimpleExampleProof}}
    \end{figure}
    The dotted lines in the graph $\mathcal{G}$ represent the non-cyclic faces of $\mathcal{G}$, which lead to the corresponding boundary edges. We equip these paths by the product of group elements $H_{1,2,3}$. Note that these paths do not have to be independent: there could be an internal edge of colour $0$ in $\mathcal{G}$, which is then contained in several paths. The external legs of $\mathcal{G}$ are labelled by the corresponding vertices of the boundary graph and we equip them with group elements $h_{v,w}$. If $\mathcal{G}$ represents a pseudomanifold without boundary singularities, we know that all its open $3$-bubbles represent $2$-balls. Hence, by Lemma \ref{LemmaSphere}, we know that the group elements $H_{1,2,3}$ satisfy
    \begin{subequations}
        \begin{align}
            &H_{1}H_{2}^{-1}\stackrel{\mathrm{CF}}{=}\mathds{1}\\
            &H_{1}H_{3}^{-1}\stackrel{\mathrm{CF}}{=}\mathds{1}\\
            &H_{2}H_{3}^{-1}\stackrel{\mathrm{CF}}{=}\mathds{1}
        \end{align}
    \end{subequations}
    where ``$\stackrel{\mathrm{CF}}{=}$'' means using all the delta functions associated to closed faces of $\mathcal{G}$, since these three products describe closed $3$-coloured paths living in some open $3$-bubble representing the disk. Using this notation, let us write down the contribution to the amplitude coming from all the faces involving boundary group elements:
    \begin{equation}
        \begin{aligned}
            \int_{\mathrm{SU}(2)^{2}}\,&\mathrm{d}h_{v}\,\mathrm{d}h_{w}\,            \delta^{N}(g_{1}h_{v}^{-1}H_{1}h_{w}^{-1})\delta^{N}(g_{2}h_{v}^{-1}H_{2}h_{w}^{-1})\delta^{N}(g_{3}h_{v}^{-1}H_{3}h_{w}^{-1})
        \end{aligned}
    \end{equation}
    Integrating over $h_{v}$ using the first delta function, this becomes
    \begin{equation}
        \begin{aligned}
            \int_{\mathrm{SU}(2)^{2}}\,&\mathrm{d}h_{v}\,\mathrm{d}h_{w}\,  \delta^{N}(g_{2}g_{1}^{-1}h_{w}\underbrace{H_{1}^{-1}H_{2}}_{\xrightarrow{\text{closed faces}}\mathds{1}}h_{w}^{-1})\delta^{N}(g_{3}g_{1}^{-1}h_{w}\underbrace{H_{1}^{-1}H_{3}}_{\xrightarrow{\text{closed faces}}\mathds{1}}h_{w}^{-1}),
        \end{aligned}
    \end{equation}
    which is equivalent to $\delta^{N}(g_{1}g_{2}^{-1})\delta^{N}(g_{1}g_{3}^{-1})$ when using the closed faces of $\mathcal{G}$ and the relations explained above. Therefore, the contribution $\mathcal{A}_{\mathcal{G}}^{\lambda}[\Psi]$ for some spin network $\Psi$ living on $\gamma$ with spin network function $\psi$ is proportional to
    \begin{equation}
        \begin{aligned}
            \int_{\mathrm{SU}(2)^{3}}\,&\mathrm{d}g_{1}\mathrm{d}g_{2}\mathrm{d}g_{3}\,\delta^{N}(g_{1}g_{2}^{-1})\delta^{N}(g_{1}g_{3}^{-1})\psi(g_{1},g_{2},g_{3})=\int_{\mathrm{SU}(2)}\,\mathrm{d}g\,\psi(g,g,g)=\psi(\mathds{1},\mathds{1},\mathds{1}),
        \end{aligned}
    \end{equation}
    where we have used the $\mathrm{SU}(2)$-invariance of $\psi$ in the last step.
\end{proof}

\begin{Remark}
    For the case of manifolds, there is also an alternative proof using dipole moves, i.e.~see \cite[p.92f.]{GabrielThesis}.
\end{Remark}

To summarize: we can write the transition amplitude (restricted to manifolds and pseudomanifolds without boundary singularities) for some arbitrary spin network $\Psi$ on $\gamma$ with corresponding spin network function $\gamma\in L^{2}(\mathrm{SU}(2)^{3}/\mathrm{SU}(2)^{2})$ in the following form:
\begin{align}
     \langle\mathcal{Z}_{\mathrm{cBM}}\vert\Psi\rangle\vert_{\substack{\mathrm{manifolds+pseudomanifolds}\\\mathrm{without}\hspace*{0.1cm}\mathrm{boundary}\hspace*{0.1cm}\mathrm{singularities}}}=\underbrace{\mathcal{C}[N,\lambda,\overline{\lambda}]}_{\text{remaining bulk integrations \& factors of $\lambda\overline{\lambda}$ and $\delta^{N}(\mathds{1})$}}\cdot\underbrace{\psi(\{g_{i}=\mathds{1}\}_{i=1,2,3})}_{\text{spin network evaluation}}
\end{align}

In other words, the transition amplitude for any boundary state living on the spherical boundary graph $\gamma$ factorizes into a sum entirely given by the combinatorics of the boundary spin network state \textit{regardless} of the bulk topology. The prefactor is in general infinite, since we sum over an infinite number of graphs. However, note that the prefactor is a priori independent of the boundary state and can always be reabsorbed in the normalization chosen for the path integral.

\subsection{The General Case of a Spherical Boundary}
After having discussed the simplest possible spherical boundary graph, let us now show that a similar factorization theorem can be obtained for a generic boundary graph representing the $2$-sphere. In order to illustrate the main proof strategy, let us first look into the next-to-simplest example, the so-called ``\textit{pillow graph}'' (see figure \ref{fig:PillowGraph}) --in the following denoted by $\gamma\in\overline{\mathfrak{G}}_{2}$.

\begin{figure}[H]
    \centering
    \includegraphics[scale=1.2]{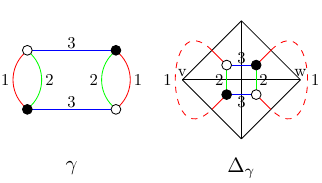}
    \caption{The pillow graph $\gamma$ and its simplicial complex $\Delta_{\gamma}$, where the gluing of edges is indicated by the dotted edges.\label{fig:PillowGraph}}
\end{figure}

As usual, we label the boundary edges by group elements, as shown on the l.h.s.~in figure \ref{NotationSphere}. Let us now consider a generic open $(3+1)$-coloured graph $\mathcal{G}\in\mathfrak{G}_{3}$ with boundary given by $\gamma = \partial\mathcal{G}$ and label some of its edges and paths as shown on the r.h.s.~of figure \ref{NotationSphere}.

\begin{figure}[H]
    \centering
    \includegraphics[scale=1.2]{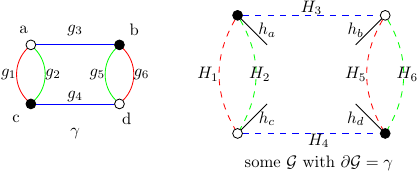}
    \caption{The pillow graph $\gamma$ equipped with group elements, as well as the general structure of an open $(3+1)$-coloured graph $\mathcal{G}$ with $\partial\mathcal{G}=\gamma$, equipped with group elements on some of its edges and paths. The dotted lines represent the non-cyclic faces, which lead to the corresponding boundary edges.\label{NotationSphere}}
\end{figure}

The notation is essentially the same as in the simple example above. We label the vertices of the boundary graph by Latin letters. Since to every vertex on the boundary graph there is a corresponding external leg in the open graph, we label these external legs by the same letters and denote the group element assigned to these edges by $h_{i}$ for $i=a,b,c,d$. Furthermore, recall that for every edge of colour $i=1,2,3$ on the boundary graph, there is a corresponding non-cyclic face of colour $0i$ in the open graph. The corresponding bicoloured paths enclosed by the two corresponding external legs are drawn as dotted lines in the figure above and are equipped with the product of all the group elements assigned to its edges, which we denote by $H_{i}$. Last but not least, we equip the edges of the boundary graph by group elements $g_{i}$ as usual. Let us stress again that the paths equipped with $H_{i}$ are not necessarily independent, as there could be an edge of colour $0$ appearing in more than one of these parts. In other words, there could be a group element appearing in more than one of the products $H_{i}$. Now, let us observe the following:

\begin{Lemma}
    \label{Lemma1}
    If $\mathcal{G}$ is an open $(3+1)$-coloured graph with $\partial\mathcal{G}=\gamma$ representing a manifold or pseudomanifold without singularities touching the boundary, then
    \begin{align*}
        &\text{(a): }H_{1}H_{2}^{-1}\stackrel{\mathrm{CF}}{=}\mathds{1}\\
        &\text{(b): }H_{5}H_{6}^{-1}\stackrel{\mathrm{CF}}{=}\mathds{1}\\
        &\text{(c): }H_{1}H_{4}^{-1}H_{5}H_{3}^{-1}\stackrel{\mathrm{CF}}{=}\mathds{1}\\
        &\text{(d): }H_{2}H_{4}^{-1}H_{6}H_{3}^{-1}\stackrel{\mathrm{CF}}{=}\mathds{1},
    \end{align*}
    where ``$\stackrel{\mathrm{CF}}{=}$'' means using all the delta functions corresponding to the closed faces of $\mathcal{G}$.
\end{Lemma}

\begin{proof}
    This is a consequence of Lemma \ref{LemmaSphere}. If $\mathcal{G}$ represents a manifold or a pseudomanifold without boundary singularities, then all the open $3$-bubbles of $\mathcal{G}$ represent disks (=$2$-balls). For example, all the dotted green and blue lines form a $3$-bubble of colour $023$ representing the $2$-ball with exactly one boundary component, namely the face of the boundary graph with colour $23$. Now, the closed path $H_{2}H_{4}^{-1}H_{6}H_{3}^{-1}H_{6}$ is totally contained in the $3$-bubble of colour $023$, and hence can be replaced by $\mathds{1}$, using the relations encoded in all the closed faces of $\mathcal{G}$.
\end{proof}

Using this lemma, we are now able to prove the following general result.

\begin{Proposition}
    If $\mathcal{G}$ is an open $(3+1)$-coloured graph with $\partial\mathcal{G}=\gamma$ representing a manifold, then its amplitude satisfies
    \begin{align*}
        \mathcal{A}^{\lambda}_{\mathcal{G}}[\{g_{e}\}_{e\in\mathcal{E}_{\gamma}}]\propto\delta^{N}(g_{1}g_{2}^{-1})\delta^{N}(g_{5}g_{6}^{-1})\delta^{N}(g_{1}g_{3}^{-1}g_{5}g_{4}^{-1}).
    \end{align*}
    The same holds true for pseudomanifolds, for which all the singularities are in the bulk. As a consequence, we get that
    \begin{align*}
        \langle\mathcal{Z}_{\mathrm{cBM}}\vert\Psi\rangle\vert_{\substack{\mathrm{manifolds+pseudomanifolds}\\\mathrm{without}\hspace*{0.1cm}\mathrm{boundary}\hspace*{0.1cm}\mathrm{singularities}}}=\mathcal{C}[N,\lambda,\overline{\lambda}]\cdot\psi(\{g_{i}=\mathds{1}\}_{i=1,2,3,4}).
    \end{align*}
\end{Proposition}

\begin{proof}
    Let us write down the contribution to the amplitude from all the faces involving boundary group elements, using the general notation introduced in figure \ref{NotationSphere}:
    \begin{equation}
        \begin{aligned}
            \int_{\mathrm{SU}(2)^{4}}\,&\mathrm{d}h_{a}\mathrm{d}h_{b}\mathrm{d}h_{c}\mathrm{d}h_{d}\,\\&\delta^{N}(g_{1}h_{a}^{-1}H_{1}h_{c}^{-1})\delta^{N}(g_{2}h_{a}^{-1}H_{2}h_{c}^{-1})\delta^{N}(g_{3}h_{a}^{-1}H_{3}h_{b}^{-1})\\&\delta^{N}(g_{4}h_{d}^{-1}H_{4}h_{c}^{-1})\delta^{N}(g_{5}h_{d}^{-1}H_{5}h_{b}^{-1})\delta^{N}(g_{6}h_{d}^{-1}H_{6}h_{b}^{-1})
        \end{aligned}
    \end{equation}
    As a next step, let us choose a maximal tree $\mathcal{T}$ in the boundary graph $\gamma$, i.e.~a subgraph containing all the vertices of $\gamma$, which does not contain cycles. A possible choice is drawn in the figure \ref{fig:Tree} below.
    \begin{figure}[H]
        \centering
        \includegraphics[scale=1.2]{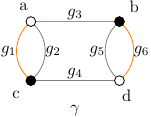}
        \caption{A maximal tree (in orange) within the boundary graph $\gamma$.\label{fig:Tree}}
    \end{figure}
    The idea is now to integrate over all the delta functions associated to the boundary edges contained in the tree using one of the group elements $h_{i}$ assigned to the external legs. Since $\mathcal{T}$ is a tree, we will end up with only one remaining integration variable and three delta functions, which correspond to the three edges, which are not contained on the tree. In this specific example, we could for example integrate over $h_{c}$ using $\delta^{N}(g_{1}h_{a}^{-1}H_{1}h_{c}^{-1})$, over $h_{a}$ using $\delta^{N}(g_{3}h_{a}^{-1}H_{3}h_{b}^{-1})$ and over $h_{b}$ using $\delta^{N}(g_{6}h_{d}^{-1}H_{6}h_{b}^{-1})$. Hence, we are left with only one remaining integration variable, namely $h_{d}$, and the contribution to the amplitude becomes
    \begin{equation}
        \begin{aligned}
            \int_{\mathrm{SU}(2)}\,&\mathrm{d}h_{d}\,\\&
            \delta^{N}(g_{6}^{-1}g_{3}g_{1}^{-1}g_{2}g_{3}^{-1}g_{6}h_{d}^{-1}H_{6}H_{3}^{-1}H_{2}H_{1}^{-1}H_{3}H_{6}^{-1}h_{d})\\&
            \delta^{N}(g_{6}^{-1}g_{3}g_{1}^{-1}g_{4}h_{d}^{-1}H_{4}H_{1}^{-1}H_{3}H_{6}^{-1}h_{d})\\&
            \delta^{N}(g_{6}^{-1}g_{5}h_{d}^{-1}H_{5}H_{6}^{-1}h_{d}).
        \end{aligned}
    \end{equation}
    Now, the point is that all the products of elements $H_{i}$ enclosed in the expression $h_{d}^{-1}\dots h_{d}$ describe closed paths in the graph drawn on the r.h.s.~figure \ref{NotationSphere}. The latter is a planar graph, as it has the same structure as the boundary graph and hence, using Lemma \ref{LemmaSphere} and Lemma \ref{Lemma1}, we can replace all of these products by the identity. Hence, we end up with
    \begin{equation}
        \begin{aligned}
            \delta^{N}(g_{1}^{-1}g_{2})\delta^{N}(g_{6}^{-1}g_{3}g_{1}^{-1}g_{4})\delta^{N}(g_{6}^{-1}g_{5})
        \end{aligned}
    \end{equation}
    as claimed. The same is of course true if $\mathcal{G}$ is a pseudomanifold with the property that all its singularities are in the interior, as in this case Lemma \ref{Lemma1} is still valid. The fact that we recover the spin network evaluation can again be shown by using the $\mathrm{SU}(2)$-invariance of the spin network function.
\end{proof}

Following the idea of the previous proof, let us now generalize the result to arbitrary boundary graphs representing the $2$-sphere.

\begin{Theorem}
    \label{GenSphereFact}
    Consider an arbitrary closed $2$-coloured graph $\gamma\in\overline{\mathfrak{G}}_{2}$ representing a $2$-sphere. If $\mathcal{G}$ is an open $(3+1)$-coloured graph with $\partial\mathcal{G}=\gamma$ representing a manifold, then its amplitude satisfies
    \begin{align*}
        \mathcal{A}^{\lambda}_{\mathcal{G}}[\{g_{e}\}_{e\in\mathcal{E}_{\gamma}}]\propto\prod_{f\in\mathcal{F}_{\gamma}}\delta^{N}\bigg (\overrightarrow{\prod_{e\in f}}g_{e}^{\varepsilon(e,f)}\bigg ),
    \end{align*}
    i.e.~we get a theory of flat boundary connections and no other constraints or mixed terms connecting bulk and boundary elements. The same holds true for pseudomanifolds, for which all the singularities are in the bulk. As a consequence, we get that
    \begin{align*}
        \langle\mathcal{Z}_{\mathrm{cBM}}\vert\Psi\rangle\vert_{\substack{\mathrm{manifolds+pseudomanifolds}\\\mathrm{without}\hspace*{0.1cm}\mathrm{boundary}\hspace*{0.1cm}\mathrm{singularities}}}=\mathcal{C}[N,\lambda,\overline{\lambda}]\cdot\psi(\{g_{e}=\mathds{1}\}_{e\in\mathcal{E}_{\gamma}}).
    \end{align*}
\end{Theorem}

\begin{proof}
    As before, we label the edges of the boundary graph $\gamma$ by group elements $\{g_{e}\}_{e\in\mathcal{E}_{\gamma}}$. Furthermore, we label the external legs of $\mathcal{G}$ by the vertices of the boundary graph and the corresponding group elements living on these edges by $\{h_{v}\}_{v\in\mathcal{V}_{\gamma}}$. We also label the non-cyclic faces of $\mathcal{G}$ by the edges of the boundary graph and the product of group elements living on the part of these faces connecting the two external legs by $\{H_{e}\}_{e\in\mathcal{E}_{\gamma}}$. Then, the contribution of all the boundary group elements to the amplitude can be written as
    \begin{equation}
        \begin{aligned}
            \int_{\mathrm{SU}(2)^{\vert\mathcal{V}_{\gamma}}\vert}\,&\bigg(\prod_{v\in\mathcal{V}_{\gamma}}\mathrm{d}h_{v}\bigg )\,\prod_{e\in\mathcal{E}_{\gamma}}\delta^{N}(g_{e}h_{t(e)}^{-1}H_{e}h_{s(e)}^{-1})
        \end{aligned}
    \end{equation}
    Let us choose a maximal tree $\mathcal{T}$ in the boundary graph $\gamma$. Integrating over all the delta functions involving a $g_{e}$, $e\in\mathcal{T}$, we are left with only one integration, which corresponds to some remaining vertex $v_{0}$. In total, there are exactly
    \begin{align}
        \vert\mathcal{E}_{\gamma}\vert-\vert\mathcal{E}_{\mathcal{T}}\vert=\vert\mathcal{E}_{\gamma}\vert-\vert\mathcal{V}_{\mathcal{\gamma}}\vert+1=\vert\mathcal{F}_{\gamma}\vert-1
    \end{align}
    delta functions left, where we have used the fact that $\gamma$ represents a $2$-sphere, i.e.~$\vert\mathcal{V}_{\gamma}\vert-\vert\mathcal{E}_{\gamma}\vert+\vert\mathcal{F}_{\gamma}\vert=2$. All of these delta functions have the following structure:
    \begin{align}
        \delta^{N}(G h_{v_{0}}^{-1}\mathcal{H} h_{v_{0}}),
    \end{align}
    where $G$ is some product of boundary group elements $\{g_{e}\}_{e\in\mathcal{E}_{\gamma}}$ and $\mathcal{H}$ is a product of elements contained in $\{H_{e}\}_{e\in\mathcal{E}_{\gamma}}$. Now, the product $\mathcal{H}$ describes some closed path on the graph $\mathcal{G}$ and using similar arguments as previously, it can be replaced by $\mathds{1}$ using all the closed faces of the graph since the dotted graph is planar and since all the closed paths in this graph can be replaces by $\mathds{1}$ as a consequence of Lemma \ref{LemmaSphere}. Hence, we are left with a product of delta functions only containing closed paths consisting of boundary group elements.
    \bigskip

    Next, let us recall the well-known fact that the Ponzano-Regge amplitude always encodes flatness of the boundary connection \cite{TorusPR4}. To see this, let us look at a generic face of the boundary graph and the general structure of some open $3$-bubble of $\mathcal{G}$ leading to this face, as sketched in figure \ref{fig:FlatBoundary} below.

    \begin{figure}[H]
        \centering
        \includegraphics[scale=1.2]{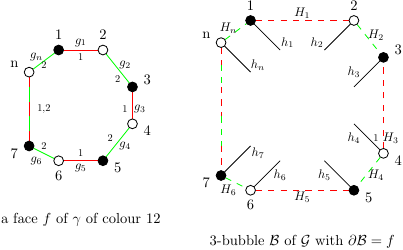}
        \caption{A face $f$ of the boundary graph $\gamma$ containing $n$ edges and the general structure of an open $3$-bubble $\mathcal{B}$ in $\mathcal{G}$ with $\partial\mathcal{B}=\gamma$.\label{fig:FlatBoundary}}
    \end{figure}

    The figure above shows some face $f$ of the boundary graph $\gamma$ of colour $12$, for definiteness, consisting of $n\in 2\mathbb{N}$ edges, which are equipped with the boundary group elements $\{g_{i}\}_{i\in\{1,\dots,n\}}$. The right-hand side shows the general structure of an open $3$-bubble $\mathcal{B}$ leading to this boundary face, i.e. $\partial\mathcal{B}=f$. The bicoloured paths are equipped with the product of group elements $\{H_{i}\}_{i\in\{1,\dots,n\}}$ and the external legs with group elements $\{h_{i}\}_{i\in\{1,\dots,n\}}$, as usual. To an edge $i\in\{1,\dots,n\}$ of $f$, there is a corresponding delta function $\delta^{N}(g_{i}h_{i}^{-1}H_{i}h_{i+1}^{-1})$ for even $i$ and $\delta^{N}(g_{i}h_{i+1}^{-1}H_{i}h_{i}^{-1})$ for odd $i$ in the amplitude of $\mathcal{G}$, where we use cyclic indices, i.e.~$i+n:=i$. Using the delta functions successively, we see that they encode the constraint:
    \begin{align}
        \mathds{1}=g_{1}h_{2}^{-1}H_{1}h_{1}^{-1}=g_{1}g_{2}^{-1}h_{3}H_{2}^{-1}H_{1}h_{1}^{-1}=\dots=g_{1}g_{2}^{-1}g_{3}\dots g_{n}^{-1}h_{1}H_{n}^{-1}H_{n-1}\dots H_{1}h_{1}^{-1}
    \end{align}
    Now, the relation $H_{n}^{-1}H_{n-1}\dots H_{1}$ can again be replaced by $\mathds{1}$, using the fact that $\mathcal{B}$ represents a disk as well as Lemma \ref{LemmaSphere}. Hence, we are left with
    \begin{align}
        \mathds{1}=g_{1}g_{2}^{-1}g_{3}\dots g_{n}^{-1},
    \end{align}
    which exactly tells us that the connection of the boundary face $f$ is flat. Applying the same logic for all faces of the boundary graph, we get the claim, i.e.~that the Ponzano-Regge amplitude always encodes flatness of the boundary connection. Note again that $\mathcal{G}$ does not need the represent a manifold for this argument to work, since we only need to require that all the open $3$-bubbles of $\mathcal{G}$ represent $2$-balls, or in other words, that $\mathcal{G}$ represents a pseudomanifold without singularities on the boundary. Furthermore, this fact applies of course also to arbitrary boundary topologies.
    \bigskip

    Since the Ponzano-Regge amplitude of some manifold with boundary always recovers the flatness of the boundary, the product of delta functions only containing boundary group elements can be rewritten in such a way that they contain the flatness condition for the boundary connection. However, there cannot be something more. Any additional constraint corresponds to some closed $3$-coloured path on the boundary graph and, by Lemma \ref{LemmaSphere}, it can be replaced by $\mathds{1}$, using all the other delta functions corresponding to the boundary faces. This also matches the fact that there are $\vert\mathcal{F}_{\gamma}\vert-1$ delta functions left, since one of the delta functions in the product over boundary faces is redundant. To sum up, the amplitude of $\mathcal{G}$ is proportional to a product of delta functions encoding flatness of the boundary connections and it is a well-known fact that the Ponzano-Regge transition amplitude in the case of a flat, spherical boundary is proportional to the spin network evaluation, i.e.~see\cite{TorusPR4}.
\end{proof}

Before moving on, it is important to point out that the result is not as trivial as it might appear, in light of the fact that we are describing here a topological field theory. For any manifold with boundary, the implication of the topological nature of the model is that the amplitude encodes the flatness of the boundary. Now, in principle, the Ponzano-Regge partition function depends on both the boundary data \textit{and} the topology of the bulk. What we have shown here in the context of the Boulatov model is that there are virtually no contributions of the bulk topology to the Boulatov partition function. They all collapse to a normalization factor.
This is {\it not} a consequence of the topological nature of the theory, but of the simple topology of the chosen boundary itself. To illustrate this fact, in the next section we quickly discuss the next-to-trivial boundary topology, i.e. the $2$-torus. In that case, we show that we obviously recover the flatness of the boundary, but the Boulatov amplitude will also have terms explicitly depending on the associated bulk topology.

\subsection{Toroidal Boundary}\label{SubSec:Torus}

The smallest closed $(2+1)$-coloured graph representing the next-to-trivial topology, the $2$-torus $T^{2}=S^{1}\times S^{1}$, which we denote in the following by $\gamma$, has six vertices and can be seen in figure \ref{TorusBoundaryGraph} below, together with its corresponding simplicial complex.

\begin{figure}[H]
    \centering
    \includegraphics[scale=1]{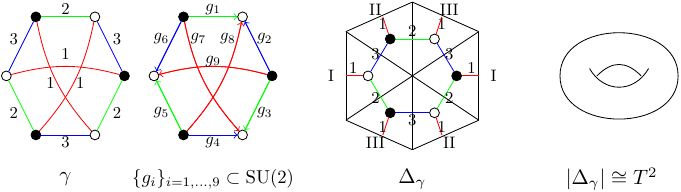}
    \caption{The smallest closed $(2+1)$-coloured graph representing the $2$-torus (l.h.s.) together with a labelling of its edges with group elements and its corresponding simplicial complex --where the gluing of edges is as indicated by the Roman numbers-- as well as its geometric realization, i.e.~the $2$-torus (r.h.s.).\label{TorusBoundaryGraph}}
\end{figure}

The smallest possible open $(3+1)$-coloured graph with boundary given by $\gamma$ --called the smallest matching graph (see Subsection \ref{subsec:SMG})-- is the graph obtained by adding an external leg to all the vertices of $\gamma$, see figure \ref{SolidTorusExamples}. However, while this graph -denoted by $\mathcal{G}_{0}$-- is a core graph, it clearly represents a pseudomanifold: its $123$-bubble is exactly given by $\gamma$ and so is non-spherical. More precisely, the pseudomanifold dual to $\mathcal{G}_{0}$ is homeomorphic to the topological cone of $T^{2}$. In figure \ref{SolidTorusExamples} below, we also represent two more complicated graphs $\mathcal{G}_{1},\mathcal{G}_{1}^{\prime}\in\mathfrak{G}_{3}$ with boundary $\gamma$, which are also both core graphs but represent a manifold (in fact the solid torus $\overline{T}^{2}$, see Appendix \ref{AppendixC:SolidTorusGraphs}).
\begin{figure}[H]
    \centering
    \includegraphics[scale=1]{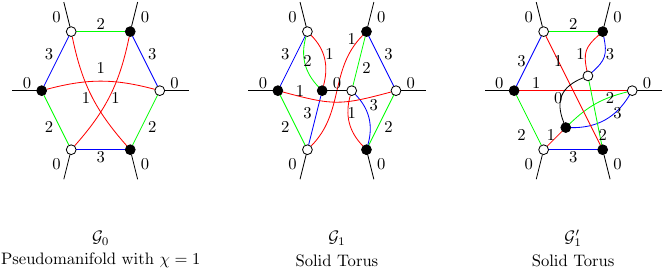}
    \caption{Three core graphs $\mathcal{G}_{0},\mathcal{G}_{1},\mathcal{G}_{1}^{\prime}\in\mathfrak{G}_{3}$ with $\partial\mathcal{G}_{0,1,2}=\gamma$. The graph $\mathcal{G}_{0}$ represents a pseudomanifold with one point-like singularity in the bulk and $\chi=1$. The graphs $\mathcal{G}_{1}$ and $\mathcal{G}_{1}^{\prime}$ represent the solid torus.\label{SolidTorusExamples}}
\end{figure}

A straightforward calculation shows that the graphs $\mathcal{G}_{1}$ and $\mathcal{G}_{1}^{\prime}$ have the same amplitude\footnote{This can also be seen by Lemma \ref{AmplChangeInt} and the fact that they are related by two internal proper $1$-dipole moves (see figure \ref{SolidTorus3}).} given by
\begin{align}
    \mathcal{A}_{\mathcal{G}_{1}}^{\lambda}[\{g_{1},\dots,g_{9}\}]=\bigg(\frac{\lambda\overline{\lambda}}{\delta^{N}(\mathds{1})}\bigg)^{4}\delta^{N}(g_{1}g_{2}^{-1}g_{3}g_{7}^{-1})\times\Delta_{\mathrm{FB}}(g_{1},\dots,g_{9}),
\end{align}
where $\Delta_{\mathrm{FB}}(g_{1},\dots,g_{9})$ is an abbreviation for the expression
\begin{align}
    \Delta_{\mathrm{FB}}(g_{1},\dots,g_{9})=\delta^{N}(g_{1}g_{2}^{-1}g_{3}g_{4}^{-1}g_{5}g_{6}^{-1})\delta^{N}(g_{1}g_{8}^{-1}g_{5}g_{9}^{-1}g_{3}g_{7}^{-1})\delta^{N}(g_{2}g_{8}^{-1}g_{4}g_{7}^{-1}g_{6}g_{9}^{-1})
\end{align}
encoding the flatness of the boundary. Since they clearly also have the same boundaries, topologies and orders, they are are contained in the same core equivalence class. Note that one of the three delta functions in $\Delta_{\mathrm{FB}}$ is actually redundant since we can always replace one of the faces by $\delta^{N}(\mathds{1})$ using the other two faces and the constraint $\delta^{N}(g_{1}g_{2}^{-1}g_{3}g_{7}^{-1})$. Hence, the degree of divergence of these graphs is actually $\delta^{N}(\mathds{1})^{-3}$. The geometric interpretation of this result is the following: The term $\Delta_{\mathrm{FB}}(g_{1},\dots,g_{9})$ encodes flatness of the boundary, which we always recover for the Ponzano-Regge model of some manifold (see the proof of Theorem \ref{GenSphereFact}), and the additional constraint $\delta^{N}(g_{1}g_{2}^{-1}g_{3}g_{7}^{-1})$ tells us which cycles of the boundary graph becomes contractible through the bulk, as sketched in figure \ref{TorusFact4}.

\begin{figure}[H]
    \centering
    \hspace*{2cm}\includegraphics[scale=1.2]{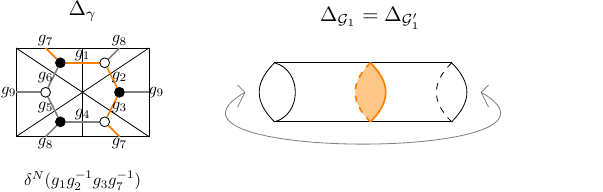}
    \caption{The boundary complex $\Delta_{\gamma}$ and the solid torus dual to $\mathcal{G}_{1},\mathcal{G}_{1}^{\prime}$. When gluing a solid torus to the torus boundary, there are two choices of which of the two non-contractible cycles of $\gamma$ becomes contractible through the bulk. This information is specified by the constraint $g_{1}g_{2}^{-1}g_{3}g_{7}^{-1}=\mathds{1}$ contained in the amplitude of $\mathcal{G}_{1,2}$.\label{TorusFact4}}
\end{figure}

Now, this is not the end of the story. At the continuum level, there are two a priori boundary cycles of $T^2$ that can become contractible through the bulk when gluing it to the solid torus ($\pi_{1}(T^{2})\cong\mathbb{Z}^{2}$). As sketched in figure \ref{TorusFact4}, our initial choice of bulk makes the ``vertical'' direction contractible. It is expected that it should exist a choice of bulk such that the ``horizontal'' direction is contractible instead. Following the notation of figure \ref{TorusFact4}, it should take the form of the constraint $\delta^N(g_1 g_2^{-1} g_9 g_6^{-1})$. As a matter of fact, it is indeed possible to find such an admissible bulk (see $\mathcal{G}_2$ below). In both cases, the constraint takes the form of a closed $3$-coloured paths on the boundary and one might wonder if more choices are possible. A list of all possible $3$-coloured cycles on the boundary graph $\gamma$ --up to flatness of the boundary and which do not go twice to the same edge-- are drawn in the figure \ref{TorusCases} below.

\begin{figure}[H]
    \centering
    \includegraphics[scale=1]{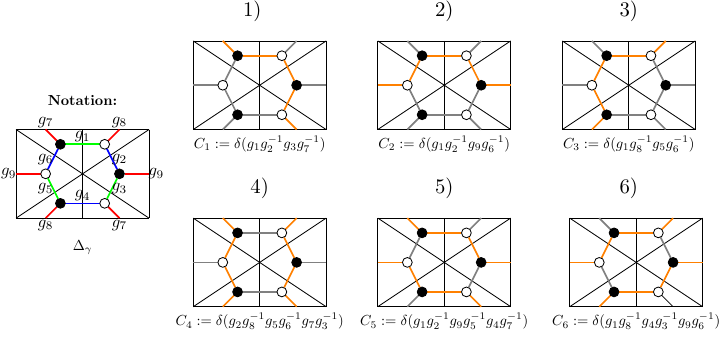}
    \caption{A list of all independent (by flatness) possible closed $4$-colour paths on the boundary, which to not go twice to some given edge and their associated boundary constraint.\label{TorusCases}}
\end{figure}

The solid torus graphs $\mathcal{G}_{1}$ and $\mathcal{G}_{1}^{\prime}$ are examples of graph encoding the constraint 1). It is not too hard to construct examples of graphs representing manifolds having the other constraints encoded in their amplitudes. Examples can be seen in figure \ref{SolidTorusCases} below.

\begin{figure}[H]
    \centering
    \includegraphics[scale=1]{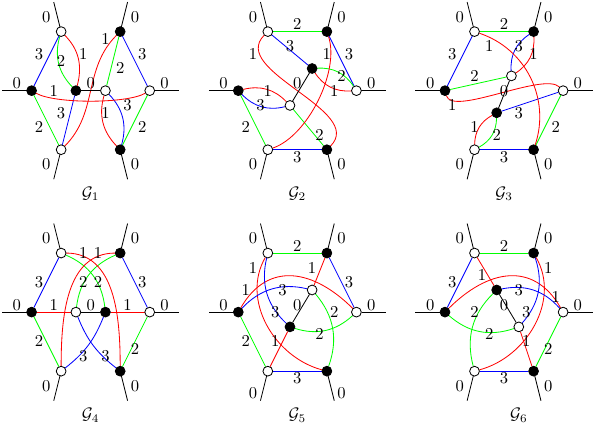}
    \caption{Six core graphs $\mathcal{G}_{i}$ with $\partial\mathcal{G}_{i}=\gamma$ representing manifolds, whose amplitudes encode the constraints $C_{i}$ sketched in figure \ref{TorusCases} above. \label{SolidTorusCases}}
\end{figure}

The graphs $\mathcal{G}_{2,3}$ clearly represent the solid tori, as they are essentially the graphs obtained by rotating $\mathcal{G}_{1}$ and by interchanging some of its colours. In other words, $\mathcal{G}_{2,3}$ are isomorphic to $\mathcal{G}_{1}$ as coloured graphs and hence represent the solid torus too.\footnote{Two coloured graphs $\mathcal{G}_{1,2}\in\mathfrak{G}_{d}$ with colouring maps $\varphi_{1,2}:\mathcal{E}_{\mathcal{G}_{1,2}}\to\mathcal{C}_{d}$ are called ``\textit{isomorphic}'', if they are isomorphic as graphs and if there colours are related by a bijective recolouring of their edges, i.e.~there is a graph isomorphism $\Phi:\mathcal{V}_{\mathcal{G}_{1}}\to\mathcal{V}_{\mathcal{G}_{2}}$ as well as a bijection $\Psi:\mathcal{C}_{d}\to\mathcal{C}_{d}$ such that $\varphi_{1}\circ\Phi=\Psi\circ\varphi_{2}$. By definition, two isomorphic coloured graphs are isomorphic if and only if they represent simplicial isomorphic complexes. \cite{GagliardiBoundaryGraph}} The three graphs $\mathcal{G}_{4,5,6}$ have been constructed by trial and error and it is a priori not clear which topology they represent. However, they clearly represent manifolds with boundary, as all of their $3$-bubbles represent spheres or balls, as one can easily check.\footnote{A closer analysis reveals that the graphs $\mathcal{G}_{4,5,6}$ indeed represent solid tori too, e.g.~by using Theorem 14 of \cite{HandleBodiesCT}, which tells us that every core graph representing a manifold with boundary $\gamma$ and with strictly less than 14 internal vertices represents the solid torus.} Furthermore, the three graphs also clearly represent the same topology, as they are again isomorphic as coloured graphs.

The amplitudes of the graph $\mathcal{G}_{i}$ is given by
\begin{align}
    &\mathcal{A}_{\mathcal{G}_{i}}^{\lambda}[\{g_{1},\dots,g_{9}\}]=\bigg(\frac{\lambda\overline{\lambda}}{\delta^{N}(\mathds{1})}\bigg)^{4}C_{i}\times\Delta_{\mathrm{FB}}(g_{1},\dots,g_{9})
\end{align}
for all $i\in\{1,\dots,6\}$, where $C_{i}$ are the constraints defined in figure \ref{TorusCases}. Since $C_{i}$ cannot be related using the flatness of the boundary, the six amplitudes are in principle different.
\bigskip

To sum up, we see that in the case of a toroidal boundary, there are different contributions to the full transition amplitude of the coloured Boulatov model. These contributions differ by the choice of which cycle becomes contractible through the bulk. That is, they differ by how the bulk is glued to the boundary and by the topology of the bulk. We have only discussed six explicit examples, but there might be many more cases, e.g. by combining the six cases, which geometrically correspond to different winding numbers, cycles and combinations thereof. A complete analysis of the torus boundary topology is left for future work \cite{ToAppearTorus}.
\section{Leading Order Contribution to a Spherical Boundary}
\label{SectionV:LeadingOrder}

In this last section, we show that the leading order contribution of the transition amplitude of some spherical boundary graph, when restricted to manifolds, is given by the equivalence class representing the closed $3$-ball defined by the smallest possible open graph matching the given boundary graph. Furthermore, we show that these graphs essentially generalize the melonic diagrams, which are the leading order diagrams in the large $N$ limit of the free energy of the Boulatov model \cite{GurauLargeN1,GurauLargeN2,GurauLargeN3}, in the sense that they have the smallest possible \textit{Gurau degree} \cite{GurauColouredTensorModelsReview}.

\subsection{Smallest Matching Graphs and Gurau Degree}\label{subsec:SMG}
Let us first introduce a certain class of graphs which are the smallest possible open $(3+1)$-coloured graphs matching some given boundary graph. As we will discuss in this section, this type of graphs can be viewed as a generalization of melonic graphs in the sense that they minimize a suitable generalization of the Gurau degree to open graphs. Let us introduce the following terminology.

\begin{Definition}[Smallest Matching Graph]
    Consider a closed $(2+1)$-coloured graph $\gamma\in\overline{\mathfrak{G}}_{2}$. We define an open $(3+1)$-coloured graph $\mathcal{G}_{\mathrm{SMG}}\in\mathfrak{G}_{3}$ with $\partial\mathcal{G}=\gamma$, called the ``smallest matching graph'', by adding an external leg of colour $0$ to all the vertices of $\gamma$, after interchanging the type of vertices (black $\leftrightarrow$ white) within $\gamma$.
\end{Definition}

Figure \ref{fig:SMG} below shows three examples of boundary graphs and their corresponding smallest matching graphs.

\vspace*{-0.1cm}
\begin{figure}[H]
    \centering
    \includegraphics[scale=0.9]{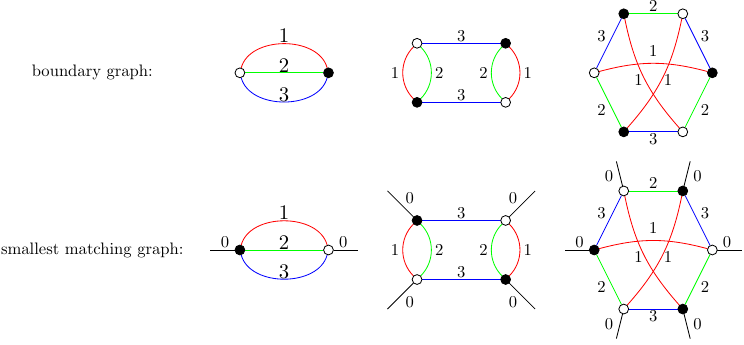}
    \caption{Three examples of closed $(2+1)$-coloured graphs and their smallest matching graphs.\label{fig:SMG}}
\end{figure}
\vspace*{-0.1cm}

The first two examples represent $2$-spheres and the corresponding smallest matching graphs clearly represent closed $3$-balls. The third example represents the $2$-torus and the corresponding smallest matching graph represents a pseudomanifold with one internal point-like singularity. The smallest matching graph corresponding to some boundary graph $\gamma$ is clearly a core graph since it only has one internal $3$-bubble, which by construction is equivalent to $\gamma$. Furthermore, it is also clear that the smallest matching graph is the unique open $(3+1)$-coloured graph with boundary $\gamma$ with minimal possible number of internal vertices, $\vert\mathcal{V}_{\gamma}\vert$. Last but not least, note that the simplicial complex $\Delta_{\mathcal{G}_{\mathrm{SMG}}}$ is precisely what is usually called the ``\textit{cone}'' of the simplicial complex $\Delta_{\gamma}$, i.e.~the simplicial complex obtained by adding to $\Delta_{\gamma}$ a vertex $v_{0}$ as well as a $(k+1)$-simplex $\{v_{0}\}\cup\sigma$ for each $k$-simplex $\sigma$ of $\Delta_{\gamma}$. Therefore, the pseudomanifold $\vert\Delta_{\mathcal{G}_{\mathrm{SMG}}}\vert$ is the topological cone of the surface $\mathcal{S}:=\vert\Delta_{\gamma}\vert$, usually denoted by $C\mathcal{S}$, e.g.~see \cite{Hatcher}. \bigskip

If $\gamma$ does not represent a $2$-sphere, then the smallest matching graph cannot be a manifold since its unique $3$-bubble of colour $123$ is equivalent to $\gamma$ and is by assumption non-spherical. In fact, the Euler characteristic of the simplicial complex dual to the smallest matching graph is generically given by one, independently of the boundary topology, as the following short calculation shows:
\begin{align}
    \chi(\Delta_{\mathcal{G}_{\mathrm{SMG}}})=\underbrace{\mathcal{B}^{[3]}}_{=1+\vert\mathcal{F}_{\gamma}\vert}-\underbrace{\vert\mathcal{F}_{\mathcal{G}_{\mathrm{SMG}}}\vert}_{=\vert\mathcal{F}_{\gamma}\vert+\vert\mathcal{E}_{\gamma}\vert}+\underbrace{\vert\mathcal{E}_{\mathcal{G}_{\mathrm{SMG}}}\vert}_{=\vert\mathcal{E}_{\gamma}\vert+\vert\mathcal{V}_{\gamma}\vert}-\underbrace{\vert\mathcal{V}_{\mathcal{G}_{\mathrm{SMG}},\mathrm{int}}\vert}_{=\vert\mathcal{V}_{\gamma}\vert}=1,
\end{align}
where $\mathcal{B}^{[3]}$ denotes the number of $3$-bubbles of $\mathcal{G}_{\mathrm{SMG}}$. Since the Euler characteristic of any odd-dimensional compact and orientable manifold $\mathcal{M}$ has to fulfil $\chi(\mathcal{M})=\frac{1}{2}\chi(\partial\mathcal{M})$, we see that
\begin{align}
    1-g_{\gamma}=\frac{1}{2}\chi(\Delta_{\gamma})\stackrel{!}{=}\chi(\Delta_{\mathcal{G}_{\mathrm{SMG}}})=1,
\end{align}
where $g_{\gamma}$ denotes the genus of the surface dual to $\gamma$,
can only be fulfilled for boundary graphs representing the $2$-sphere, i.e.~the genus $0$-surface. Indeed, if $\gamma$ is a $2$-sphere, then one can easily check that the smallest matching graph always represents a manifold, since all its $3$-bubbles are either $2$-spheres or disks (=$2$-balls). Furthermore, a closer look reveals that they generically represent the $3$-ball in this case:

\begin{Proposition}
    The smallest matching graph $\mathcal{G}_{\mathrm{SMG}}$ of an arbitrary closed $(2+1)$-coloured graph $\gamma\in\overline{\mathfrak{G}}_{2}$ representing the $2$-sphere represents the $3$-ball.
\end{Proposition}

\begin{proof}
    As explained above, the smallest matching graph represents the cone over the boundary surface. Therefore, if $\gamma$ represents a $2$-sphere, than $\mathcal{G}_{\mathrm{SMG}}$ clearly represents a $3$-ball.
\end{proof}

We now show that the smallest matching graph of some spherical boundary graph generalizes melonic graphs in the sense that they have the smallest possible \textit{Gurau degree}. In the case of closed graphs, this combinatorial quantity labels the large $N$ expansion of the free energy of coloured tensor models. Furthermore, it allows to give bounds on the amplitudes of coloured GFTs \cite{GurauLargeN3,GurauColouredTensorModelsReview}. For coloured graphs with non-empty boundary, the Gurau degree can be defined following \cite{DegreeOpenGraphs,SenseTM}.

\begin{Definition}[Gurau Degree]
    Let $\mathcal{G}\in\mathfrak{G}_{d}$ be a connected open $(d+1)$-coloured graph with boundary graph $\gamma:=\partial\mathcal{G}$. Then, the ``(Gurau) degree of $\mathcal{G}$'' is defined to be the rational number
    \begin{align*}
        \omega(\mathcal{G}):=\frac{(d-1)!}{2}\bigg (\frac{d(d-1)}{4}\vert\mathcal{V}_{\mathcal{G},\mathrm{int}}\vert+d-\vert\mathcal{F}_{\mathcal{G},\mathrm{int}}\vert\bigg )-\frac{(d-1)!}{2}\bigg (\frac{d-1}{2}\vert\mathcal{V}_{\gamma}\vert+C(\gamma)\bigg )\in\frac{(d-1)!}{2}\cdot\mathbb{Z},
    \end{align*}
    where $C(\gamma)$ denotes the number of connected components of the boundary graph $\gamma$.
\end{Definition}

\begin{Remarks}
    \begin{itemize}
        \item[]
        \item[(a)]In the two-dimensional case, the degree is equivalent to the genus of the surface dual to the graph. To see this, let us write the degree for some $(2+1)$-coloured graph $\mathcal{G}$ with boundary $\partial\mathcal{G}=\gamma$ as
        \begin{equation}
            \begin{aligned}
                \omega(\mathcal{G})=& 1-\frac{1}{2}\underbrace{(\vert\mathcal{F}_{\mathcal{G},\mathrm{int}}\vert+\vert\mathcal{E}_{\gamma}\vert-\vert\mathcal{E}_{\mathcal{G}}\vert+\vert\mathcal{V}_{\mathcal{G},\mathrm{int}}\vert)}_{=\chi(\Delta_{\mathcal{G}})}-\frac{C(\partial\mathcal{G})}{2}+\frac{1}{2}\bigg (\vert\mathcal{E}_{\gamma}\vert-\vert\mathcal{E}_{\mathcal{G}}\vert-\frac{1}{2}\vert\mathcal{V}_{\gamma}\vert+\frac{3}{2}\vert\mathcal{V}_{\mathcal{G},\mathrm{int}}\vert\bigg ).
            \end{aligned}
        \end{equation}
        Let us analyse the last term. First of all, since $\gamma$ is a closed $(1+1)$-coloured graph, we clearly have that $\vert\mathcal{E}_{\gamma}\vert=\vert\mathcal{V}_{\gamma}\vert$. Furthermore, we have that $3\vert\mathcal{V}_{\mathcal{G},\mathrm{int}}\vert+\vert\mathcal{V}_{\gamma}\vert=2\vert\mathcal{E}_{\mathcal{G}}\vert$ for every open $(2+1)$-coloured graph. To sum up, the last term vanishes and we get
        \begin{equation}
            \begin{aligned}
                \omega(\mathcal{G})=& 1-\frac{\chi(\Delta_{\mathcal{G}})+C(\partial\mathcal{G})}{2}=g_{\mathcal{G}},
            \end{aligned}
        \end{equation}
        where we used that the genus of a surface with boundary is defined by $\chi(\Delta_{\mathcal{G}})=:2-2g_{\mathcal{G}}-C(\partial\mathcal{G})$.
        \item[(b)]Clearly, the definition reduces to the definition of the Gurau degree for closed graphs \cite{GurauLargeN3,GurauColouredTensorModelsReview} in the case of empty boundary $\partial\mathcal{G}=\emptyset$.
        \item[(c)]Similarly as for closed graph, the degree can be rewritten in terms of the genera of the ``\textit{jackets}'' of the graph. More precisely, it can be written as the sum over the genera of the ``\textit{pinched jackets}'' of the open graph minus the sum over the genera of jackets of the boundary graph, i.e.~the degree of its boundary graph. See \cite{DegreeOpenGraphs,SenseTM} for more details. Note that the definition of the degree in these papers differs by a factor of $(d-1)!/2$ compared to our definition, since they define what is usually called the ``\textit{reduced degree}'' \cite{ReducedDegree}.
    \end{itemize}
\end{Remarks}

The Gurau degree of some open coloured graph is in general \textit{not} a topological invariant, i.e.~graphs representing the same topology might have a different degree. However, it is invariant under internal proper $1$-dipole moves.

\begin{Lemma}[Dipole Contractions and Degree]
    \label{DipoleDegree}
    Let $\mathcal{G}\in\mathfrak{G}_{d}$ be a connected $(d+1)$-coloured graph and $d_{k}$ be an internal $k$-dipole, i.e.~at least one of the two $(d+1-k)$-bubbles separated by $d_{k}$ is closed. Then
    \begin{align*}
        \omega(\mathcal{G})=\frac{(d-1)!}{2}(k(d+1-k)-d)+\omega(\mathcal{G}/d_{k}).
    \end{align*}
    In particular, it follows that $\omega(\mathcal{G})=\omega(\mathcal{G}/d_{1})$, i.e.~the degree is invariant under internal $1$-dipole moves.
\end{Lemma}

\begin{proof}
    The proof is a straightforward generalization of the proof for closed graphs given in \cite{GurauLargeN3,GurauColouredTensorModelsReview}: The number of internal vertices of $\mathcal{G}$ and $\mathcal{G}/d_{k}$ is decreased by two, $\vert\mathcal{V}_{\mathcal{G}/d_{k},\mathrm{int}}\vert=\vert\mathcal{V}_{\mathcal{G},\mathrm{int}}\vert-2$, and the number of faces by
    \begin{align}
        \vert\mathcal{F}_{\mathcal{G}/d_{k}}\vert=\vert\mathcal{F}_{\mathcal{G}}\vert-\frac{k(k-1)}{2}-\frac{(d+1-k)(d-k)}{2}.
    \end{align}
    By assumption, the dipole is internal. Hence, the number of non-cyclic faces is left untouched, since $\partial\mathcal{G}=\partial(\mathcal{G}/d_{k})$ (see Proposition \ref{PropDipoleBound}). As a consequence, we have that
    \begin{align}
        \vert\mathcal{F}_{\mathcal{G}/d_{k},\mathrm{int}}\vert=\vert\mathcal{F}_{\mathcal{G},\mathrm{int}}\vert-\frac{k(k-1)}{2}-\frac{(d+1-k)(d-k)}{2}.
    \end{align}
    Using these relations and the definition of the degree, we obtain the claimed relation.
\end{proof}

This lemma is a generalization of the result stated in \cite{GurauLargeN3,GurauColouredTensorModelsReview} for closed graphs. Furthermore, as shown in \cite{GurauLargeN3,GurauColouredTensorModelsReview}, the degree of some closed coloured graph can be written in terms of the degrees of its bubbles. Let us generalize this to the case of open coloured graphs.

\begin{Proposition}
    \label{PropositionDegreeBubbles}
    Let $\mathcal{G}\in\mathfrak{G}_{3}$ be a connected open $(d+1)$-coloured graph with boundary graph $\gamma:=\partial\mathcal{G}$. Then
    \begin{align*}
        \omega(\mathcal{G})=\frac{(d-1)!}{2}\bigg (\frac{(d-1)\vert\mathcal{V}_{\mathcal{G},\mathrm{int}}\vert-\vert\mathcal{V}_{\gamma}\vert}{2(d-1)}+d-\mathcal{B}^{[d]}+\frac{\partial\mathcal{B}^{[d-1]}}{d-1}-C(\gamma)\bigg )+\sum_{i=0}^{d}\sum_{\rho}\omega(\mathcal{B}^{\hat{i}}_{(\rho)}),
    \end{align*}
    where $\mathcal{B}^{[d]}$ denotes the number of $d$-bubbles of $\mathcal{G}$ and where $\mathcal{B}^{\hat{i}}_{(\rho)}$ are the internal $d$-bubbles without colour $i$ of $\mathcal{G}$, labelled by some parameter $\rho$.
\end{Proposition}

\begin{proof}
    Let us split the set of $d$-bubbles for some colour $i$ into internal $d$-bubbles, labelled by $\rho_{\mathrm{int}}$, and open $d$-bubbles (those including external legs), labelled by $\rho_{\partial}$. Hence, we can write
    \begin{align}
        \label{eq1proof}&\sum_{i=0}^{d}\sum_{\rho}\omega(\mathcal{B}^{\hat{i}}_{(\rho)})=\sum_{i=0}^{d}\sum_{\rho_{\mathrm{int}}}\omega(\mathcal{B}^{\hat{i}}_{(\rho_{\mathrm{int}})})+\sum_{i=1}^{d}\sum_{\rho_{\partial}}\omega(\mathcal{B}^{\hat{i}}_{(\rho_{\partial})}).
    \end{align}
    Note that $\hat{0}$-bubbles cannot be open, since all the external legs have colour $0$, which is the reason why the second sum starts at $i=1$. As the $d$-bubbles are by themselves $d$-coloured graphs, we can define their degrees. The internal $d$-bubbles are closed $d$-coloured graphs and their degrees are given by
    \begin{align}
        \omega(\mathcal{B}^{\hat{i}}_{(\rho_{\mathrm{int}})})=\frac{(d-2)!}{2}\bigg (\frac{(d-1)(d-2)}{4}\vert\mathcal{V}_{\mathcal{B}^{\hat{i}}_{(\rho_{\mathrm{int}})}}\vert+(d-1)-\vert\mathcal{F}_{\mathcal{B}^{\hat{i}}_{(\rho_{\mathrm{int}})}}\vert\bigg )
    \end{align}
    and open $d$-bubbles are open $d$-coloured graphs and hence, their degrees are given by
    \begin{equation}
        \begin{aligned}
            \omega(\mathcal{B}^{\hat{i}}_{(\rho_{\partial})})=\frac{(d-2)!}{2}\bigg (\frac{(d-1)(d-2)}{4}\vert\mathcal{V}_{\mathcal{B}^{\hat{i}}_{(\rho_{\partial})},\mathrm{int}}\vert&+(d-1)-\vert\mathcal{F}_{\mathcal{B}^{\hat{i}}_{(\rho_{\partial})},\mathrm{int}}\vert\bigg )\\&-\frac{(d-2)!}{2}\bigg (\frac{d-2}{2}\vert\mathcal{V}_{\partial\mathcal{B}^{\hat{i}}_{(\rho_{\partial})}}\vert+C(\partial\mathcal{B}^{\hat{i}}_{(\rho_{\partial})})\bigg ).
        \end{aligned}
    \end{equation}
    Now, an internal vertex of $\mathcal{G}$ is by definition $(d+1)$-valent such that all adjacent edges have different colours. As a consequence, every internal vertex of $\mathcal{G}$ appears precisely in $\binom{d+1}{d}=(d+1)$ of its open or closed $d$-bubbles. This means that
    \begin{align}
        \sum_{i=1}^{d}\sum_{\rho_{\partial}}\vert\mathcal{V}_{\mathcal{B}^{\hat{i}}_{(\rho)},\mathrm{int}}\vert+\sum_{i=0}^{d}\sum_{\rho_{\mathrm{int}}}\vert\mathcal{V}_{\mathring{\mathcal{B}}^{\hat{i}}_{(\rho_{\mathrm{int}})}}\vert=(d+1)\vert\mathcal{V}_{\mathcal{G},\mathrm{int}}\vert.
    \end{align}
    A $1$-valent boundary vertex of $\mathcal{G}$ appears exactly in $\binom{d}{d-1}=d$ of its $d$-bubbles $\mathcal{B}^{\hat{i}}_{(\rho_{\partial})}$, i.e.
    \begin{align}
        \sum_{i=1}^{d}\sum_{\rho_{\partial}}\vert\mathcal{V}_{\partial\mathcal{B}^{\hat{i}}_{(\rho_{\partial})}}\vert=d\vert\mathcal{V}_{\gamma}\vert.
    \end{align}
    Furthermore, an internal face of $\mathcal{G}$ is a bicoloured path with colours $i,j$. Such a path is part of exactly $\binom{d-1}{d-1}=(d-1)$ open or closed $d$-bubbles of $\mathcal{G}$. This means that
    \begin{align}
        \sum_{i=1}^{d}\sum_{\rho_{\partial}}\vert\mathcal{F}_{\mathcal{B}^{\hat{i}}_{(\rho_{\partial})},\mathrm{int}}\vert+\sum_{i=0}^{d}\sum_{\rho_{\mathrm{int}}}\vert\mathcal{F}_{\mathcal{B}^{\hat{i}}_{(\rho_{\mathrm{int}})}}\vert=(d-1)\vert\mathcal{F}_{\mathcal{G},\mathrm{int}}\vert.
    \end{align}
    Lastly, we have to discuss the number of boundary components of the $\hat{i}$-bubbles. Every boundary component of a $d$-bubble corresponds to a $(d-1)$-bubble of the boundary graph. In other words, we have that
    \begin{align}
        \sum_{i=1}^{d}\sum_{\rho_{\partial}}C(\partial\mathcal{B}^{\hat{i}}_{(\rho_{\partial})})=\partial\mathcal{B}^{[d-1]},
    \end{align}
    where $\partial\mathcal{B}^{[d-1]}$ denotes the number of $(d-1)$-bubbles of the boundary graph $\partial\mathcal{G}$. Plugging all these relations into Equation \eqref{eq1proof} yields the required result
    \begin{align}
        \omega(\mathcal{G})=\frac{(d-1)!}{2}\bigg (\frac{(d-1)\vert\mathcal{V}_{\mathcal{G},\mathrm{int}}\vert-\vert\mathcal{V}_{\gamma}\vert}{2(d-1)}+d-\mathcal{B}^{[d]}+\frac{\partial\mathcal{B}^{[d-1]}}{d-1}-C(\gamma)\bigg )+\sum_{i=0}^{d}\sum_{\rho}\omega(\mathcal{B}^{\hat{i}}_{(\rho)}).
    \end{align}
\end{proof}

Let us now consider the case we are interested in. Let $\mathcal{G}\in\mathfrak{G}_{3}$ be some connected open $(3+1)$-coloured graph with connected boundary graph $\gamma:=\partial\mathcal{G}$. In this case, the Gurau degree is given by
\begin{align}
    \label{Degree3D}
    \omega(\mathcal{G})=\frac{3}{2}\vert\mathcal{V}_{\mathcal{G},\mathrm{int}}\vert+2-\vert\mathcal{F}_{\mathcal{G},\mathrm{int}}\vert-\vert\mathcal{V}_{\gamma}\vert,
\end{align}
which, according to Proposition \ref{PropositionDegreeBubbles}, is equivalent to
\begin{align}
    \label{eqdegree3d}
    \omega(\mathcal{G})=\frac{1}{2}\vert\mathcal{V}_{\mathcal{G},\mathrm{int}}\vert-\frac{1}{4}\vert\mathcal{V}_{\gamma}\vert+2-\mathcal{B}^{[3]}+\frac{1}{2}\vert\mathcal{F}_{\gamma}\vert+\sum_{\text{$3$-bubbles } \mathcal{B}}g_{\mathcal{B}},
\end{align}
where $g_{\mathcal{B}}$ denotes the genus of the surface (possibly with boundary) represented by the $3$-bubble $\mathcal{B}$. Note that the sum in the expression above is always greater or equal to $0$ and it equals $0$ if and only if $\mathcal{G}$ represents a manifold.
\bigskip

In the case of closed graphs, it is a well-known fact that the Gurau degree is a non-negative quantity in arbitrary dimensions and can be bounded from below by a function depending on the degrees of its $d$-bubbles of some fixed colours \cite{GurauColouredTensorModelsReview}. For the case of open graphs, let us prove the following lower bound for the degree.

\begin{Theorem}[Lower Bound for Gurau Degree]
    \label{LowerBoundDegre}
    Let $\mathcal{G}\in\mathfrak{G}_{3}$ be some connected open $(3+1)$-coloured graph with $\partial\mathcal{G}=\gamma$, such that all its $3$-bubbles are simple (see Definition \ref{DefSimpleBubbles}) so that there are no pinching effects on the boundary complex. Then
    \begin{align*}
        \omega(\mathcal{G})\geq 2g_{\gamma}.
    \end{align*}
\end{Theorem}

\begin{proof}
    Let us start with the following general expression of the degree of $\mathcal{G}$ (see Equation \eqref{eqdegree3d}):
    \begin{align}
        \omega(\mathcal{G})=\frac{1}{2}\vert\mathcal{V}_{\mathcal{G},\mathrm{int}}\vert-\frac{1}{4}\vert\mathcal{V}_{\gamma}\vert+2-\mathcal{B}^{[3]}+\frac{1}{2}\vert\mathcal{F}_{\gamma}\vert+\sum_{\text{$3$-bubbles } \mathcal{B}}g_{\mathcal{B}}.
    \end{align}
    Now, since we have assumed that all $3$-bubbles are simple, we can write $\mathcal{B}^{[3]}=\vert\mathcal{F}_{\gamma}\vert+\mathcal{B}_{\mathrm{int}}^{[3]}$, where $\mathcal{B}_{\mathrm{int}}^{[3]}$ denotes the set of closed $3$-bubbles, since to every $2$-bubble $f$ on the boundary there is a unique corresponding open $3$-bubble $\mathcal{B}$ in $\mathcal{G}$ with $\partial\mathcal{B}=f$. Furthermore, let us use the fact that $\gamma$ represents the genus $g_{\gamma}$-surface, which yields the relation
    \begin{align}
        \label{EulerRelationGenus}
        2-2g_{\gamma}=\chi(\Delta_{\gamma})=\vert\mathcal{F}_{\gamma}\vert-\vert\mathcal{E}_{\gamma}\vert+\vert\mathcal{V}_{\gamma}\vert=\vert\mathcal{F}_{\gamma}\vert-\frac{1}{2}\vert\mathcal{V}_{\gamma}\vert\hspace*{1cm}\Rightarrow\hspace*{1cm} -\frac{1}{2}\vert\mathcal{F}_{\gamma}\vert=g_{\gamma}-\frac{1}{4}\vert\mathcal{V}_{\gamma}\vert-1.
    \end{align}
    Using the splitting $\mathcal{B}^{[3]}=\vert\mathcal{F}_{\gamma}\vert+\mathcal{B}_{\mathrm{int}}^{[3]}$ and plugging in the formula for $\vert\mathcal{F}_{\gamma}\vert$ from above, we find that the degree is given by
    \begin{align}
        \omega(\mathcal{G})=\frac{1}{2}(\vert\mathcal{V}_{\mathcal{G},\mathrm{int}}\vert-\vert\mathcal{V}_{\gamma}\vert)+1-\mathcal{B}_{\mathrm{int}}^{[3]}+g_{\gamma}+\sum_{\text{$3$-bubbles } \mathcal{B}}g_{\mathcal{B}}.
    \end{align}
    Now, note that $\vert\mathcal{E}_{\mathcal{G},\mathrm{int},0}\vert=\frac{1}{2}(\vert\mathcal{V}_{\mathcal{G},\mathrm{int}}\vert-\vert\mathcal{V}_{\gamma}\vert)$, where $\mathcal{E}_{\mathcal{G},\mathrm{int},0}$ denotes the set of internal edges of colour $0$, where internal means edges connecting two $4$-valent vertices of $\mathcal{G}$, as usual. Hence, we have that
    \begin{align}
        \label{DegreeInternalEdges}
        \omega(\mathcal{G})=\vert\mathcal{E}_{\mathcal{G},\mathrm{int},0}\vert+1-\mathcal{B}_{\mathrm{int}}^{[3]}+g_{\gamma}+\sum_{\text{$3$-bubbles } \mathcal{B}}g_{\mathcal{B}}.
    \end{align}
    In order to find an appropriate expression for $\vert\mathcal{E}_{\mathcal{G},\mathrm{int},0}\vert$, let us use the equality
    \begin{align}
        \omega(\mathcal{G})=\vert\mathcal{E}_{\mathcal{G},\mathrm{int},0}\vert+1-\mathcal{B}_{\mathrm{int}}^{[3]}+g_{\gamma}+\sum_{\text{$3$-bubbles } \mathcal{B}}g_{\mathcal{B}}=\frac{3}{2}\vert\mathcal{V}_{\mathcal{G},\mathrm{int}}\vert+3-\vert\mathcal{F}_{\mathcal{G},\mathrm{int}}\vert-(\vert\mathcal{V}_{\gamma}\vert+1),
    \end{align}
    where the expression on the right-hand side is just the definition of the degree (see \eqref{Degree3D}) for the $3$-dimensional case, which yields
    \begin{align}
        \label{eqEdges}
        \vert\mathcal{E}_{\mathcal{G},\mathrm{int},0}\vert=\frac{3}{2}\vert\mathcal{V}_{\mathcal{G},\mathrm{int}}\vert+1-\vert\mathcal{F}_{\mathcal{G},\mathrm{int}}\vert-\vert\mathcal{V}_{\gamma}\vert+\mathcal{B}_{\mathrm{int}}^{[3]}-g_{\gamma}-\sum_{\text{$3$-bubbles } \mathcal{B}}g_{\mathcal{B}}.
    \end{align}
    As a next step, note that we have the following relation between all the internal faces of colour $ij$ with $i\neq 0\neq j$ and internal edges of colour $\neq 0$:
    \begin{subequations}
        \begin{align}
            \sum_{\substack{i,j=1,\\i\neq j}}^{3}\vert\mathcal{F}_{\mathcal{G},\mathrm{int},ij}\vert-\sum_{i=1}^{3}\vert\mathcal{E}_{\mathcal{G},\mathrm{int},i}\vert+&\vert\mathcal{V}_{\mathcal{G},\mathrm{int}}\vert=\sum_{\substack{\text{internal $3$-bubbles}\\\text{of colour $123$}\,\mathcal{B}}}(2-2g_{\mathcal{B}})=\\&=\sum_{\substack{\text{internal}\\ \text{$3$-bubbles}\,\mathcal{B}}}(2-2g_{\mathcal{B}})-\sum_{\substack{\text{internal $3$-bubbles}\\\text{involving colour $0$}\,\mathcal{B}}}(2-2g_{\mathcal{B}})=\\&=2(\mathcal{B}_{\mathrm{int}}^{[3]}-\mathcal{B}_{\mathrm{int},0}^{[3]})-2\bigg(\sum_{\substack{\text{internal}\\ \text{$3$-bubbles}\,\mathcal{B}}}g_{\mathcal{B}}-\sum_{\substack{\text{internal $3$-bubbles}\\\text{involving colour $0$}\,\mathcal{B}}}g_{\mathcal{B}}\bigg),
        \end{align}
    \end{subequations}
    where $\mathcal{E}_{\mathcal{G},\mathrm{int},i}$ denotes the set of internal edges of colour $i$, where $\mathcal{F}_{\mathcal{G},\mathrm{int},ij}$ denotes the set of internal (cyclic) faces of $\mathcal{G}$ of colour $ij$ and where $\mathcal{B}_{\mathrm{int},0}^{[3]}$ denotes the set of internal $3$-bubbles involving colour $0$, i.e.~the number of internal $012$, $013$ and $023$-bubbles. Using the fact that $3\vert\mathcal{V}_{\mathcal{G},\mathrm{int}}\vert=2\sum_{i=1}^{3}\vert\mathcal{E}_{\mathcal{G},\mathrm{int},i}\vert$ and using the formula above for the number of internal faces not containing colour $0$, we arrive at
    \begin{subequations}
        \begin{align}
            \vert\mathcal{F}_{\mathcal{G},\mathrm{int}}\vert=&\sum_{i=1}^{3}\vert\mathcal{F}_{\mathcal{G},\mathrm{int},i0}\vert+\sum_{\substack{i,j=1,\\i\neq j}}^{3}\vert\mathcal{F}_{\mathcal{G},\mathrm{int},ij}\vert=\\=&\sum_{i=1}^{3}\vert\mathcal{F}_{\mathcal{G},\mathrm{int},i0}\vert+\frac{1}{2}\vert\mathcal{V}_{\mathcal{G},\mathrm{int}}\vert+2(\mathcal{B}_{\mathrm{int}}^{[3]}-\mathcal{B}_{\mathrm{int},0}^{[3]})-2\bigg(\sum_{\substack{\text{internal}\\ \text{$3$-bubbles}\,\mathcal{B}}}g_{\mathcal{B}}-\sum_{\substack{\text{internal $3$-bubbles}\\\text{involving colour $0$}\,\mathcal{B}}}g_{\mathcal{B}}\bigg).
        \end{align}
    \end{subequations}
    Plugging this expression for the number of internal faces back into our formula for $\vert\mathcal{E}_{\mathcal{G},\mathrm{int},0}\vert$, i.e.~Equation \eqref{eqEdges}, we find
    \begin{equation}
        \begin{aligned}
            \vert\mathcal{E}_{\mathcal{G},\mathrm{int},0}\vert=-\sum_{i=1}^{3}\vert\mathcal{F}_{\mathcal{G},\mathrm{int},i0}\vert+&\underbrace{\vert\mathcal{V}_{\mathcal{G},\mathrm{int}}\vert-\vert\mathcal{V}_{\gamma}\vert}_{2\vert\mathcal{E}_{\mathcal{G},\mathrm{int},0}\vert}+1-\mathcal{B}_{\mathrm{int}}^{[3]}-g_{\gamma}+\sum_{\text{$3$-bubbles } \mathcal{B}}g_{\mathcal{B}}+\\+&2\bigg(\mathcal{B}_{\mathrm{int},0}^{[3]}-\sum_{\substack{\text{internal $3$-bubbles}\\\text{involving colour $0$}\,\mathcal{B}}}g_{\mathcal{B}}\bigg)-2\sum_{\text{open $3$-bubbles}\,\mathcal{B}}g_{\mathcal{B}}
        \end{aligned}
    \end{equation}
    and hence we finally arrive at the equality
    \begin{equation}
        \begin{aligned}
            \label{InternalEdgesandFaces}
            \vert\mathcal{E}_{\mathcal{G},\mathrm{int},0}\vert=\sum_{i=1}^{3}\vert\mathcal{F}_{\mathcal{G},\mathrm{int},i0}\vert+\mathcal{B}_{\mathrm{int}}^{[3]}+(g_{\gamma}-1)&-\sum_{\text{$3$-bubbles } \mathcal{B}}g_{\mathcal{B}}\\&-2\bigg(\mathcal{B}_{\mathrm{int},0}^{[3]}-\sum_{\substack{\text{internal $3$-bubbles}\\\text{involving colour $0$}\,\mathcal{B}}}g_{\mathcal{B}}\bigg)+2\sum_{\text{open $3$-bubbles}\,\mathcal{B}}g_{\mathcal{B}}.
        \end{aligned}
    \end{equation}
    Plugging this relation for the number of internal edges of colour $0$ back into our expression for the degree \eqref{DegreeInternalEdges}, we arrive at the following formula for the degree of $\mathcal{G}$:
    \begin{align}
        \omega(\mathcal{G})=2g_{\gamma}+\underbrace{\sum_{i=1}^{3}\vert\mathcal{F}_{\mathcal{G},\mathrm{int},i0}\vert}_{\geq 0}+2\bigg(\sum_{\substack{\text{internal $3$-bubbles}\\\text{involving colour $0$}\,\mathcal{B}}}g_{\mathcal{B}}-\mathcal{B}_{\mathrm{int},0}^{[3]}\bigg)+2\underbrace{\sum_{\text{open $3$-bubbles}\,\mathcal{B}}g_{\mathcal{B}}}_{\geq 0}.
    \end{align}
    Let us assume without loss of generality that $\mathcal{G}$ is a core graph, since if it is not, then we can apply our rooting procedure and by Lemma \ref{DipoleDegree}, we know that the degree does not change. If $\mathcal{G}$ is a core graph, then all the internal $3$-bubbles involving colour $0$ are non-spherical and hence, we see that the bracket is non-negative in this case. Hence, we conclude that $\omega(\mathcal{G})\geq 2g_{\gamma}$, as claimed.
\end{proof}

In particular, this proof shows that the Gurau degree is always non-negative. As a straightforward application of the lower bound theorem proven above, we can show that the number of internal vertices of graphs representing manifolds can be bounded from below.

\begin{Corollary}
    \label{Cor:ManifoldVertices}
    Let $\gamma\in\overline{\mathfrak{G}}_{2}$ be a closed $(2+1)$-coloured graph representing a genus $g_{\gamma}$-surface. If $\mathcal{G}\in\mathfrak{G}_{3}$ is an open $(3+1)$-coloured graph with $\partial\mathcal{G}=\gamma$ representing a manifold, then $\vert\mathcal{V}_{\mathcal{G},\mathrm{int}}\vert\geq 2g_{\gamma}+\vert\mathcal{V}_{\gamma}\vert$.
\end{Corollary}

\begin{proof}
    Let us assume without loss of generality that $\mathcal{G}$ is a core graph, since if it is not a core graph, we can apply our rooting procedure in order to obtain a core graph, which by construction represents the same manifold as $\mathcal{G}$, but with a smaller number of internal vertices. If $\mathcal{G}$ is a core graph representing a manifold, then it has in total $1+\vert\mathcal{F}_{\gamma}\vert$ $3$-bubbles, one internal one of colour $123$ representing a $2$-sphere, and for each face $f$ of $\gamma$ a corresponding open $3$-bubbles $\mathcal{B}$ representing a disk with $\partial\mathcal{B}=f$. Furthermore, recall that the Euler characteristic of any odd-dimensional compact and orientable manifold $\mathcal{M}$ satisfies $\chi(\mathcal{M})=\frac{1}{2}\chi(\partial\mathcal{M})$. Hence, we can write
    \begin{subequations}
        \begin{align}
            1-g_{\gamma}&=\frac{1}{2}\chi(\Delta_{\gamma})=\chi(\Delta_{\mathcal{G}})=1+\vert\mathcal{F}_{\gamma}\vert-\underbrace{\vert\mathcal{F}_{\mathcal{G}}\vert}_{\vert\mathcal{F}_{\mathcal{G},\mathrm{int}}\vert+\vert\mathcal{E}_{\gamma}\vert}+\underbrace{\vert\mathcal{E}_{\mathcal{G}}\vert}_{\vert\mathcal{E}_{\mathcal{G},\mathrm{int}}\vert+\vert\mathcal{V}_{\gamma}\vert}-\vert\mathcal{V}_{\mathcal{G},\mathrm{int}}\vert=\\&=1+\underbrace{(\vert\mathcal{F}_{\gamma}\vert-\vert\mathcal{E}_{\gamma}\vert+\vert\mathcal{V}_{\gamma}\vert)}_{=\chi(\Delta_{\gamma})=2-2g_{\gamma}}-\vert\mathcal{F}_{\mathcal{G},\mathrm{int}}\vert+\vert\mathcal{E}_{\mathcal{G},\mathrm{int}}\vert-\vert\mathcal{V}_{\mathcal{G},\mathrm{int}}\vert,
        \end{align}
    \end{subequations}
    which yields the following relation:
    \begin{align}
        \label{eq:IntVert}
        \vert\mathcal{V}_{\mathcal{G},\mathrm{int}}\vert=2-g_{\gamma}-\vert\mathcal{F}_{\mathcal{G},\mathrm{int}}\vert+\vert\mathcal{E}_{\mathcal{G},\mathrm{int}}\vert
    \end{align}
    Next, we can use the definition of the degree (see Equation \eqref{Degree3D}), i.e.
    \begin{align}
        \omega(\mathcal{G})=\frac{3}{2}\vert\mathcal{V}_{\mathcal{G},\mathrm{int}}\vert+2-\vert\mathcal{F}_{\mathcal{G},\mathrm{int}}\vert-\vert\mathcal{V}_{\gamma}\vert.
    \end{align}
    Using the lower bound theorem for the degree, Theorem \ref{LowerBoundDegre}, we can deduce the following estimate
    \begin{align}
        -\vert\mathcal{F}_{\mathcal{G},\mathrm{int}}\vert=\omega(\mathcal{G})-\frac{3}{2}\vert\mathcal{V}_{\mathcal{G},\mathrm{int}}\vert-2+\vert\mathcal{V}_{\gamma}\vert\geq 2g_{\gamma}-\frac{3}{2}\vert\mathcal{V}_{\mathcal{G},\mathrm{int}}\vert-2+\vert\mathcal{V}_{\gamma}\vert.
    \end{align}
    Applying this inequality to Formula \eqref{eq:IntVert}, we get
    \begin{align}
        \vert\mathcal{V}_{\mathcal{G},\mathrm{int}}\vert\geq g_{\gamma}+\vert\mathcal{E}_{\mathcal{G},\mathrm{int}}\vert-\frac{3}{2}\vert\mathcal{V}_{\mathcal{G},\mathrm{int}}\vert+\vert\mathcal{V}_{\gamma}\vert.
    \end{align}
    Last but not least, we use the fact that $2\vert\mathcal{E}_{\mathcal{G},\mathrm{int}}\vert=4\vert\mathcal{V}_{\mathcal{G},\mathrm{int}}\vert-\vert\mathcal{V}_{\gamma}\vert$, which results into the inequality
    \begin{align}
        \vert\mathcal{V}_{\mathcal{G},\mathrm{int}}\vert\geq g_{\gamma}+\frac{1}{2}(\vert\mathcal{V}_{\mathcal{G},\mathrm{int}}\vert+\vert\mathcal{V}_{\gamma}\vert)\hspace*{1cm}\Leftrightarrow\hspace*{1cm}\vert\mathcal{V}_{\mathcal{G},\mathrm{int}}\vert\geq 2g_{\gamma}+\vert\mathcal{V}_{\gamma}\vert,
    \end{align}
    as claimed.
\end{proof}

\begin{Remark}
    \label{Rem:HandleBodies}
    A similar result has recently been proven in \cite{HandleBodiesCT} using a completely different approach by applying techniques from crystallization theory. More precisely, it was shown that $\frac{\vert\mathcal{V}_{\mathcal{G},\mathrm{int}}\vert}{2}-1\geq 3g_{\gamma}$ for every ``\textit{crystallization}'' $\mathcal{G}$ (see Appendix \ref{SecCryTheo}) of a three-dimensional manifold with connected boundary given by a genus $g_{\gamma}$-surface, which in our language is an open $(3+1)$-coloured core graph representing a manifold, for which also the boundary graph $\gamma$ is a (closed) core graph. Using the fact that every core graph $\gamma$ representing a genus $g_{\gamma}$-surface has precisely $4g_{\gamma}+2$ vertices, we see that the above statement is equivalent to $\vert\mathcal{V}_{\mathcal{G},\mathrm{int}}\vert\geq 2g_{\gamma}+\vert\mathcal{V}_{\gamma}\vert$.
\end{Remark}

In the case of a spherical boundary, the lower bound theorem tells us that the Gurau degree is always non-negative. Using this, we are finally in the position to prove the following result.

\begin{Proposition}[Smallest Matching Graphs and Degree]
    Let $\gamma\in\overline{\mathfrak{G}}_{2}$ be some closed $(2+1)$-coloured graph representing a $2$-sphere. If $\mathcal{G}\in\mathfrak{G}_{3}$ represents a manifold or pseudomanifold without boundary singularities and with $\partial\mathcal{G}=\gamma$, then $\mathcal{G}$ roots back to the core equivalence class defined by the smallest matching graph if and only if $\omega(\mathcal{G})=0$. In other words, the family of graphs rooting back to the smallest matching graph are exactly the graphs matching the given boundary with minimal degree.
\end{Proposition}

    \begin{proof}First of all, let us observe the following: A $1$-dipole move, with the property that at least one of the two separated $3$-bubbles is closed, reduces the number of internal vertices by two and the number of $3$-bubbles by one. Hence, the quantity
    \begin{align}
        \frac{\vert\mathcal{V}_{\mathcal{G},\mathrm{int}}\vert}{2}-\mathcal{B}^{[3]}
    \end{align}
    is conserved under arbitrary internal $1$-dipole moves. Now, let us apply as many internal $1$-dipole moves as possible. In the end, we will end up with a graph $\mathcal{G}_{c}$ having precisely $1+\vert\mathcal{F}_{\gamma}\vert$ $3$-bubbles, i.e.~one internal one of colour $123$ and for each face on the boundary a corresponding $3$-bubble representing a disk whose boundary is given by that face. The number of internal vertices has to satisfy $\vert\mathcal{V}_{\mathcal{G}_{c},\mathrm{int}}\vert\geq \vert\mathcal{V}_{\gamma}\vert$. Of course, the topology of $\mathcal{G}_{c}$ is in general different from $\mathcal{G}$, however, this is not so important at this point. What is important is that we have that
    \begin{align}
        \frac{\vert\mathcal{V}_{\mathcal{G},\mathrm{int}}\vert}{2}-\mathcal{B}^{[3]}=\frac{\vert\mathcal{V}_{\mathcal{G}_{\mathrm{c}},\mathrm{int}}\vert}{2}-(1+\vert\mathcal{F}_{\gamma}\vert).
    \end{align}
    and hence, we have that (c.f.~Corollary \ref{Cor:ManifoldVertices} for the trivial case $g_{\gamma}=0$)
    \begin{align}
        \frac{\vert\mathcal{V}_{\mathcal{G},\mathrm{int}}\vert}{2}-\mathcal{B}^{[3]}=\frac{\vert\mathcal{V}_{\mathcal{G}_{\mathrm{c}},\mathrm{int}}\vert}{2}-(1+\vert\mathcal{F}_{\gamma}\vert)\geq \frac{\vert\mathcal{V}_{\gamma}\vert}{2}-(1+\vert\mathcal{F}_{\gamma}\vert).
    \end{align}
    Applying this to the definition of the degree (Equation \ref{eqdegree3d}), we get the following general inequality
    \begin{align}
        \omega(\mathcal{G})\geq\underbrace{\frac{1}{4}\vert\mathcal{V}_{\gamma}\vert+1-\frac{1}{2}\vert\mathcal{F}_{\gamma}\vert}_{=0}+\sum_{\text{$3$-bubbles}\,\mathcal{B}}g_{\mathcal{B}},
    \end{align}
    where we have used the fact that $\frac{1}{2}\vert\mathcal{F}_{\gamma}\vert=\frac{1}{4}\vert\mathcal{V}_{\gamma}\vert+1$ as derived previously (see Equation \eqref{EulerRelationGenus}). Using this, we see that the degree of pseudomanifolds is strictly positive, since $\sum_{\text{$3$-bubbles}\,\mathcal{B}}g_{\mathcal{B}}>0$, and hence, they never saturate the bound $\omega(\mathcal{G})\geq 2g_{\gamma}=0$. Let us now turn our attention to manifolds. By Lemma \ref{DipoleDegree} it is enough to look at core graphs. Now, since $\mathcal{G}$ is a core graph representing a manifold, we have that $\mathcal{B}^{[3]}=1+\vert\mathcal{F}_{\gamma}\vert$. Therefore, according to Formula \eqref{eqdegree3d}, its degree is given by
    \begin{align}
        \omega(\mathcal{G})=\frac{1}{2}\vert\mathcal{V}_{\mathcal{G},\mathrm{int}}\vert-\frac{1}{4}\vert\mathcal{V}_{\gamma}\vert+1-\frac{1}{2}\vert\mathcal{F}_{\gamma}\vert.
    \end{align}
    Using again the fact that $\frac{1}{2}\vert\mathcal{F}_{\gamma}\vert=\frac{1}{4}\vert\mathcal{V}_{\gamma}\vert+1$, this can be written as
    \begin{align}
        \omega(\mathcal{G})=\frac{1}{2}(\vert\mathcal{V}_{\mathcal{G},\mathrm{int}}\vert-\vert\mathcal{V}_{\gamma}\vert)=\vert\mathcal{E}_{\mathcal{G},\mathrm{int},0}\vert.
    \end{align}
    With this equality, it is clear that $\omega(\mathcal{G})=0$ if and only if $\vert\mathcal{V}_{\mathcal{G},\mathrm{int}}\vert=\vert\mathcal{V}_{\gamma}\vert$ and the only possible graph $\mathcal{G}$ with $\partial\mathcal{G}=\gamma$ satisfying this condition is the smallest matching graph. According to Lemma \ref{LowerBoundDegre}, the degree of these graphs is minimal.
\end{proof}

\begin{Remark}
    If $\gamma$ is some boundary graph representing a general genus $g_{\gamma}$-surface, then the Gurau degree of the smallest matching graph is also minimal and hence given by $\omega(\mathcal{G}_{\mathrm{SMG}})=2g_{\gamma}$. To see this, note that the smallest matching graph has by definition one internal $3$-bubble $\mathcal{B}$, which is equivalent to the graph $\gamma$ and hence represents a genus $g_{\gamma}$-surface. Therefore, by Equation \eqref{eqdegree3d}, its degree is given by
    \begin{align}
        \omega(\mathcal{G}_{\mathrm{SMG}})=\frac{1}{2}\vert\mathcal{V}_{\mathcal{G}_{\mathrm{SMG}},\mathrm{int}}\vert-\frac{1}{4}\vert\mathcal{V}_{\gamma}\vert+1-\frac{1}{2}\vert\mathcal{F}_{\gamma}\vert+\underbrace{\omega(\mathcal{B})}_{=g_{\gamma}}.
    \end{align}
    Using again Equation \eqref{EulerRelationGenus}, we hence get
    \begin{align}
        \omega(\mathcal{G}_{\mathrm{SMG}})=\frac{1}{2}(\underbrace{\vert\mathcal{V}_{\mathcal{G}_{\mathrm{SMG}},\mathrm{int}}\vert-\vert\mathcal{V}_{\gamma}\vert}_{=\vert\mathcal{E}_{\mathcal{G}_{\mathrm{SMG}},\mathrm{int},0}\vert=0})+2g_{\gamma}=2g_{\gamma}.
    \end{align}
    for the degree of the smallest matching graph. However, in this case, it turns out that the smallest matching graph is not the only graph with the minimal degree. As an example, consider the simplest possible boundary graph representing a $2$-torus, as drawn in figure \ref{TorusBoundaryGraph}. In this case, a straightforward calculation shows that not only the smallest matching graph ($\mathcal{G}_{0}$ in figure \ref{SolidTorusExamples}) has degree $2$, but also the core graphs $\mathcal{G}_{1}$ and $\mathcal{G}_{1}^{\prime}$ (figure \ref{SolidTorusExamples}), which represent the solid torus, have degree $2$.
\end{Remark}

To sum up, we see that the family of graphs with the minimal possible Gurau degree for some given spherical boundary graphs, includes exactly those graphs rooting back to the core equivalence class induced by the smallest matching graph. In that sense, they can be viewed as generalizations of melonic diagrams used in the discussion of the large $N$ limit of the free energy.

\subsection{Leading Order Contribution}
Let us now show that the core equivalence class defined by the smallest matching forms the leading order contribution to the Boulatov transition amplitude with respect to some spherical boundary graph when we consider only manifolds. Before stating the main result, we need the following preliminary technical lemma.

\begin{Lemma}
    \label{LemmaInt}
    Consider an arbitrary closed $(2+1)$-coloured graph $\gamma\in\overline{\mathfrak{G}}_{2}$ representing the genus $g_{\gamma}$-surface. Let $\mathcal{G}$ be a connected open $(3+1)$-coloured core graph with boundary $\partial\mathcal{G}=\gamma$, which is dual to a manifold. Then
    \begin{align*}
        g_{\gamma}=\vert\mathcal{E}_{\mathcal{G},\mathrm{int},0}\vert-\sum_{i=1}^{3}\vert\mathcal{F}_{\mathcal{G},\mathrm{int},i0}\vert
    \end{align*}
    where $\mathcal{E}_{\mathcal{G},\mathrm{int},0}$ denotes the set of internal edges of colour $0$ and where $\mathcal{F}_{\mathcal{G},\mathrm{int},i0}$ denotes the set of internal (cyclic) faces of $\mathcal{G}$ of colour $0i$ for $i\in\{1,2,3\}$.
\end{Lemma}

\begin{proof}
    This is a special case of Equation \eqref{InternalEdgesandFaces}: Since $\mathcal{G}$ is a core graph representing a manifold, there is only one internal $3$-bubble, which is spherical and has colour $123$, and all the other $3$-bubbles are open and represent $2$-balls.
\end{proof}

In particular, if $\gamma$ is a spherical graph ($g_{\gamma}=0$), we obtain that the number of internal edges of colour $0$ is the same as the number of internal faces involving colour $0$. Using this observation, let us prove that the core equivalence class defined by the smallest matching graph is the dominant contribution to the transition amplitude when restricted to manifolds:

\begin{Theorem}[Leading Order and Bound for Core Graphs]
    Consider an arbitrary closed $(2+1)$-coloured graph $\gamma\in\overline{\mathfrak{G}}_{2}$ representing a $2$-sphere and let $\mathcal{G}$ be a connected open $(3+1)$-coloured graph with boundary $\partial\mathcal{G}=\gamma$, which is dual to a manifold and which is itself a core graph. Then
    \begin{align*}
        \vert\mathcal{A}_{\mathcal{G}}^{\lambda}[\{g_{e}\}_{e\in\mathcal{E}_{\gamma}}]\vert\leq (\lambda\overline{\lambda})^{\frac{\vert\mathcal{V}_{\mathcal{G},\mathrm{int}}\vert}{2}}\delta^{N}(\mathds{1})^{1-\frac{\vert\mathcal{V}_{\gamma}\vert}{2}}\bigg\{\frac{1}{\delta^{N}(\mathds{1})}\prod_{f\in\mathcal{F}_{\gamma}}\delta^{N}\bigg (\overrightarrow{\prod_{e\in f}}g_{e}^{\varepsilon(e,f)}\bigg )\bigg\}
    \end{align*}
    i.e.~its degree of divergence is smaller or equal to $1-\vert\mathcal{V}_{\gamma}\vert/2$ (recall that there is one redundant delta function encoded in the product over boundary faces). Furthermore, the only core graph saturating this bound is the smallest matching graph, i.e.~its amplitude is exactly given by
    \begin{align*}
        \mathcal{A}_{\mathcal{G}_{\mathrm{SMG}}}^{\lambda}[\{g_{e}\}_{e\in\mathcal{E}_{\gamma}}]= (\lambda\overline{\lambda})^{\frac{\vert\mathcal{V}_{\gamma}\vert}{2}}\delta^{N}(\mathds{1})^{1-\frac{\vert\mathcal{V}_{\gamma}\vert}{2}}\bigg\{\frac{1}{\delta^{N}(\mathds{1})}\prod_{f\in\mathcal{F}_{\gamma}}\delta^{N}\bigg (\overrightarrow{\prod_{e\in f}}g_{e}^{\varepsilon(e,f)}\bigg )\bigg\}.
    \end{align*}
\end{Theorem}

\begin{proof}
    To start with, let us write down the general expression of the amplitude of some open $(3+1)$-coloured core graph representing a manifold. For this, we use the same terminology as previously: group elements assigned to external legs are denoted by $\{h_{v}\}_{v\in\mathcal{V}_{\gamma}}$, bicoloured paths leading to the boundary edges by $\{H_{e}\}_{e\in\mathcal{E}_{\gamma}}$ and boundary edges by $\{g_{e}\}_{e\in\mathcal{E}_{\gamma}}$. We also denote the group elements assigned to all the internal edges of $\mathcal{G}$ by $\{k_{e}\}_{e\in\mathcal{E}_{\mathcal{G},\mathrm{int}}}$. With this notation, the amplitude can be written as
    \begin{equation}
        \begin{aligned}
            \mathcal{A}_{\mathcal{G}}^{\lambda}[\{g_{e}\}_{e\in\mathcal{E}_{\gamma}}]=\bigg(\frac{\lambda\overline{\lambda}}{\delta^{N}(\mathds{1})}\bigg)^{\frac{\vert\mathcal{V}_{\mathcal{G},\mathrm{int}}\vert}{2}}\int_{\mathrm{SU}(2)^{\vert\mathcal{V}_{\gamma}\vert+\vert\mathcal{E}_{\mathcal{G},\mathrm{int}}\vert}}\,&\bigg(\prod_{v\in\mathcal{V}_{\gamma}}\mathrm{d}h_{v}\bigg )\bigg(\prod_{e\in\mathcal{E}_{\mathcal{G},\mathrm{int}}}\mathrm{d}k_{e}\bigg )\times\\&\times\prod_{e\in\mathcal{E}_{\gamma}}\delta^{N}(g_{e}h_{t(e)}^{-1}H_{e}h_{s(e)}^{-1})\prod_{f\in\mathcal{F}_{\mathcal{G},\mathrm{int}}}\delta^{N}\bigg(\overrightarrow{\prod_{e\in f}}k_{e}^{\varepsilon(e,f)}\bigg).
        \end{aligned}
    \end{equation}
    As explained in Theorem \ref{GenSphereFact}, we can replace the first product of delta functions, which contains all the boundary group elements, by the flatness of the boundary up to a redundancy
    \begin{align}
        \prod_{e\in\mathcal{E}_{\gamma}}\delta^{N}(g_{e}h_{t(e)}^{-1}H_{e}h_{s(e)}^{-1})\to\frac{1}{\delta^{N}(\mathds{1})}\prod_{f\in\mathcal{F}_{\gamma}}\delta^{N}\bigg (\overrightarrow{\prod_{e\in f}}g_{e}^{\varepsilon(e,f)}\bigg ).
    \end{align}
    We are left with the product over internal faces. To start with, let us split the product as follows:
    \begin{align}
        \prod_{f\in\mathcal{F}_{\mathcal{G},\mathrm{int}}}\delta^{N}\bigg(\overrightarrow{\prod_{e\in f}}k_{e}^{\varepsilon(e,f)}\bigg)=\bigg\{\prod_{f\in\bigcup_{i,j\neq 0}\mathcal{F}_{\mathcal{G},\mathrm{int},ij}}\delta^{N}\bigg(\overrightarrow{\prod_{e\in f}}k_{e}^{\varepsilon(e,f)}\bigg)\bigg\}\bigg\{\prod_{f\in \bigcup_{i\neq 0}\mathcal{F}_{\mathcal{G},\mathrm{int},i0}}\delta^{N}\bigg(\overrightarrow{\prod_{e\in f}}k_{e}^{\varepsilon(e,f)}\bigg)\bigg\},
    \end{align}
    where $\mathcal{F}_{\mathcal{G},\mathrm{int},ij}$ denotes the set of internal (cyclic) faces of colour $ij$, as before. In other words, the first product only contains faces of colour $ij$ with $i\neq 0\neq j$ and the second product contains all the internal faces involving colour $0$. Now, since $\mathcal{G}$ is a core graph representing a manifold, we know that there is only one internal $3$-bubble, which has colour $123$ and which is spherical. Let us denote this bubble by $\mathcal{B}$. Hence, the first product of delta functions exactly contains all the delta functions associated to the faces of a spherical $(2+1)$-coloured graph, namely $\mathcal{B}$. Hence, we know that, after integrating over some internal edges of colour $123$, this product can be reduced to $\delta^{N}(\mathds{1})\delta^{N}(\overrightarrow{\prod_{e\in f_{0}}}k_{e}^{\varepsilon(e,f)})$ for some closed path $f_{0}$, in relation to the discrete Bianchi identity. Of course, when performing these integrations, the delta functions contained in the second product also change. Their number however stays the same.

    Now, if $\mathcal{G}$ is the smallest matching graph, the second product is empty. Therefore, we can trivially integrate over the group element associated to any internal edge contained in $f_0$. In that case, the amplitude is simply given by
    \begin{align}
        \mathcal{A}_{\mathcal{G}}^{\lambda}[\{g_{e}\}_{e\in\mathcal{E}_{\gamma}}]=\bigg(\frac{\lambda\overline{\lambda}}{\delta^{N}(\mathds{1})}\bigg)^{\frac{\vert\mathcal{V}_{\gamma}\vert}{2}}\delta^{N}(\mathds{1})\bigg\{\frac{1}{\delta^{N}(\mathds{1})}\prod_{f\in\mathcal{F}_{\gamma}}\delta^{N}\bigg (\overrightarrow{\prod_{e\in f}}g_{e}^{\varepsilon(e,f)}\bigg )\bigg\}
    \end{align}
    as claimed.

    If $\mathcal{G}$ is not the smallest matching graph, then the second product is not empty. Indeed, recall that the number of internal edges of colour $0$ is the same as the number of internal faces involving colour $0$, i.e.~the faces contained in the second product, according to Lemma \ref{LemmaInt}. Since for every graph with $\vert\mathcal{V}_{\mathcal{G},\mathrm{int}}\vert>\vert\mathcal{V}_{\gamma}\vert$ there is at least one internal edge of colour $0$, we conclude that the second product is non-empty. The total number of delta functions contained in this second product is hence
    \begin{align}
        \sum_{i=1}^{3}\vert\mathcal{F}_{\mathcal{G},\mathrm{int},i0}\vert=\vert\mathcal{E}_{\mathcal{G},\mathrm{int},0}\vert.
    \end{align}
    Since there exists at least one internal group element of colour $0$, we can freely integrate at least over one of them. Similarly, we can also freely integrate over the delta function corresponding to the face $f_{0}$, which is the only remaining delta function of the first product. Bounding all the remaining $\vert\mathcal{E}_{\mathcal{G},\mathrm{int},0}\vert-1$ delta functions simply by $\delta^{N}(\mathds{1})$, we are hence left with a maximal possible degree of divergence of
    \begin{align}
        -\frac{\vert\mathcal{V}_{\mathcal{G},\mathrm{int}}\vert}{2}+1+(\vert\mathcal{E}_{\mathcal{G},\mathrm{int},0}\vert-1)=-\frac{\vert\mathcal{V}_{\mathcal{G},\mathrm{int}}\vert}{2}+\vert\mathcal{E}_{\mathcal{G},\mathrm{int},0}\vert=-\frac{\vert\mathcal{V}_{\gamma}\vert}{2},
    \end{align}
    where we used the fact that $\vert\mathcal{E}_{\mathcal{G},\mathrm{int},0}\vert=\frac{1}{2}(\vert\mathcal{V}_{\mathcal{G},\mathrm{int}}\vert-\vert\mathcal{V}_{\gamma}\vert)$ in the last step.
\end{proof}

An immediate consequence of the previous theorem is the following corollary.

\begin{Corollary}
    Consider an arbitrary closed $(2+1)$-coloured graph $\gamma\in\overline{\mathfrak{G}}_{2}$ representing a $2$-sphere. Then the leading order contribution to the transition amplitude restricted to manifolds $\langle\mathcal{Z}_{\mathrm{cBM}}\vert\Psi\rangle_{\text{manifolds}}$ for some spin network $\Psi$ defined on $\gamma$ is the $3$-ball represented by the core equivalence class defined by the smallest matching graph.
\end{Corollary}
\section{Conclusion and Outlook}
In the present work, we analysed the transition amplitudes of 3d Riemannian quantum gravity, in the context of the simplicial coloured Boulatov GFT model. For this, techniques from crystallization theory turned out to be useful. In particular, the concept of dipole moves allowing to relate different graphs in a topology and boundary preserving way, has been central in our analysis. The boundary states of this model are spin network states and boundary observables associated to them are $\mathrm{SU}(2)$-invariant functionals of the GFT fields. These are constructed from spin network states living on a fixed boundary triangulation, and encoding quantum geometric data. The transition amplitudes of the Boulatov model are then defined to be the expectation values of these observables and can be interpreted as the corresponding probability amplitudes for a transition between two components of the given boundary complex, or in case of a connected boundary, for a transition from the (full, non-perturbative) vacuum state (similar to the Hartle-Hawking state). By construction, these amplitudes generically involve a sum over all simplicial complexes matching our given boundary triangulation, where each complex is weighted by a Ponzano-Regge spin foam amplitude. In other words, by the general existence theorems of coloured graphs of crystallization theory, the transition amplitudes includes a sum over all admissible (bulk) topologies in addition to a sum over geometries. Three-dimensional general relativity is a particular example of a topological (BF) field theory, since it has no local degrees of freedom, and hence, the sum over geometries is somewhat trivial in this case (it computes only the volume of the space of flat connections on the given topology). On the level of the Ponzano-Regge model, this is reflected by the fact that the spin foam amplitudes are invariant under the chosen bulk triangulations for some fixed topology and hence only depend on the boundary data. However, a sum over topologies in three-dimensional quantum gravity still includes non-trivial features: different topologies lead to different amplitudes and there might be a non-trivial gluing between the bulk and the boundary complex leading to non-trivial information about admissible bulk topologies. Furthermore, in the (coloured) GFT approach, one also has to consider more singular topologies than manifolds, namely pseudomanifolds. This gives additional, and a priori very different, contributions, which may not even encode flatness of the discrete connection of the boundary complex, when they involve singularities touching the boundary.
\bigskip

Then, we generalized the rooting procedure developed in the series of papers \cite{GurauLargeN1,GurauLargeN2,GurauLargeN3} in the context of the large $N$ limit of the free energy, to coloured graphs with non-empty boundaries. This procedure allows us to reduce the discussion only to core graphs, which from the geometrical point of view correspond to simplicial complexes with the minimal number of vertices in the bulk triangulation. From the graph-theoretical point of view, the rooting procedure contracts all the internal proper $1$-dipoles of a coloured graphs, which are exactly those dipoles, which leave the topology, boundary as well as the degree of divergence unchanged. The number of contractions is independent of the order in which they are contracted, which shows that each graph roots back to a unique equivalence class of core graphs, each of which having the same number of vertices, boundary and topology. In particular, graphs rooting back to some equivalence class of core graphs have the same amplitude up to a factor of the interaction coupling and a possible symmetry factor, and hence the rooting procedure allows us to write the transition amplitudes as topological expansions, where each term appearing in the sum is given by an equivalence class of core graphs representing a fixed bulk topology. Note however that, in general, the same topology appears more than once in the expansion, since there exist infinitely many core graphs for a given bulk topology.
\bigskip

To illustrate the formalism developed in the present work, we analysed the case of boundary graphs representing the simplest boundary topology, the $2$-sphere. In this case, we were able to show that every manifold and every pseudomanifold without singularities touching the boundary complex yields the same contribution from the boundary spin network state to the transition amplitude, namely the spin network evaluation, which encodes flatness of the discrete boundary connection. That is, the contribution to the transition amplitude of \textit{any} bulk topology is morally the same. The transition amplitude --when restricted to those topologies only admitting bulk singularities-- factorizes into a prefactor consisting of all the factors coming from the interaction term and some remaining contributions coming from the bulk of various topologies, times the spin network evaluation. The prefactor can of course always be cancelled by choosing an appropriate scaling of the path integral. This results is also particularly interesting from the point of view of the holographic principle. It is well known that for a certain choice of boundary state, namely the generating function of spin network, the Ponzano-Regge model is dual to two copies of the Ising model living on the spherical boundary \cite{PRBallIsing1,PRBallIsing2}. Since the Boulatov transition amplitude for a spherical boundary graph factorizes and is proportional to the spin network evaluation, the same conclusion applies also to the Boulatov model.

Therefore, the results presented in this work provide a first insight into the holographic nature of the (coloured) Boulatov model for three-dimensional quantum gravity.

However, there remain several open questions. Within the full transition amplitude, one has to take into account pseudomanifolds with singularities touching the boundary and more work is needed to understand their contribution. In particular, for a topology with boundary singularities, we do not expect to recover flatness of the boundary connection and hence to get an amplitude proportional to the spin network evaluation. In other words, these topologies will have different contributions, which need to be studied systematically. In particular, it would be interesting to study their relations to local defects and particles within the context of discrete quantum gravity models. 
\bigskip

The result for a spherical boundary topology discussed above can be explained as follows: For manifolds without boundary singularities, the model encodes the flatness of the boundary, disregarding the topology of the bulk. In the case of a sphere, there does not exist any non-trivial flat connection. That is, any spin network on the sphere respecting its flatness must collapse to its spin network evaluation. Therefore, even if the transition amplitude should in principle depend on the topology of the bulk, due to the simple choice of the boundary topology, it all collapses to the spin network evaluation whatever the bulk topology is. One might naively assume that the same is true for more complicated boundary topologies. However, it turns out not to be as simple. As we have illustrated, this intuition already fails in the case of a torus boundary. The fundamental group of the $2$-torus $T^{2}$ is given by $\pi(T^{2})\cong\mathbb{Z}^{2}$ and the corresponding generators can be interpreted as the two non-contractible cycles. When considering a manifold with torus boundary, like the handlebody of genus $1$, the solid torus, then there are a priori two possible ways to glue the bulk to the boundary, differing by the choice of which cycle becomes contractible through the bulk. Following the logic of the spherical case, one could have expected only two contributions to the Boulatov transition amplitudes. However, as we have shown with a simple example, the situation is more complicated, since we obtain, at the very least, as many contribution as independent $3$-coloured closed paths on the boundary (without taking into account any possible winding). A more detailed analysis of the structure of the transition amplitude is in progress \cite{ToAppearTorus}, and it is necessary in order to understand how the choice of boundary topology affects the transition amplitude and the possible dual theory of the GFT model. 
\bigskip

As a next step, we have shown that the leading order contribution to the transition amplitude of some spherical boundary graph, when restricted only to manifolds, is given by certain graphs representing the closed $3$-ball. More precisely, these graphs are given by the smallest open coloured graph matching our given spherical boundary graph. Furthermore, we have shown that the class of graphs rooting back to this core equivalence class is precisely the collection of graphs for which a suitable generalization of the Gurau degree to graphs with non-empty boundary is minimal. In this sense, these graphs can be viewed as a generalization of melonic diagrams, which are the leading order graphs in the expansion of the free energy in the large $N$ limit. A question which remains open is whether this result still holds when including pseudomanifolds to the discussion. In the closed case, pseudomanifolds can be shown to be bounded and suppressed \cite{GFTVertex} and hence, one could hope that a similar result can be obtained for the case of open graphs in order to generalize the statement about the leading order made above. Additionally, it would be interesting to pursue a similar analysis for more complicated topologies, for example in the case of a torus boundary.
\section*{Acknowledgements}
\addcontentsline{toc}{section}{\hspace{15pt}Acknowledgements}
The authors would like to thank the whole ``Quantum Gravity and Foundations of Physics'' research group at LMU for useful discussions and comments. CG was supported by the Alexander von Humboldt Foundation. CG and GS gratefully acknowledge funding by the Deutsche Forschungsgemeinschaft (DFG) via a seed funding from the Munich Center of Quantum Science and Technology (MCQST) for the project ``Holographic dualities in 3D discrete models of quantum gravity'' AOST 862909-9. DO acknowledges funding from DFG research grants OR432/3-1 and OR432/4-1.
\appendix
\section{Simplicial Complexes and Pseudomanifolds}
\label{AppendixA:PseudoManifolds}
In this section, we briefly recall the definition of pseudomanifolds, in order to fix the terminology and notation used throughout this paper. First of all, let us fix the following terminology: Let $\Delta$ be an (abstract) (pseudo)simplicial complex with vertex set $\mathcal{V}$. Then:
\begin{itemize}
    \item An element $v\in\mathcal{V}$ is called \textit{vertex} and an element $\sigma\in\Delta$ is called \textit{simplex}. Any non-empty subset $\tau\subset\sigma$ is called \textit{face} of $\sigma$.
    \item The \textit{dimension} of a simplex $\sigma\in\Delta$ is the number $d\in\mathbb{N}$ defined by $d:=\vert\sigma\vert-1$. A $d$-dimensional simplex is also called \textit{$d$-simplex} and a $k$-dimensional face of $\sigma$ is also called a \textit{$k$-face of $\sigma$}. Let us denote the set of all $d$-simplices by $\Delta_{d}$. Vertices are by definition $0$-simplices, i.e.~$\mathcal{V}=\Delta_{0}$. The \textit{dimension of a simplicial complex} $\Delta$ is the maximal number $d\in\mathbb{N}$ such that $\Delta_{d}\neq\emptyset$.
    \item The collection of all simplices with dimension smaller equal to some $k\in\{0,\dots,d\}$ is called the \textit{$k$-skeleton} of the complex $\Delta$.
\end{itemize}

Let now $S\subset\Delta$ be a subset of some abstract simplicial complex $\Delta$. If $S$ is by itself an abstract simplicial complex, then it is called a \textit{subcomplex of $\Delta$}. Let us further introduce the following terminology:

\begin{itemize}
    \item[(1)]The \textit{closure} $\mathrm{Cl}_{\Delta}(S)$ of $S$ is the smallest subcomplex of $\Delta$ containing $S$, i.e.
    \begin{align}
        \mathrm{Cl}_{\Delta}(S):=\{\sigma\in\Delta\mid\exists\tau\in S:\sigma\subset\tau\}.
    \end{align}
    If $S$ is a subcomplex, then clearly $\mathrm{Cl}_{\Delta}(S)=S$.
    \item[(2)]The \textit{star} of a single simplex $\sigma\in\Delta$ is defined to be set of all simplices in $\Delta$ having $\sigma$ as a face, i.e.
    \begin{align}
        \mathrm{St}_{\Delta}(\sigma):=\{\tau\in\Delta\mid\sigma\subset\tau\}.
    \end{align}
    The \textit{star of $S$} is then the union of the stars of all its simplices. Note that the star is in general not a subcomplex. Therefore, one often defines the \textit{closed star}, which is the subcomplex $\mathrm{Cl}_{\Delta}(\mathrm{St}_{\Delta}(S))$. Note that some authors define the star directly in this way.
    \item[(3)]The \textit{link} of $S$ is defined to be $\mathrm{Lk}_{\Delta}(S):=\mathrm{Cl}_{\Delta}(\mathrm{St}_{\Delta}(S))\backslash\mathrm{St}_{\Delta}(\mathrm{Cl}_{\Delta}(S))$. If $\sigma\in\Delta$ is a single simplex, then its link is given by
    \begin{align}
        \mathrm{Lk}_{\Delta}(\sigma)=\{\tau\in\mathrm{Cl}_{\Delta}(\mathrm{St}_{\Delta}(\sigma))\mid \tau\cap\sigma=\emptyset\}=\{\tau\in\Delta\mid \tau\cup\sigma\in\Delta\text{ and }\tau\cap\sigma=\emptyset\}.
    \end{align}
    The link of some subset $\mathcal{S}$ is again a subcomplex of $\Delta$. Furthermore, if $\sigma\in\Delta$ is a $k$-simplex in a $d$-dimensional abstract simplicial complex, then the dimension of $\mathrm{Lk}_{\Delta}(\sigma)$ is at most $d-(k+1)$.
\end{itemize}

Figure \ref{Fig:SimplComplexes} below shows a $2$-dimensional simplicial complex $\Delta$ as well as the star, closed star and link of a vertex $v$ of $\Delta$ drawn in blue.

\begin{figure}[H]
    \centering
    \includegraphics[scale=0.8]{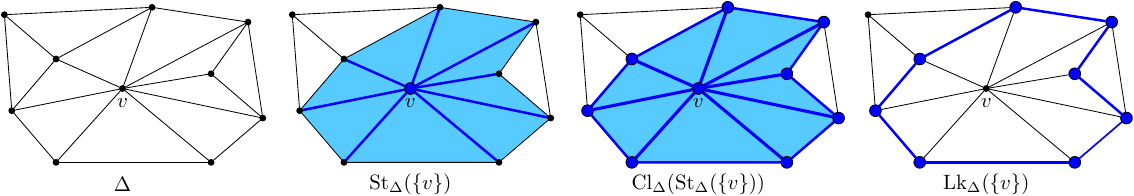}
    \caption{A $2$-dimensional simplicial complex $\Delta$ and the star $\mathrm{St}_{\Delta}(\{v\})$, closed star $\mathrm{Cl}_{\Delta}(\mathrm{St}_{\Delta}(\{v\}))$ and link $\mathrm{Lk}_{\Delta}(\{v\})$ of a vertex $v$ in $\Delta$ drawn in blue.\label{Fig:SimplComplexes}}
\end{figure}

Pseudomanifolds are topologies, which are manifolds in most of their points, but can fail to be locally-Euclidean at a finite number of isolated ``singularities''. They are defined as follows:

\begin{Definition}[Pseudomanifolds \cite{SeifertTopology}]
    Let $\Delta$ be a finite abstract $d$-dimensional simplicial complex. We call its geometric realization $\vert\Delta\vert$ a ``$d$-dimensional pseudomanifold'', if and only if the following three conditions are fulfilled:
    \begin{itemize}
        \item[(1)]$\Delta$ is ``\textit{pure}'', i.e.~every simplex $\sigma\in\Delta$ of dimension $<d$ is the face of some $d$-simplex.
        \item[(2)]$\Delta$ is ``\textit{non-branching}'', i.e.~every $(d-1)$-simplex is face of exactly one or two $d$-simplices.
        \item[(3)]$\Delta$ is ``\textit{strongly-connected}'', i.e.~for every two $d$-simplices $\sigma,\tau\in\Delta_{d}$, there is a sequence of $d$-simplices $\sigma=\sigma_{1},\sigma_{2},\dots,\sigma_{k}=\tau$ such that $\sigma_{l}\cap\sigma_{l+1}$ is a $(d-1)$-simplex $\forall l\in\{1,\dots,k-1\}$.
    \end{itemize}
\end{Definition}

The \textit{boundary of a pseudomanifold} $\Delta$, usually denoted by $\partial\Delta$, is the closure of the subset consisting of all the $(d-1)$-simplices, which are the face of only one $d$-simplex. More generally, one can define the boundary of any \textit{pure} abstract simplicial complex in this way. Furthermore, we call a pseudomanifold \textit{orientable}, if and only if there is a choice of orientation for each $d$-simplex, such that each internal $(d-1)$-simplex gets the opposite induced orientation from the two $d$-simplices to which it belongs. More generally, one can define the concept of orientability for any \textit{non-branching} abstract simplicial complex in this way. To sum up, the first two conditions in the definition allow us to talk about a boundary and about orientability. Last but not least, \textit{strongly-connectedness} tells us that a pseudomanifold can be understood as being the result of gluing $d$-simplices along their $(d-1)$-faces. \cite{SeifertTopology}
\bigskip

One can easily show that every compact, connected and triangulable manifold is a pseudomanifold. However, the converse is in general not true. As an example, pseudomanifolds may contain isolated singularities, around which they fail to be locally Euclidean. An often cited example is the \textit{pinched torus}, which is a $2$-dimensional pseudomanifold obtained by identifying two distinct points on the $2$-sphere. An important class of pseudomanifolds are ``normal pseudomanifolds'', which are defined as follows:

\begin{Definition}[Normal Pseudomanifolds]
    Let $\vert\Delta\vert$ be a $d$-dimensional pseudomanifold. We call it ``normal'' if the link of every simplex of dimension $\leq d-2$ represents a pseudomanifold.
\end{Definition}

The crucial condition in this definition is strongly-connectedness: In general, every link of a pseudomanifold is pure and non-branching, but can fail to be strongly-connected and in fact, even to be connected at all \cite{GurauColouredGFTPseudo}.

\begin{Example}
    The pinched torus is an example of a pseudomanifold, which is \textit{not} normal, since the link of its singular point consists of two distinct circles and is hence disconnected.
\end{Example}

\section{Further Details on Coloured Graphs and Crystallization Theory}\label{AppendixB:CryTheo}
In this section of the appendix, we provide some more details on crystallization theory and coloured graphs. More explicitly, we will briefly review two central theorems of crystallization theory regarding the existence of coloured graphs for manifolds. Furthermore, we will discuss a connected sum operation for graphs and its topological interpretation, which we used in the main text to show that certain types of dipole moves are proper (c.f.~Theorem \ref{DipoleProper}).

\subsection{Existence of Coloured Graphs and Crystallizations}\label{SecCryTheo}
In general, every open $(d+1)$-coloured graph represents a normal and orientable pseudomanifold with boundary, as discussed in Section \ref{SectionI:Model}. However, it is a priori not clear for which type of topologies there exists a coloured graph representing them. In this section, the goal is to review some central results from crystallization theory, which show that at least for every (PL-)manifold there is a special type of coloured graph representing it. First of all, let us introduce the notion of ``manifold crystallizations'' \cite{CTReview,Review2018,GagliardiBoundaryGraph}, the central objects of crystallization theory, which are dual to triangulations of manifolds with the smallest possible number of vertices:

\begin{Definition}[Contracted Graphs and Crystallizations]
    \begin{itemize}
        \item[]
        \item[(1)]A closed $(d+1)$-coloured graph $\mathcal{G}\in\overline{\mathfrak{G}}_{d}$ is called ``contracted'', if it admits exactly one $d$-bubble without colour $i$ for all $i\in\mathcal{C}_{d}$, i.e.~the total number of $d$-bubbles is $\mathcal{B}^{[d]}=d+1$.
        \item[(2)]Let $\mathcal{G}\in\mathfrak{G}_{d}$ be an open $(d+1)$-coloured graph with $C(\partial\mathcal{G})\in\mathbb{N}$ boundary components. Then $\mathcal{G}$ is called ``$\partial$-contracted'', if there is exactly one $d$-bubble without colour $0$ and exactly $C(\partial\mathcal{G})$ $d$-bubbles without colour $i$ for all $i\in\mathcal{C}_{d}\backslash\{0\}$, i.e.~the total number of $d$-bubbles is $\mathcal{B}^{[d]}=1+d\cdot C(\partial\mathcal{G})$.
        \item[(3)]Let $\mathcal{G}$ be a closed (resp. open) $(d+1)$-coloured graph representing a manifold $\mathcal{M}$. If $\mathcal{G}$ is contracted (resp. $\partial$-contracted), it is called a ``crystallization of $\mathcal{M}$''.
    \end{itemize}
\end{Definition}

In other words, a closed contracted graph has the smallest possible number of $d$-bubbles and hence, the corresponding simplicial complex has the smallest possible number of vertices. A $\partial$-contracted graph is a graph, for which the boundary is contracted and for which there is only a single internal $d$-bubble, or in other words, its corresponding complex has only one internal vertex and each of its boundary components has the minimal number of $d$ vertices.
\bigskip

For the case of closed manifolds, M. Pezzana was able to prove the following general existence theorem in 1974 \cite{Pezzana,Pezzana2}, which also provides the foundation of crystallization theory:

\begin{Theorem}[of Pezzana]
    Every closed and connected $d$-dimensional PL-manifold admits a crystallization representing it.
\end{Theorem}

The idea of the proof is basically to explicitly construct a contracted triangulation out of a given piecewise-linear triangulation. The full proof can be found in the original paper by M. Pezzana \cite{Pezzana} and a sketch of the proof in English, using the notion of dipole moves, can be found in \cite{GagliardiFerri}. A generalization of the above theorem for manifolds with boundary was proven by A. Cavicchioli and C. Gagliardi in 1980 \cite{GagliardiExistence} (for the case of manifolds with connected boundary) and by C. Gagliardi in 1983 \cite{GagliardiCorbodant} (general case):

\begin{Theorem}[of Cavicchioli-Gagliardi]
    \label{ThmCavGag}
    For every crystallization $\gamma$ of the boundary of some compact and connected $d$-dimensional PL-manifold $\mathcal{M}$ with (possibly disconnected) boundary, there exists a crystallization $\mathcal{G}$ of $\mathcal{M}$ whose boundary graph is (colour-isomorphic to) $\gamma$.
\end{Theorem}

\subsection{Connected Sum of Coloured Graphs}\label{Subsec:ConSum}
One way to build new manifolds out of some given manifolds is provided by performing their ``connected sum''. For two compact and connected $d$-dimensional manifolds $\mathcal{M}$ and $\mathcal{N}$ with at most one boundary component, there are two different notions one has to distinguish:

\begin{itemize}
      \item[(1)]Let us choose two closed $d$-balls $B_{1}$ and $B_{2}$ inside $\mathcal{M}$ and $\mathcal{N}$, such that they do not intersect the boundaries of $\mathcal{M}$ and $\mathcal{N}$. The ``\textit{(internal) connected sum}'' is the manifold denoted by $\mathcal{M}\#\mathcal{N}$, which is obtained by cutting out the interior of the balls from $\mathcal{M}$ and $\mathcal{N}$ and gluing\footnote{If both $\mathcal{M}$ and $\mathcal{N}$ are oriented, then we should assume in addition that the ``\textit{gluing map}'' is orientation-reversing, since the connected sum then comes equipped with a canonical orientation. Note that the connected sum in general depends on the chosen orientations, however, it does in general not depend on all the other choices as a consequence of the annulus theorem \cite{KirbyAT,Quinn}.} the two created boundary spheres together. As a consequence, it holds that $\partial (\mathcal{M}\#\mathcal{N})=(\partial\mathcal{M})\coprod (\partial\mathcal{N})$. Furthermore, note that the $d$-sphere $S^{d}$ is the neutral element of this operation, i.e.~$\mathcal{M}\# S^{d}\cong\mathcal{M}$ for all $\mathcal{M}$.
      \item[(2)]If $\partial\mathcal{M}\neq\emptyset\neq\partial\mathcal{N}$, we can choose two closed $(d-1)$-dimensional balls $B_{1}$ and $B_{2}$ inside $\partial\mathcal{M}$ and $\partial\mathcal{N}$. The ``\textit{boundary connected sum}'' is the manifold denoted by $\mathcal{M}\#_{\partial}\mathcal{N}$, which is obtained by identifying the two balls to each other. Note that it holds that $\partial (\mathcal{M}\#_{\partial} \mathcal{N})=(\partial\mathcal{M})\# (\partial\mathcal{N})$. Furthermore, note that the closed $d$-ball $B^{d}$ is the neutral element of this operation, i.e.~$\mathcal{M}\# B^{d}\cong\mathcal{M}$ for all $\mathcal{M}$.
\end{itemize}

Let us now discuss how to define the connected sum on the level of coloured graphs. To start with, let us make the following definition \cite{GagliardiBoundaryGraph,GagliardiConnectedSum}:

\begin{Definition}[Graph Connected Sum]
    Let $\mathcal{G}_{1},\mathcal{G}_{2}\in\mathfrak{G}_{d}$ be two open $(d+1)$-coloured graphs. Then, let us define the following graph: Lets take an internal vertex $v$ of $\mathcal{G}_{1}$ and an internal vertex $w$ of $\mathcal{G}_{2}$ of different types (i.e.~one black and one white). Then, we denote by $\mathcal{G}_{1}\#_{\{v,w\}}\mathcal{G}_{2}$ the open $(d+1)$-coloured graph obtained by deleting the two vertices and gluing the ``hanging'' pairs of edges together respecting their colouring. We call this graph the ``graph connected sum of $\mathcal{G}_{1}$ and $\mathcal{G}_{2}$ at $v$ and $w$''.
\end{Definition}

\begin{Remark}
    Note that if both vertices $v$ and $w$ do admit an adjacent external leg, then the procedure would produce a disconnected part containing a single edge of colour $0$ connecting two boundary vertices. In this case, we do not include this additional disconnected piece in the definition of $\mathcal{G}_{1}\#_{v,w}\mathcal{G}_{2}$, as a convention (see the example in Figure \ref{FigConSum}(a)).
\end{Remark}

The example below shows the graph connected sum of two copies of some closed $(2+1)$-coloured graph $\mathcal{G}\in\overline{\mathfrak{G}}_{2}$:

\begin{figure}[H]
    \centering
    \includegraphics[scale=1.2]{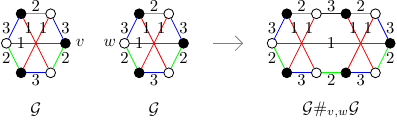}
    \caption{The graph connected sum of twice the graph $\mathcal{G}$.}
\end{figure}

The graph $\mathcal{G}$ represents the $2$-torus $T^{2}:=S^{1}\times S^{1}$. Furthermore, it is not too hard to check that the graph $\mathcal{G}\#_{v,w}\mathcal{G}$ represents the genus $g=2$ surface $\Sigma_{2}:=T^{2}\# T^{2}$, e.g.~by calculating its Euler characteristic. It turns out that the graph connected sum represents the connected sum of manifolds in more general cases. Before stating the theorem, let us introduce the following terminology: We call an internal vertex of some open $(d+1)$-coloured graph ``\textit{strictly internal}'' \cite{GagliardiConnectedSum}, if all the $d$-bubbles to which the vertex belongs, are closed. In other words, a vertex in some open coloured graph is strictly internal if and only if the corresponding $d$-simplex is not touching the boundary in the sense that all its faces of all dimensions are not contained in the boundary complex.

\begin{Theorem}
    \label{ThmConSum}
    Let $\mathcal{G}_{1},\mathcal{G}_{2}\in\mathfrak{G}_{d}$ be two open $(d+1)$-coloured graphs representing manifolds $\mathcal{M}_{1}$ and $\mathcal{M}_{2}$. Furthermore, let $v$ be an internal vertex of $\mathcal{G}_{1}$ and $w$ be an internal vertex of $\mathcal{G}_{2}$.
    \begin{itemize}
        \item[(1)]If both $v$ and $w$ admit an adjacent external leg, then $\mathcal{G}_{1}\#_{\{v,w\}}\mathcal{G}_{2}$ represents the oriented boundary connected sum $M_{1}\#_{\partial} M_{2}$.
        \item[(2)]If both vertices $v$ and $w$ do not admit an adjacent external leg and if at least one of them is strictly internal, then $\mathcal{G}_{1}\#_{\{v,w\}}\mathcal{G}_{2}$ represents the oriented internal connected sum $M_{1}\# M_{2}$.
        \item[(3)]If $v$ is an strictly internal vertex and $w$ admits an adjacent external leg, then $\mathcal{G}_{1}\#_{\{v,w\}}\mathcal{G}_{2}$ represents the manifold $(\mathcal{M}_{1}\# B^{d})\#_{\partial}\mathcal{M}_{2}$, where $B^{d}$ denotes the closed $d$-ball.
    \end{itemize}
\end{Theorem}

\begin{proof}
    The detailed proof can be found in \cite{GagliardiConnectedSum}. As an example, in case (2), we delete an internal $d$-simplex in one of the triangulations and another $d$-simplex (possibly touching the boundary with some of its faces of dimension $<d-1$) in another complex. Now, since a $d$-simplex represents a $d$-ball, removing these simplices results into removing balls inside the corresponding manifolds. Furthermore, connecting the hanging pair of edges of the coloured graph obtained by deleting these two vertices precisely corresponds to gluing the created boundary $d$-spheres together. Taking the two vertices of different types ensures that the gluing map is orientation-reversing.
\end{proof}

The figure below shows two examples of the previous theorem. Figure (a) shows the boundary connected sum of two disks, which is again a disk, and figure (b) shows an example of the (internal) connected sum of two disks, which is homeomorphic to the cylinder $S^{1}\times [0,1]$, i.e.~the unique (up to homeomorphism) surface with genus zero and two boundary components.

\begin{figure}[H]
    \centering
    \includegraphics[scale=1.1]{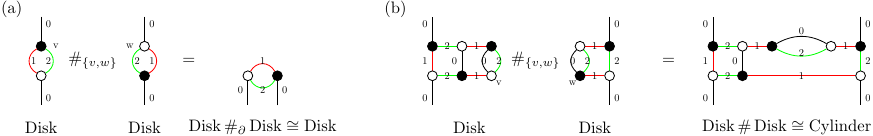}
    \caption{Two examples of graph-connected sums of open $(2+1)$-coloured graphs representing the disk and their geometric realizations according to Theorem \ref{ThmConSum}. \label{FigConSum}}
\end{figure}

The fact that these graphs indeed describe disks and cylinders can be seen by explicitly calculating their Euler characteristic as well as the number of boundary components. Let us collect two immediate consequences of the theorem above, concerning the mixed case of the graph-connected sum of a closed graph with an open graph:

\begin{Corollary}
    \label{CorConSum}
    Let $\mathcal{G}_{1}\in\overline{\mathfrak{G}}_{d}$ be a closed $(d+1)$-coloured graph representing a manifold $\mathcal{M}_{1}$ and $\mathcal{G}_{2}\in\mathfrak{G}_{d}$ be an open $(d+1)$-coloured graph representing a manifold $\mathcal{M}_{2}$. Furthermore, let $v$ be a vertex of $\mathcal{G}_{1}$ and $w$ be an internal vertex of $\mathcal{G}_{2}$. Then:
    \begin{itemize}
        \item[(1)]If $w$ is an internal vertex, which does not admit an adjacent external leg, then $\mathcal{G}_{1}\#_{\{v,w\}}\mathcal{G}_{2}$ represents the oriented internal connected sum $M_{1}\# M_{2}$.
        \item[(2)]If $\mathcal{M}_{1}\cong S^{d}$ and if $w$ is an internal vertex, which admits an adjacent external leg, then $\mathcal{G}_{1}\#_{\{v,w\}}\mathcal{G}_{2}$ represents the manifold $\mathcal{M}_{2}$.
    \end{itemize}
\end{Corollary}

\begin{proof}
    Claim (1) follows directly from Theorem \ref{ThmConSum}(2), since in a closed graph every vertex is strictly internal. For claim (2), recall that the sphere is the neutral element of $\#$ whereas the ball is the neutral element of $\#_{\partial}$ and hence, by \ref{ThmConSum}(3), $\mathcal{G}_{1}\#_{\{v,w\}}\mathcal{G}_{2}$ represents the manifold $(S^{d}\# B^{d})\#_{\partial}\mathcal{M}_{2}\cong B^{d}\#_{\partial}\mathcal{M}_{2}\cong\mathcal{M}_{2}$.
\end{proof}
\section{Construction of Graphs Representing the Solid Torus}
\label{AppendixC:SolidTorusGraphs}
The goal of this section is to construct open $(3+1)$-coloured graphs representing the solid torus $D^{2}\times S^{1}$, where $D^{2}$ denotes the closed $2$-ball (=disk). In general, a $d$-dimensional simplicial complex $\Delta$ representing a manifold with boundary admits a coloured graph contained in $\mathfrak{G}_{d}$ representing it if and only if has the following two properties:\footnote{For a simplicial complex $\Delta$ triangulating a \textit{pseudo}manifold, we have to assume in addition that the (disjoint) star of every vertex is strongly-connected, since otherwise, it can happen that the complex $\Delta_{\mathcal{G}(\Delta)}$ obtained from the coloured graph $\mathcal{G}(\Delta)$ dual to $\Delta$ \textit{does not coincide} with the original complex $\Delta$, because we loose some information regarding possible pinching effects. As an example, take the complex representing the pinched torus drawn in figure \ref{MultRes} (with the vertices $v$ and $w$ identified). \cite[p.198/99]{CasaliGrasselli89}.}

\begin{itemize}
    \item[(1)]It admits a $(d+1)$-vertex-colouring, i.e.~a map $\gamma:\mathcal{V}\to\{0,\dots,d\}$, where $\mathcal{V}$ denotes the set of vertices of $\Delta$, which is injective on every $d$-simplex. Equivalently, this defines a proper face-colouring of the complex by assigning to each $(d-1)$-simplex of a $d$-simplex the colour of the vertex on the opposite side.
    \item[(2)]None of the vertices on the boundary complex has colour $0$. Equivalently in the face-coloured picture, this means that all the $(d-1)$-simplices of the boundary complex have the same colour $0$.
\end{itemize}

If such a complex represents an \textit{orientable} manifold, then it will automatically be bipartite in the sense that there are two types of $d$-simplices and only $d$-simplices of different types share a common $(d-1)$-face, because bipartiteness and orientability are equivalent for coloured graphs \cite{GagliardiCorbodant, GagliardiBoundaryGraph} and hence also for colourable complexes.
\bigskip

Let us now construct a family of open $(3+1)$-coloured graphs representing the solid torus. As a starting point, we consider the family of discretizations constructed in \cite{TorusPR2,ChristopheThesis}. For this, consider the following general discretization of the solid cylinder $D^{2}\times [0,1]$ (topologically a closed $3$-ball):

\begin{figure}[H]
    \centering
    \includegraphics[scale=0.8]{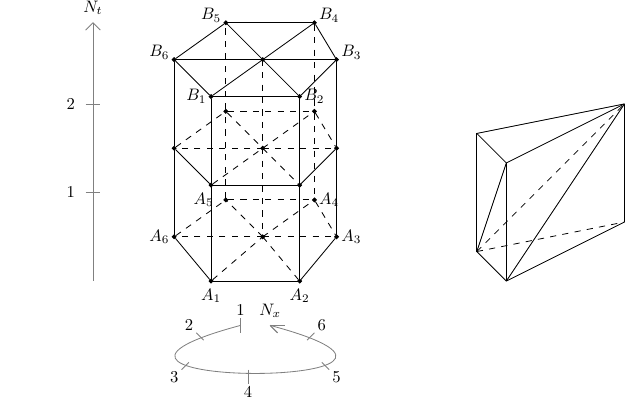}
    \caption{A cellular decomposition of the solid cylinder consisting of prisms characterized by the number of horizontal and vertical layers (l.h.s.). In order to obtain a triangulation, we have to discretize every prism by tetrahedra. The triangulation of a prism with the minimal number of tetrahedra is drawn on the right-hand side. \label{TrianSolidTorus}}
\end{figure}

The cellular complex is characterized by two natural numbers: The number of vertical layers of prisms denoted by $N_{t}\in\mathbb{N}$, as well as the number of horizontal layers, i.e.~the numbers of prisms in each horizontal slice, which we denote by $N_{x}\in\mathbb{N}$. In order to obtain a cellular decomposition of the solid torus, we have to identify the top and bottom of the complex drawn above. Note that there is some freedom in doing so, since the gluing can be done in several ways. Hence, we introduce the ``\textit{twist parameter}'' $N_{\gamma}\in\{0,\dots,N_{x}-1\}$ defined by the equation
\begin{align}
    A_{i}\doteq B_{i+N_\gamma}\hspace{1cm}\forall i\in\{1,\dots,N_{x}\},
\end{align}
where the indices in this equation have to be understood as being cyclic, e.g.~$N_{x}+i=i$, and where the ``\textit{twist angle}'' $\gamma$, corresponding to a discrete Dehn twist \cite{DehnTwist}, is defined by
\begin{align}
    \gamma:=2\pi\frac{N_{\gamma}}{N_{x}}.
\end{align}
To sum up, we have constructed general cellular decompositions of the solid torus characterised by the three numbers $N_{x},N_{t}\in\mathbb{N}$ and $N_{\gamma}\in\{0,\dots,N_{x}-1\}$.
\bigskip

In order to turn the cellular complex of the solid torus into a simplicial one, we have to triangulate each prism, as shown on the right-hand side of figure \ref{TrianSolidTorus} above. Now, it is clear that we cannot just triangulate each prism in the complex in precisely the same way, since if we glue two such prisms horizontally, the resulting complex is not proper vertex colourable. A closer analysis reveals that we need at least two vertical layers and at least two horizontal layers, where the prisms in each layer are triangulated symmetrically to each other. In other words, a colourable simplicial complex of the type introduced above consists of basic building blocks with four prisms, triangulated and coloured as shown in figure \ref{TorusDisc2} below.

\begin{figure}[H]
    \centering
    \includegraphics[scale=0.8]{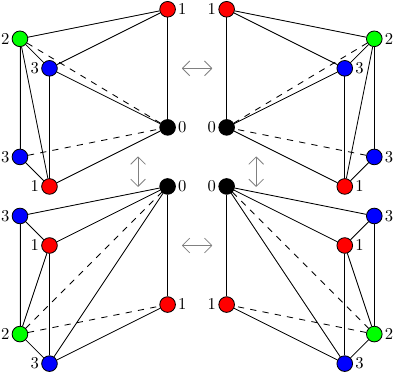}
    \caption{Basic building block of a coloured and bipartite simplicial complex of the solid cylinder consisting of four prisms. The gluing of prisms is indicated by the grey arrows.\label{TorusDisc2}}
\end{figure}

Using this triangulation, we finally arrive at general proper face-coloured and bipartite triangulations of the solid torus characterized by three numbers $N_{x},N_{t}\in 2\mathbb{N}$ and $N_{\gamma}\in\{0,2,4,\dots,N_{x}-2\}$. Note that, due to the colouring, only even twists are possible, since we are only allowed to glue basic building block consisting of two vertical and two horizontal layers together. Using the figure above, it is straightforward to draw the coloured graph corresponding to a basic building block, i.e.~see figure \ref{TorusDisc3} below.

\begin{figure}[H]
    \centering
    \includegraphics[scale=0.8]{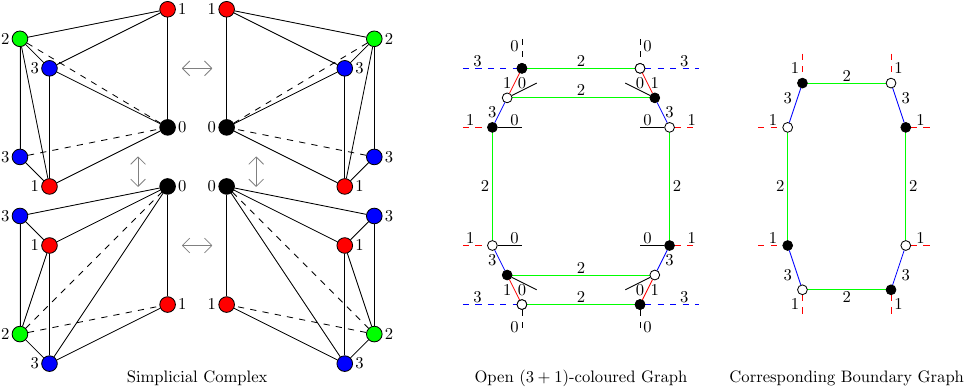}
    \caption{Basic building block as a (vertex)-coloured simplicial complex and as an open $(3+1)$-coloured graph with its corresponding boundary graph. Dotted edges are those to which we glue further building blocks. \label{TorusDisc3}}
\end{figure}

The dotted lines in the figure above are those edges, to which we glue further building blocks. Each building block has in total eight faces living on the boundary and hence, each part of the graph dual to such a building block has eight external legs of colour $0$. To sum up, we have constructed a family of open $(3+1)$-coloured graphs belonging to $\mathfrak{G}_{3}$, which are dual to the solid torus. Such a graph is labelled and uniquely determined by the three parameters $N_{x},N_{t}\in 2\mathbb{N}$ and $N_{\gamma}\in\{0,2,4,\dots,N_{x}-2\}$ and has the following general form:

\begin{figure}[H]
    \centering
    \includegraphics[scale=0.9]{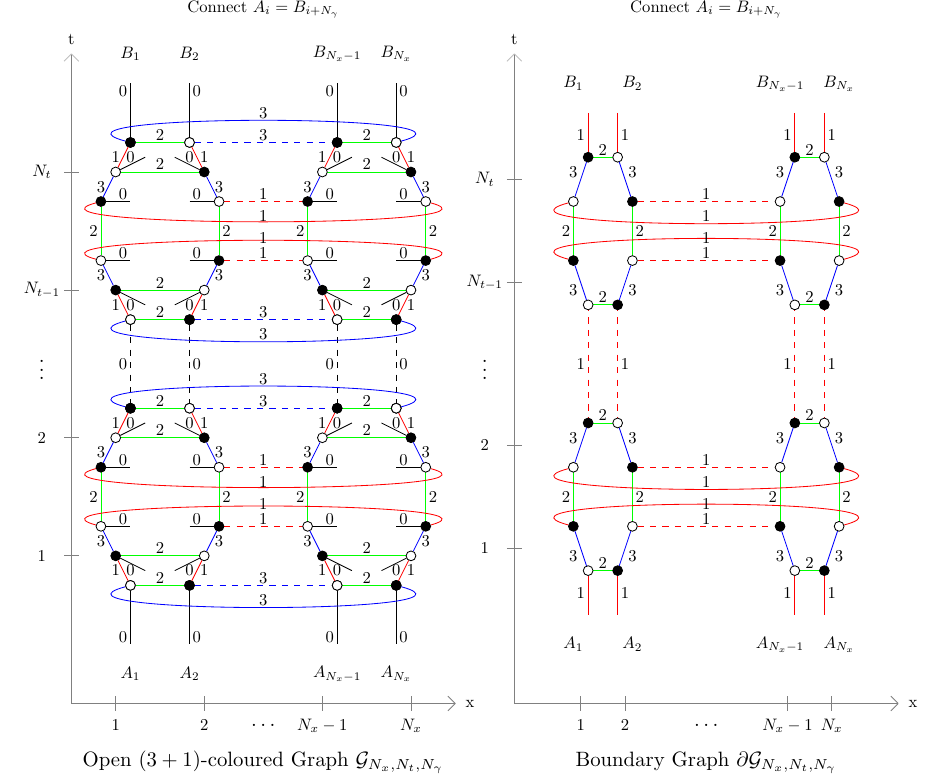}
    \caption{A family of open $(3+1)$-coloured graphs $\mathcal{G}_{N_{x},N_{t},N_{\gamma}}\in\mathfrak{G}_{3}$ labelled by three parameters $N_{x},N_{t}\in 2\mathbb{N}$ and $N_{\gamma}\in\{0,2,4,\dots,N_{x}-2\}$, each representing the solid torus, as well as their boundary graphs $\partial\mathcal{G}_{N_{x},N_{t},N_{\gamma}}\in\overline{\mathfrak{G}}_{2}$.}
\end{figure}

It is straightforward to count the number of $k$-bubbles of these graphs for $k\in\{0,1,2,3\}$, or equivalently, the number of $(3-k)$-simplices of the corresponding simplicial complex:

\begin{itemize}
    \item The number of $0$-bubbles, i.e.~internal vertices of the graph $\mathcal{G}_{N_{x},N_{t},N_{\gamma}}$, or equivalently, the number of tetrahedra of the complex $\Delta_{\mathcal{G}_{N_{x},N_{t},N_{\gamma}}}$, is given by $\mathcal{B}^{[0]}=3N_{x}N_{t}$.
    \item The number of $1$-bubbles, i.e.~edges of the graph $\mathcal{G}_{N_{x},N_{t},N_{\gamma}}$, or equivalently, the number of triangles of the complex $\Delta_{\mathcal{G}_{N_{x},N_{t},N_{\gamma}}}$, is given by $\mathcal{B}^{[1]}=7N_{x}N_{t}$ from which $2N_{x}N_{t}$ are external legs, i.e. triangles living purely on the boundary complex $\partial\Delta_{\mathcal{G}_{N_{x},N_{t},N_{\gamma}}}$.
    \item The number of $2$-bubbles, i.e.~faces of the graph $\mathcal{G}_{N_{x},N_{t},N_{\gamma}}$, or equivalently, the number of edges of the complex $\Delta_{\mathcal{G}_{N_{x},N_{t},N_{\gamma}}}$, is given by $\mathcal{B}^{[2]}=N_{t}+5N_{t}N_{x}$ from which $3N_{t}N_{x}$ are non-cyclic faces, i.e.~edges living purely on the boundary complex $\partial\Delta_{\mathcal{G}_{N_{x},N_{t},N_{\gamma}}}$.
    \item The number of $3$-bubbles of the graph $\mathcal{G}_{N_{x},N_{t},N_{\gamma}}$, or equivalently, the number of vertices of the complex $\Delta_{\mathcal{G}_{N_{x},N_{t},N_{\gamma}}}$, is given by $\mathcal{B}^{[3]}=N_{t}+N_{t}N_{x}$ from which $N_{t}N_{x}$ are open $3$-bubbles, i.e.~vertices living purely on the boundary complex $\partial\Delta_{\mathcal{G}_{N_{x},N_{t},N_{\gamma}}}$.
\end{itemize}

As a quick consistency check, let us calculate the Euler characteristic of the simplicial complex dual to the open graph $\mathcal{G}_{N_{x},N_{t},N_{\gamma}}$ as well as of the boundary complex, which gives
\begin{align}
    \chi(\Delta_{\mathcal{G}_{N_{x},N_{t},N_{\gamma}}})=&N_{t}+N_{t}N_{x}-(N_{t}+5N_{t}N_{x})+7N_{t}N_{x}-3N_{t}N_{x}=0\\\chi(\partial\Delta_{\mathcal{G}_{N_{x},N_{t},N_{\gamma}}})=&\chi(\Delta_{\partial\mathcal{G}_{N_{x},N_{t},N_{\gamma}}})=N_{t}N_{x}-3N_{t}N_{x}+2N_{t}N_{x}=0
\end{align}
as it should.
\bigskip

As an example of the family of graphs constructed above, let us consider the simplest graph, i.e.~the graph $\mathcal{G}_{N_{x},N_{t},N_{\gamma}}$ with $N_{x}=N_{t}=2$ and $N_{\gamma}=0$. The graph together with its boundary graph is drawn in figure \ref{SimpleTorusGraph} below.

\begin{figure}[H]
    \centering
    \includegraphics[scale=1.1]{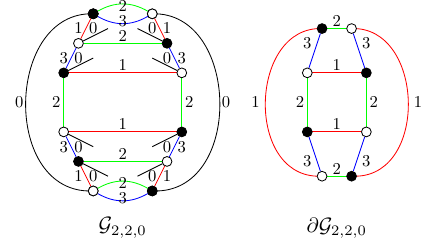}
    \caption{The solid torus graph with $N_{x}=N_{t}=2$ and $N_{\gamma}=0$ together with its boundary graph.\label{SimpleTorusGraph}}
\end{figure}

This graph is clearly not a core graph, since there are two $3$-bubbles of colour $023$, from which one represents the $2$-sphere. Hence, in order to obtain a core graph, we have to contract one internal proper $1$-dipole. There are in total four choices of edges of colour $1$, which can be contracted. The core graph obtained by contracting the upper left one looks as follows:

\begin{figure}[H]
    \centering
    \includegraphics[scale=1]{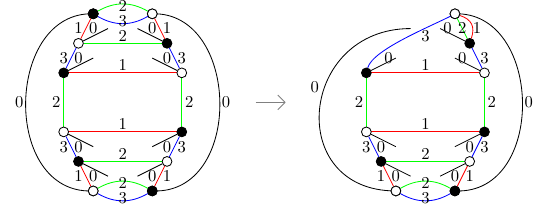}
    \caption{A core graph representing the solid torus obtained by an internal proper $1$-dipole move from $\mathcal{G}_{2,2,0}$.\label{SolidTorusCoreGraph}}
\end{figure}

The boundary graph of this core has in total eight vertices. As discussed in the main text (Subsection \ref{SubSec:Torus}), the simplest possible closed $(2+1)$-coloured graph representing a $2$-torus has only six vertices and looks as follows:

\begin{figure}[H]
    \centering
    \includegraphics[scale=1]{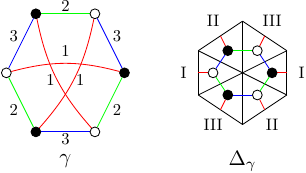}
    \caption{The smallest closed $(2+1)$-coloured graph representing the $2$-torus (l.h.s.) and its corresponding simplicial complex (r.h.s.), where the gluing of edges is as indicated by the Roman numbers.\label{TorusBoundary}}
\end{figure}

Note that this boundary graph can be obtained by performing an internal $1$-dipole move of colour $3$ within the boundary graph drawn in figure \ref{SimpleTorusGraph}. In order to obtain core graphs representing the solid torus with the simplest $2$-torus graph as its boundary graph, we have to perform a non-internal proper $1$-dipole move of colour $3$ in the core graph drawn in figure \ref{SolidTorusCoreGraph} above. There are in total four possibilities to do so. Taking the edge of colour $3$ on the top right, the result looks as follows:

\begin{figure}[H]
    \centering
    \includegraphics[scale=1]{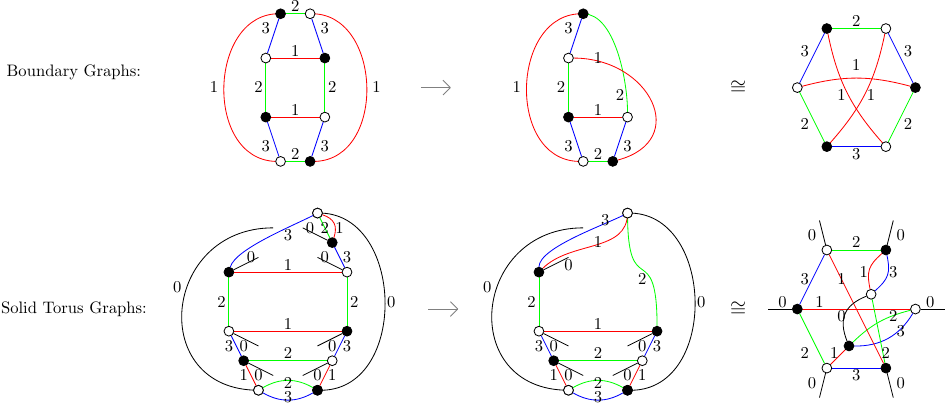}
    \caption{A core graph with the simplest torus boundary graph obtained by performing a non-internal proper $1$-dipole move in the core drawn on the r.h.s.~in figure \ref{SolidTorusCoreGraph}. The boundary graphs in each step are drawn in the top line.}
\end{figure}

Yet another core graph representing a solid torus with boundary given by the smallest $2$-torus graph can be obtained by performing a sequence of internal proper dipole moves within the core graph drawn above. An example is drawn in figure \ref{SolidTorus3} below.

\begin{figure}[H]
    \centering
    \includegraphics[clip,trim=0 1cm 0 0,scale=0.9]{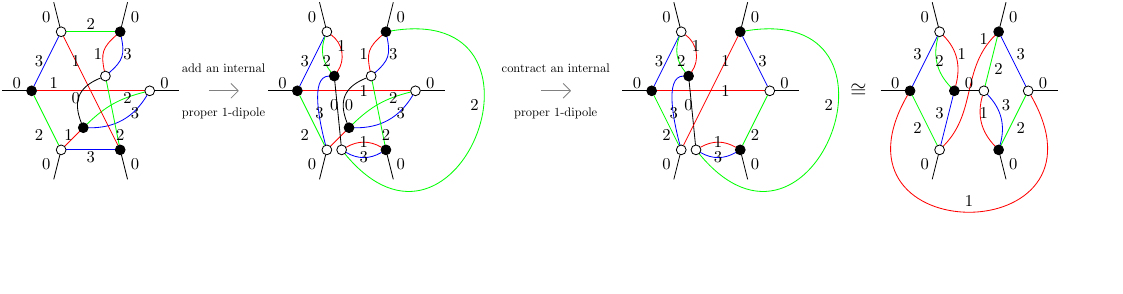}
    \vspace*{-0.2cm}
    \caption{A core graph obtained by firstly adding an internal proper $1$-dipole of colour $0$ and by cancelling an internal proper $1$-dipole of colour $0$ afterwards. The boundary graph is left untouched, since all the dipoles are internal ones.\label{SolidTorus3}}
\end{figure}

To sum up, we have found two core graphs representing the solid torus $D^{2}\times S^{1}$ (see figure \ref{SolidTorus4} below), whose boundary graphs are given by the simplest closed $(2+1)$-coloured graph representing a $2$-torus (figure \ref{TorusBoundary}):\footnote{The graph on the left-hand side can also be found as a special case in Example 11 of \cite{HandleBodiesCT}. However, the authors of this paper give no geometric construction of the corresponding complex, but rather argue that it has to represent the solid torus by its properties.}

\begin{figure}[H]
    \centering
    \includegraphics[clip,trim=0 0.5cm 0 0,scale=0.9]{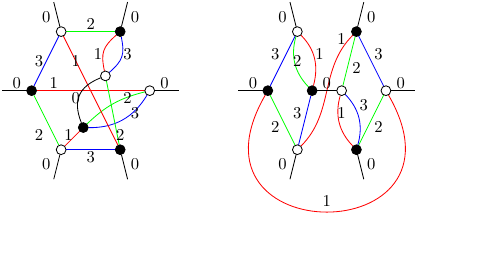}
    \vspace*{-0.2cm}
    \caption{Two core graphs representing the solid torus with boundary graph given by figure \ref{TorusBoundary}.\label{SolidTorus4}}
\end{figure}

Note also that there are no open $(3+1)$-coloured graphs contained in $\mathfrak{G}_{3}$ representing the solid torus with less than eight internal vertices, because the smallest torus boundary graph has six vertices and the smallest open graph matching this boundary graph, which is the graph obtained by adding an external leg to all the vertices of the boundary graph, is clearly a pseudomanifold (c.f.~Subsection \ref{subsec:SMG}). In the simplicial picture, this means that the smallest proper-colourable simplicial complex triangulating the solid torus with the property that all its boundary faces have the same colour consists of at least eight tetrahedra and six boundary faces. In other words, the two graphs drawn above are examples of graphs contained in $\mathfrak{G}_{3}$ representing the solid torus with the minimal possible number of vertices. Note that this observation also matches with Corollary \ref{Cor:ManifoldVertices}.
\bigskip

As a last remark, let us note that we can use the connected sum operation defined in Subsection \ref{Subsec:ConSum} as well as the solid torus graphs drawn above in order to obtain open $(3+1)$-coloured graphs representing a handlebody of genus $g$ whose boundary graph is given by the smallest closed $(2+1)$-coloured graph representing a genus $g$-surface, as shown in figure \ref{fig:HandleBody} below.

\begin{figure}[H]
    \centering
    \includegraphics[clip,trim=0 0cm 2cm 0,scale=0.9]{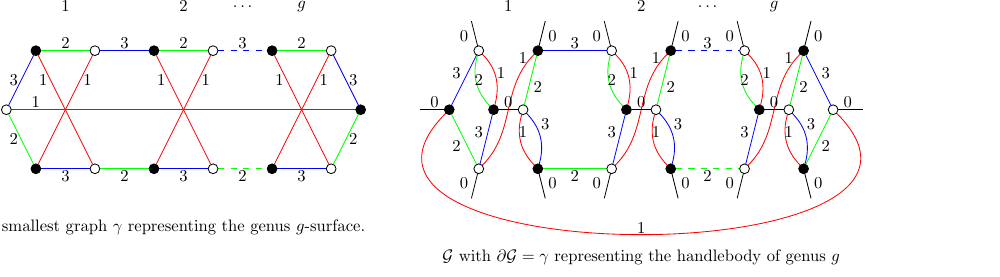}
    \vspace*{-0.2cm}
    \caption{The smallest closed $(2+1)$-coloured graph $\gamma$ representing a genus $g$-surface together with an open $(3+1)$-coloured graph $\mathcal{G}$ representing the handlebody of genus $g$.\label{fig:HandleBody}}
\end{figure}

The boundary graph $\gamma$ has in total $\vert\mathcal{V}_{\gamma}\vert=2+4g$ vertices and the corresponding handlebody graph has $\vert\mathcal{V}_{\mathcal{G},\mathrm{int}}\vert=\vert\mathcal{V}_{\gamma}\vert+2g$ internal vertices. According to Corollary \ref{Cor:ManifoldVertices}, we conclude that the graph $\mathcal{G}$ is an example of a graph representing a manifold with the minimal number of internal vertices among all the possible open $(3+1)$-coloured graphs in $\mathfrak{G}_{3}$ with boundary $\gamma$. Furthermore, this matches the result obtained in Example 11 of \cite{HandleBodiesCT} regarding the minimal possible number of internal vertices of graphs representing handlebodies.
\addcontentsline{toc}{section}{\hspace{15pt}References}
\bibliographystyle{ieeetr}
\bibliography{Bibliography.bib}

\begin{thebibliography}{100}

\bibitem{BekensteinBH}
J.~D. Bekenstein, ``{Black holes and entropy},'' {\em Physical Review D},
  vol.~7, no.~8, p.~2333–2346, 1973.

\bibitem{HawkingsBH}
S.~W. Hawkings, ``{Particle Creation by Black Holes},'' {\em Communications in
  Mathematical Physics}, vol.~43, no.~3, pp.~199--220, 1975.

\bibitem{Susskind}
L.~Susskind, ``{The World as a Hologram},'' {\em Journal of Mathematical
  Physics}, vol.~36, no.~11, p.~6377–6396, 1995.
\newblock \href{https://arxiv.org/abs/hep-th/9409089v2}{arXiv:hep-th/9409089}.

\bibitem{Hooft}
G.~'t~Hooft, ``{Dimensional Reduction in Quantum Gravity},'' {\em Conference
  Proceedings C}, pp.~284--296, 1993.
\newblock \href{https://arxiv.org/abs/gr-qc/9310026v2}{arXiv:gr-qc/9310026}.

\bibitem{AdSCFT}
J.~M. Maldacena, ``{The Large N Limit of Superconformal Field Theories and
  Supergravity},'' {\em International Journal of Theoretical Physics}, vol.~38,
  no.~4, p.~1113–1133, 1999.
\newblock \href{https://arxiv.org/abs/hep-th/9711200v3}{arXiv:hep-th/9711200}.

\bibitem{ChernSimonsGravity}
A.~Achúcarro and P.~K. Townsend, ``{A Chern-Simons action for
  three-dimensional anti-de Sitter supergravity theories},'' {\em {Physics
  Letters B}}, vol.~180, no.~1-2, 1986.

\bibitem{HorowitzBFTheory}
G.~Horowith, ``{Exactly soluable diffeomorphism invariant theories},'' {\em
  {Communications in Mathematical Physics}}, vol.~125, no.~3, pp.~417--437,
  1989.

\bibitem{PonzanoReggeModel}
G.~Ponzano and T.~Regge, ``{Semiclassical limit of Racah coefficients},'' in
  {\em {Spectroscopic and Group Theoretical Methods in Physics}} (F.~Bloch,
  ed.), pp.~1--58, Amsterdam: Elsevier, 1968.

\bibitem{BarrettPonzanoRegge}
J.~W. Barrett and I.~{Naish-Guzman}, ``{The Ponzano-Regge model},'' {\em
  {Classical and Quantum Gravity}}, vol.~26, no.~15, p.~5014, 2009.
\newblock \href{https://arxiv.org/abs/0803.3319v2}{arXiv:0803.3319}.

\bibitem{FreidelPonzanoRegge1}
L.~Freidel and D.~Louapre, ``{Ponzano-Regge model revisited I: Gauge fixing,
  observables and interacting spinning particles},'' {\em {Classical and
  Quantum Gravity}}, vol.~21, no.~24, pp.~5685--5726, 2004.
\newblock \href{https://arxiv.org/abs/hep-th/0401076v1}{arXiv:hep-th/0401076}.

\bibitem{FreidelPonzanoRegge2}
L.~Freidel and D.~Louapre, ``{Ponzano-Regge model revisited II: Equivalence
  with Chern-Simons},'' 2004.
\newblock \href{https://arxiv.org/abs/gr-qc/0410141v3}{arXiv:gr-qc/0410141}.

\bibitem{FreidelPonzanoRegge3}
L.~Freidel and E.~R. Livine, ``{Ponzano-Regge model revisited III: Feynman
  diagrams and Effective field theory},'' {\em {Classical and Quantum
  Gravity}}, vol.~23, no.~6, pp.~2021--2061, 2006.
\newblock \href{https://arxiv.org/abs/hep-th/0502106v2}{arXiv:hep-th/0502106}.

\bibitem{BaezBFTheory}
J.~C. Baez, ``{An Introduction to Spin Foam Models of Quantum Gravity and BF
  Theory},'' in {\em {Geometry and Quantum Physics. Proceedings of the 38.
  Internationale Universitätswochen für Kern- und Teilchenphysik, Schladming,
  Austria, January 9–16, 1999}} (H.~Gausterer, H.~Grosse, and P.~L., eds.),
  vol.~543 of {\em {Lecture Notes in Physics}}, pp.~25--93, Berlin, Heidelberg:
  Springer, 2000.
\newblock \href{https://arxiv.org/abs/gr-qc/9905087v1}{arXiv:gr-qc/9905087}.

\bibitem{PerezSM}
A.~Perez, ``{Spin Foam Models for Quantum Gravity},'' {\em {Classical and
  Quantum Gravity}}, vol.~20, no.~6, 2003.
\newblock \href{https://arxiv.org/abs/gr-qc/0301113v2}{arXiv:gr-qc/0301113}.

\bibitem{FreidelPRDisc}
L.~Freidel and K.~Krasnov, ``{Spin Foam Models and the Classical Action
  Principle},'' {\em Advances in Theoretical and Mathematical Physics}, vol.~2,
  no.~6, pp.~1183--1247, 1998.
\newblock \href{https://arxiv.org/abs/hep-th/9807092v2}{arXiv:hep-th/9807092}.

\bibitem{NouiPerez}
K.~Noui and A.~Perez, ``{Three-dimensional loop quantum gravity: Physical
  scalar product and spin foam models},'' {\em {Classical and Quantum
  Gravity}}, vol.~22, no.~9, pp.~1739--1761, 2005.
\newblock \href{https://arxiv.org/abs/gr-qc/0402110}{arXiv:gr-qc/0402110}.

\bibitem{OoguriSasakura}
H.~Ooguri and N.~Sasakura, ``{Discrete and Continuum Approaches to
  Three-Dimensional Quantum Gravity},'' {\em {Modern Physics Letters A}},
  vol.~6, no.~39, pp.~3591--3600, 1991.
\newblock \href{https://arxiv.org/abs/hep-th/9108006}{arXiv:hep-th/9108006}.

\bibitem{TuraevViroModel}
V.~G. Turaev and O.~Y. Viro, ``{State sum invariants of 3-manifolds and quantum
  6j-symbols},'' {\em {Topology}}, vol.~31, no.~4, p.~865–902, 1992.

\bibitem{TuraevReshetikhin}
V.~G. Turaev and N.~Reshetikhin, ``{Invariants of 3-manifolds via link
  polynomials and quantum groups},'' {\em {Inventiones mathematicae}},
  vol.~103, p.~547–597, 1991.

\bibitem{RobertsTuraevViro}
J.~Roberts, ``{Skein theory and Turaev-Viro invariants},'' {\em {Topology}},
  vol.~34, no.~4, pp.~771--787, 1995.

\bibitem{Witten1}
E.~Witten, ``{2+1 dimensional gravity as an exactly soluble system},'' {\em
  {Nuclear Physics B}}, vol.~311, no.~1, pp.~46--78, 1988.

\bibitem{Witten2}
E.~Witten, ``{Quantum field theory and the Jones polynomial},'' {\em
  {Communications in Mathematical Physics}}, vol.~121, no.~3, p.~351–399,
  1989.

\bibitem{Witten3}
E.~Witten, ``{Three-Dimensional Gravity Revisited},'' 2007.
\newblock \href{https://arxiv.org/abs/0706.3359v1}{arXiv:0706.3359}.

\bibitem{PRBallIsing1}
B.~Dittrich and J.~Hnybida, ``{Ising Model from Intertwiners},'' 2013.
\newblock \href{https://arxiv.org/abs/1312.5646v3}{arXiv:1312.5646}.

\bibitem{PRBallIsing2}
V.~Bonzom, F.~Costantino, and E.~R. Livine, ``{Duality between Spin networks
  and the 2D Ising model},'' {\em Communications in Mathematical Physics},
  vol.~344, no.~2, pp.~531--579, 2016.
\newblock \href{https://arxiv.org/abs/1504.02822v1}{arXiv:1504.02822}.

\bibitem{TorusPR1}
B.~Dittrich, C.~Goeller, E.~R. Livine, and A.~Riello, ``{Quasi-local
  holographic dualities in non-perturbative 3d quantum gravity},'' {\em
  Classical and Quantum Gravity}, vol.~35, no.~13, 2018.
\newblock \href{https://arxiv.org/abs/1803.02759v1}{arXiv:1803.02759}.

\bibitem{TorusPR2}
B.~Dittrich, C.~Goeller, E.~R. Livine, and A.~Riello, ``{Quasi-local
  holographic dualities in non-perturbative 3d quantum gravity I - Convergence
  of multiple approaches and examples of Ponzano-Regge statistical duals},''
  {\em Nuclear Physics B}, vol.~938, pp.~807--877, 2019.
\newblock \href{https://arxiv.org/abs/1710.04202v1}{arXiv:1710.04202}.

\bibitem{TorusPR3}
B.~Dittrich, C.~Goeller, E.~R. Livine, and A.~Riello, ``{Quasi-local
  holographic dualities in non-perturbative 3d quantum gravity II - From
  coherent quantum boundaries to BMS3 characters},'' {\em Nuclear Physics B},
  vol.~938, pp.~878--934, 2019.
\newblock \href{https://arxiv.org/abs/1710.04237v1}{arXiv:1710.04237}.

\bibitem{TorusPR4}
C.~Goeller, E.~R. Livine, and A.~Riello, ``{Non-Perturbative 3D Quantum
  Gravity: Quantum Boundary States and Exact Partition Function},'' {\em
  General Relativity and Gravitation}, vol.~52, no.~3, 2020.
\newblock \href{https://arxiv.org/abs/1912.01968v1}{arXiv:1912.01968}.

\bibitem{ChristopheThesis}
C.~Goeller, {\em {Quasi-Local 3D Quantum Gravity: Exact Amplitude and
  Holography}}.
\newblock PhD thesis, {Université de Lyon}, May 2019.
\newblock \href{https://arxiv.org/abs/2005.09985v1}{arXiv:2005.09985}.

\bibitem{BMS1}
H.~Bondi, M.~G. J. v.~d. Burg, and A.~W.~K. Metzner, ``{Gravitational waves in
  general relativity, VII. Waves from axi-symmetric isolated system},'' {\em
  Proceedings of the Royal Society of London A}, vol.~269, no.~1336,
  pp.~21--52, 1962.

\bibitem{BMS2}
R.~K. Sachs, ``{Asymptotic symmetries in gravitational theory},'' {\em Physical
  Review}, vol.~128, no.~6, p.~2851–2864, 1962.

\bibitem{TopChangeHorowitz}
G.~T. Horowitz, ``{Topology change in classical and quantum gravity},'' {\em
  Classical and Quantum Gravity}, vol.~8, no.~4, pp.~587--601, 1991.

\bibitem{StringTheoryTopChange}
P.~Aspinwall, B.~R. Greene, and D.~R. Morrison, ``{Calabi-Yau Moduli Space,
  Mirror Manifolds and Spacetime Topology Change in String Theory},'' {\em
  Nuclear Physics B}, vol.~416, no.~2, pp.~414--480, 1994.
\newblock \href{https://arxiv.org/abs/hep-th/9309097}{arXiv:hep-th/9309097}.

\bibitem{Geons}
H.~F. Dowker and R.~D. Sorkin, ``{A spin-statistics theorem for certain
  topological geons},'' {\em Classical and Quantum Gravity}, vol.~15, no.~5,
  pp.~1153--1167, 1998.
\newblock \href{https://arxiv.org/abs/gr-qc/9609064}{arXiv:gr-qc/9609064}.

\bibitem{FreidelGFTOverview}
L.~Freidel, ``{Group Field Theory: An overview},'' {\em {International Journal
  of Theoretical Physics}}, vol.~44, no.~10, p.~1769–1783, 2005.
\newblock \href{https://arxiv.org/abs/hep-th/0505016v1}{arXiv:hep-th/0505016}.

\bibitem{OritiGFTApproach}
D.~Oriti, ``{The group field theory approach to quantum gravity},'' in {\em
  {Approaches to Quantum Gravity. Toward a New Understanding of Space, Time and
  Matter}} (D.~Oriti, ed.), pp.~310--331, Cambridge, United Kingdom: Cambridge
  University Press, 2009.
\newblock \href{https://arxiv.org/abs/gr-qc/0607032v3}{arXiv:gr-qc/0607032}.

\bibitem{OritiMicroscopicDynamics}
D.~Oriti, ``{The microscopic dynamics of quantum space as a group field
  theory},'' in {\em {Foundations of Space and Time. Reflections on Quantum
  Gravity}} (J.~Murugan, A.~Weltman, and G.~F.~R. Ellis, eds.), pp.~257--320,
  Cambridge, United Kingdom: Cambridge University Press, 2012.
\newblock \href{https://arxiv.org/abs/1110.5606v1}{arXiv:1110.5606}.

\bibitem{MatrixModels1}
F.~David, ``{A model of random surfaces with non-trivial critical behaviour},''
  {\em {Nuclear Physics B}}, vol.~257, no.~1, pp.~543--576, 1985.

\bibitem{MatrixModels2}
P.~Ginsparg, ``{Matrix models of 2d gravity},'' 1991.
\newblock \href{https://arxiv.org/abs/hep-th/9112013v1}{arXiv:hep-th/9112013}.

\bibitem{TensorModels1}
J.~Ambjørn, B.~Durhuus, and T.~Jónsson, ``{Three-Dimensional Simplicial
  Quantum Gravity and Generalized Matrix Models},'' {\em {Modern Physics
  Letters A}}, vol.~6, no.~12, pp.~1133--1146, 1991.

\bibitem{TensorModels2}
N.~Sasakura, ``{Tensor Models for Quantum Gravity and Orientability of
  Manifolds},'' {\em {Modern Physics Letters A}}, vol.~6, no.~28,
  pp.~2613--2623, 1991.

\bibitem{TensorModels3}
M.~Gross, ``{Tensor models and simplicial quantum gravity in $>$2-D},'' {\em
  {Nuclear Physics B - Proceedings Supplements}}, vol.~25A, pp.~144--149, 1992.

\bibitem{Chirco:2017wgl}
G.~Chirco, D.~Oriti, and M.~Zhang, ``{Ryu-Takayanagi Formula for Symmetric
  Random Tensor Networks},'' {\em Physical Review D}, vol.~97, no.~12,
  p.~126002, 2018.
\newblock \href{https://arxiv.org/abs/1711.09941v1}{arXiv:1711.09941}.

\bibitem{Chirco:2019dlx}
G.~Chirco, A.~Goe\ss{}mann, D.~Oriti, and M.~Zhang, ``{Group field theory and
  holographic tensor networks: dynamical corrections to the
  Ryu\textendash{}Takayanagi formula},'' {\em Classical and Quantum Gravity},
  vol.~37, no.~9, p.~095011, 2020.
\newblock \href{https://arxiv.org/abs/1903.07344}{arXiv:1903.07344}.

\bibitem{Colafranceschi:2021acz}
E.~Colafranceschi, G.~Chirco, and D.~Oriti, ``{Holographic maps from quantum
  gravity states as tensor networks},'' {\em Physical Review D}, vol.~105,
  no.~6, p.~066005, 2022.
\newblock \href{https://arxiv.org/abs/2105.06454}{arXiv:2105.06454}.

\bibitem{Colafranceschi:2022dig}
E.~Colafranceschi, S.~Langenscheidt, and D.~Oriti, ``{Holographic properties of
  superposed spin networks},'' 2022.
\newblock \href{https://arxiv.org/abs/2207.07625}{arXiv:2207.07625}.

\bibitem{Colafranceschi:2022ual}
E.~Colafranceschi and G.~Adesso, ``{Holographic entanglement in spin network
  states: A focused review},'' {\em AVS Quantum Science}, vol.~4, no.~2,
  p.~025901, 2022.
\newblock \href{https://arxiv.org/abs/2202.05116}{arXiv:2202.05116}.

\bibitem{GFTandLQG1}
D.~Oriti, ``{Group field theory as the 2nd quantization of Loop Quantum
  Gravity},'' {\em {Classical and Quantum Gravity}}, vol.~33, no.~8, 2016.
\newblock \href{https://arxiv.org/abs/1310.7786v2}{arXiv:1310.7786}.

\bibitem{GFTandLQG2}
D.~Oriti, J.~P. Ryan, and J.~Thürigen, ``{Group field theories for all loop
  quantum gravity},'' {\em {New Journal of Physics}}, vol.~17, no.~2, 2015.
\newblock \href{https://arxiv.org/abs/1409.3150v2}{arXiv:1409.3150}.

\bibitem{BoulatovModel}
D.~V. Boulatov, ``{A Model of Three-Dimensional Lattice Gravity},'' {\em
  {Modern Physics Letters A}}, vol.~7, no.~18, pp.~1629--1646, 1992.
\newblock \href{https://arxiv.org/abs/hep-th/9202074v1}{arXiv:hep-th/9202074}.

\bibitem{OoguriModel}
H.~Ooguri, ``{Topological Lattice Models in Four Dimensions},'' {\em {Modern
  Physics Letters A}}, vol.~7, no.~30, pp.~2799--2810, 1992.
\newblock \href{https://arxiv.org/abs/hep-th/9205090v1}{arXiv:hep-th/9205090}.

\bibitem{GurauColouredGFT}
R.~Gurau, ``{Colored Group Field Theory},'' {\em {Communications in
  Mathematical Physics}}, vol.~304, no.~1, pp.~69--93, 2011.
\newblock \href{https://arxiv.org/abs/0907.2582v1}{arXiv:0907.2582}.

\bibitem{BosonicGFT}
J.~{Ben Geloun}, J.~Magnen, and V.~Rivasseau, ``{Bosonic Colored Group Field
  Theory},'' {\em {The European Physical Journal C}}, vol.~70, no.~4,
  p.~1119–1130, 2010.
\newblock \href{https://arxiv.org/abs/0911.1719v1}{arXiv:0911.1719}.

\bibitem{GurauColouredGFTPseudo}
R.~Gurau, ``{Lost in Translation: Topological Singularities in Group Field
  Theory},'' {\em {Classical and Quantum Gravity}}, vol.~27, no.~23, 2010.
\newblock \href{https://arxiv.org/abs/1006.0714v1}{arXiv:1006.0714}.

\bibitem{GurauLargeN1}
R.~Gurau, ``{The 1/N expansion of colored tensor models},'' {\em Annales Henri
  Poincaré}, vol.~12, no.~5, 2011.
\newblock \href{https://arxiv.org/abs/1011.2726v2}{arXiv:1011.2726}.

\bibitem{GurauLargeN2}
R.~Gurau and V.~Rivasseau, ``{The 1/N expansion of colored tensor models in
  arbitrary dimension},'' {\em {EPL (Europhysics Letters)}}, vol.~95, no.~5,
  2011.
\newblock \href{https://arxiv.org/abs/1101.4182v1}{ arXiv:1101.4182}.

\bibitem{GurauLargeN3}
R.~Gurau, ``{The complete 1/N expansion of colored tensor models in arbitrary
  dimension},'' {\em {Annales Henri Poincaré}}, vol.~13, no.~3, pp.~399--423,
  2011.
\newblock \href{https://arxiv.org/abs/1102.5759v1}{arXiv:1102.5759}.

\bibitem{CritTM}
V.~Bonzom, R.~Gurau, A.~Riello, and V.~Rivasseau, ``{Critical behavior of
  colored tensor models in the large N limit},'' {\em {Nuclear Physics B}},
  vol.~853, no.~1, pp.~174--195, 2011.
\newblock \href{https://arxiv.org/abs/1105.3122v1}{arXiv:1105.3122}.

\bibitem{CritTM2}
V.~Bonzom, R.~Gurau, and V.~Rivasseau, ``{Random tensor models in the large N
  limit. Uncoloring the colored tensor models},'' {\em {Physical Review D}},
  vol.~85, no.~8, 2012.
\newblock \href{https://arxiv.org/abs/1202.3637v1}{arXiv:1202.3637}.

\bibitem{GFTDiff1}
F.~Girelli and E.~R. Livine, ``{A Deformed Poincaré Invariance for Group Field
  Theories},'' {\em {Classical and Quantum Gravity}}, vol.~27, no.~24, 2010.
\newblock \href{https://arxiv.org/abs/1001.2919v1}{arXiv:1001.2919}.

\bibitem{GFTDiff2}
A.~Baratin, F.~Girelli, and D.~Oriti, ``{Diffeomorphisms in group field
  theories},'' {\em {Physical Review D}}, vol.~83, no.~10, 2011.
\newblock \href{https://arxiv.org/abs/1101.0590v1}{arXiv:1101.0590}.

\bibitem{DiffeoSM}
L.~Freidel and D.~Louapre, ``{Diffeomorphisms and spin foam models},'' {\em
  {Nuclear Physics B}}, vol.~662, pp.~279--298, 2003.
\newblock \href{https://arxiv.org/abs/gr-qc/0212001}{arXiv:gr-qc/0212001}.

\bibitem{DiffeoReview}
B.~Dittrich, ``{Diffeomorphism symmetry in quantum gravity models},'' {\em
  {Advanced Science Letters}}, vol.~2, no.~2, pp.~151--163, 2009.
\newblock \href{https://arxiv.org/abs/0810.3594}{arXiv:0810.3594}.

\bibitem{GurauBook}
R.~Gurau, {\em Random Tensors}.
\newblock Oxford, United Kingdom: Oxford University Press, 2017.

\bibitem{GurauColouredTensorModelsReview}
R.~Gurau and J.~P. Ryan, ``{Colored Tensor Models - a Review},'' {\em
  {Symmetry, Integrability and Geometry: Methods and Applications (SIGMA)}},
  vol.~8, 2012.
\newblock \href{https://arxiv.org/abs/1109.4812v3}{ arXiv:1109.4812}.

\bibitem{GurauColouredTensorModelsReview2}
J.~P. Ryan, ``{(D+1)-Colored Graphs - a Review of Sundry Properties},'' {\em
  {Symmetry, Integrability and Geometry: Methods and Applications (SIGMA)}},
  vol.~12, 2016.
\newblock \href{https://arxiv.org/abs/1603.07220v2}{arXiv:1603.07220}.

\bibitem{Bonzom:2010zh}
V.~Bonzom and M.~Smerlak, ``{Bubble divergences from twisted cohomology},''
  {\em Communications in Mathematical Physics}, vol.~312, pp.~399--426, 2012.
\newblock \href{https://arxiv.org/abs/1008.1476}{arXiv:1008.1476}.

\bibitem{GurauBoundaryGraph}
R.~Gurau, ``{Topological Graph Polynomials in Colored Group Field Theory},''
  {\em {Annales Henri Poincaré}}, vol.~11, no.~4, pp.~565--584, 2010.
\newblock \href{https://arxiv.org/abs/0911.1945v1}{arXiv:0911.1945}.

\bibitem{DegreeOpenGraphs}
J.~{Ben Geloun} and V.~Rivasseau, ``{A Renormalizable 4-Dimensional Tensor
  Field Theory},'' {\em {Communications in Mathematical Physics }}, vol.~318,
  no.~1, pp.~69--109, 2013.
\newblock \href{https://arxiv.org/abs/1111.4997v3}{arXiv:1111.4997}.

\bibitem{GFTRen}
S.~Carrozza, D.~Oriti, and V.~Rivasseau, ``{Renormalization of an SU(2)
  Tensorial Group Field Theory in Three Dimensions},'' {\em {Communications in
  Mathematical Physics}}, vol.~330, no.~2, pp.~581--637, 2014.
\newblock \href{https://arxiv.org/abs/1303.6772}{arXiv:1303.6772}.

\bibitem{SenseTM}
V.~Rivasseau and F.~Vignes-Tourneret, ``{Can we make sense out of ``Tensor
  Field Theory''?},'' {\em {SciPost Physics Core}}, vol.~4, no.~4, 2021.
\newblock \href{https://arxiv.org/abs/2101.04970v2}{arXiv:2101.04970}.

\bibitem{GagliardiBoundaryGraph}
M.~Ferri, C.~Gagliardi, and L.~Grasselli, ``{A graph-theoretical representation
  of PL-manifolds - A survey on crystallizations},'' {\em {Aequationes
  mathematicae}}, vol.~31, pp.~121--141, 1986.

\bibitem{CTReview}
P.~Bandieri, M.~R. Casali, and C.~Gagliardi, ``{Representing manifolds by
  crystallization theory: foundations, improvements and related results},''
  {\em {Atti del Seminario Matematico e Fisico dell'Università di Modena}},
  vol.~49, pp.~283--337, 2001.

\bibitem{Review2018}
M.~R. Casali, P.~Cristofori, S.~Dartois, and L.~Grasselli, ``{Topology in
  colored tensor models via crystallization theory},'' {\em {Journal of
  Geometry and Physics}}, vol.~129, pp.~142--167, 2018.
\newblock \href{https://arxiv.org/abs/1704.02800}{arXiv:1704.02800}.

\bibitem{GFTVertex}
S.~Carrozza and D.~Oriti, ``{Bounding bubbles: the vertex representation of 3d
  Group Field Theory and the suppression of pseudo-manifolds},'' {\em {Physical
  Review D}}, vol.~85, no.~4, 2012.
\newblock \href{https://arxiv.org/abs/1104.5158v2}{arXiv:1104.5158}.

\bibitem{DePietriPetronio}
R.~{De Pietri} and C.~Petronio, ``{Feynman Diagrams of Generalized Matrix
  Models and the Associated Manifolds in Dimension 4},'' {\em {Journal of
  Mathematical Physics}}, vol.~41, no.~10, pp.~6671--6688, 2000.
\newblock \href{https://arxiv.org/abs/gr-qc/0004045v2}{arXiv:gr-qc/0004045}.

\bibitem{Freidel:2009hd}
L.~Freidel, R.~Gurau, and D.~Oriti, ``{Group field theory renormalization - the
  3d case: Power counting of divergences},'' {\em Physical Review D}, vol.~80,
  p.~044007, 2009.
\newblock \href{https://arxiv.org/abs/0905.3772}{arXiv:0905.3772}.

\bibitem{Carrozza:2016vsq}
S.~Carrozza, ``{Flowing in Group Field Theory Space: a Review},'' {\em
  {Symmetry, Integrability and Geometry: Methods and Applications (SIGMA)}},
  vol.~12, p.~070, 2016.
\newblock \href{https://arxiv.org/abs/1603.01902}{arXiv:1603.01902}.

\bibitem{Caravelli}
F.~Caravelli, ``{A simple proof of orientability in colored group field
  theory},'' {\em {SpringerPlus}}, vol.~1, 2012.
\newblock \href{https://arxiv.org/abs/1012.4087v2}{arXiv:1012.4087}.

\bibitem{Caravelli2}
F.~Caravelli, ``{GEMs and amplitude bounds in the colored Boulatov model},''
  {\em {Journal of Theoretical and Applied Physics}}, vol.~7, no.~63, 2013.
\newblock \href{https://arxiv.org/abs/1304.7730}{arXiv:1304.7730}.

\bibitem{Barbieri}
A.~Barbieri, ``{Quantum tetrahedra and simplicial spin networks},'' {\em
  {Nuclear Physics B}}, vol.~518, no.~3, pp.~714--728, 1998.
\newblock Preprint:
  \href{https://arxiv.org/abs/gr-qc/9707010}{arXiv:gr-qc/9707010}.

\bibitem{LinsCrystallization}
S.~Lins, ``{A simple proof of Gagliardi's handle recognition theorem},'' {\em
  {Discrete Mathematics}}, vol.~57, pp.~253--260, 1985.

\bibitem{Gagliardi87}
C.~Gagliardi, ``{On a Class of 3-Dimensional Polyhedra},'' {\em {Annali
  dell’Università’ di Ferrara}}, vol.~33, pp.~51--88, 1987.

\bibitem{SeifertTopology}
H.~Seifert and W.~Threlfall, ``{A textbook of topology},'' in {\em Seifert and
  Threlfall: A textbook of topology and Seifert: Topology of 3-dimensional
  fibered spaces} (J.~S. Birman and J.~Eisner, eds.), vol.~89 of {\em {Pure and
  Applied Mathematics}}, Academic Press, 1980.

\bibitem{GagliardiCorbodant}
C.~Gagliardi, ``{Cobordant crystallizations},'' {\em {Discrete Mathematics}},
  vol.~45, no.~1, pp.~61--73, 1983.

\bibitem{CritTM3}
R.~Gurau, ``{A generalization of the Virasoro algebra to arbitrary
  dimensions},'' {\em {Nuclear Physics B}}, vol.~852, no.~3, pp.~592--614,
  2011.
\newblock \href{https://arxiv.org/abs/1105.6072v1}{arXiv:1105.6072}.

\bibitem{Planar}
E.~Brézin, C.~Itzykson, G.~Parisi, and J.~B. Zuber, ``{Planar Diagrams},''
  {\em {Communications in Mathematical Physics}}, vol.~59, no.~1, pp.~35--51,
  1978.

\bibitem{Hudson}
J.~F.~P. Hudson, {\em {Piecewise Linear Topology}}.
\newblock New York, Amsterdam: {WA Benjamin Inc.}, 1969.

\bibitem{GagliardiExistence}
A.~Cavicchioli and C.~Gagliardi, ``{Crystallizations of PL-manifolds with
  connected boundary},'' {\em {Bollettino dell'Unione Matematica Italiana B}},
  vol.~17, pp.~902--917, 1980.

\bibitem{GagliardiMultRes}
M.~Ferri and C.~Gagliardi, ``{Multiple Residues in Dimension Three},'' {\em
  {Journal of Combinatorial Theory, Series B}}, vol.~44, no.~3, pp.~263--275,
  1988.

\bibitem{Pachner}
U.~Pachner, ``{P.L. Homeomorphic Manifolds are Equivalent by Elementary
  Shellings},'' {\em {European Journal of Combinatorics}}, vol.~12, no.~2,
  pp.~129--145, 1991.

\bibitem{GagliardiFerri}
M.~Ferri and C.~Gagliardi, ``{Crystallisation Moves},'' {\em {Pacific Journal
  of Mathematics}}, vol.~100, no.~1, pp.~85--103, 1982.

\bibitem{CasaliPL}
M.~R. Casali, ``{An equivalence criterion for PL-manifolds},'' {\em {Rendiconti
  Seminario Facoltà Scienze Università Cagliari}}, vol.~72, no.~2, 2002.

\bibitem{GabrielThesis}
G.~Schmid, ``{On 3-Dimensional Quantum Gravity and Quasi-Local Holography in
  Spin Foam Models and Group Field Theory},'' Master's thesis, jointly at
  Ludwig-Maximilians-Universit\"at M\"unchen and Technische Universit\"at
  M\"unchen, March 2022.
\newblock \href{https://arxiv.org/abs/2205.05079}{arXiv:2205.05079}.

\bibitem{Casali}
M.~R. Casali, ``{An equivalence criterion for 3-manifo1ds},'' {\em {Revista
  Matemática Complutense}}, vol.~10, no.~1, 1997.

\bibitem{Livine:2002rh}
E.~R. Livine and D.~Oriti, ``{Implementing causality in the spin foam quantum
  geometry},'' {\em Nuclear Physics B}, vol.~663, pp.~231--279, 2003.
\newblock \href{https://arxiv.org/abs/gr-qc/0210064}{arXiv:gr-qc/0210064}.

\bibitem{Oriti:2004mu}
D.~Oriti, ``{The Feynman propagator for spin foam quantum gravity},'' {\em
  Physical Review Letters}, vol.~94, 2005.
\newblock \href{https://arxiv.org/abs/gr-qc/0410134}{arXiv:gr-qc/0410134}.

\bibitem{Oriti:2006wq}
D.~Oriti and T.~Tlas, ``{Causality and matter propagation in 3-D spin foam
  quantum gravity},'' {\em Physical Review D}, vol.~74, 2006.
\newblock \href{https://arxiv.org/abs/gr-qc/0608116}{arXiv:gr-qc/0608116}.

\bibitem{Bianchi:2021ric}
E.~Bianchi and P.~Martin-Dussaud, ``{Causal structure in spin-foams},'' 2021.
\newblock \href{https://arxiv.org/abs/2109.00986}{arXiv:2109.00986}.

\bibitem{PenroseSpinNetwork}
R.~Penrose, ``{Angular momentum: an approach to combinatorial spacetime},'' in
  {\em {Quantum Theory and Beyond. Essays and Discussions Arising from a
  Colloquium}} (T.~Bastin, ed.), pp.~147--150, Cambridge, United Kingdom:
  Cambridge University Press, 1971.

\bibitem{BaezGaugeTheory}
J.~C. Baez, ``{Spin Network States in Gauge Theory},'' {\em {Advances in
  Mathematics}}, vol.~117, no.~2, pp.~253--272, 1996.
\newblock \href{https://arxiv.org/abs/gr-qc/9411007v1}{arXiv:gr-qc/9411007}.

\bibitem{BaezSpinFoams}
J.~C. Baez, ``{Spin Foam Models},'' {\em {Classical and Quantum Gravity}},
  vol.~15, no.~7, pp.~1827--1858, 1998.
\newblock \href{https://arxiv.org/abs/gr-qc/9709052}{arXiv:gr-qc/9709052}.

\bibitem{Hatcher}
A.~Hatcher, {\em Algebraic Topology}.
\newblock Cambridge, United Kingdom: Cambridge University Press, 2002.

\bibitem{KirbyAT}
R.~C. Kirby, ``{Stable Homeomorphisms and the Annulus Conjecture},'' {\em
  {Annals of Mathematics. Second Series}}, vol.~89, no.~3, pp.~575--582, 1969.

\bibitem{Quinn}
F.~Quinn, ``{Ends of maps. III. Dimensions 4 and 5},'' {\em {Journal of
  Differential Geometry}}, vol.~17, no.~3, pp.~503--521, 1982.

\bibitem{HandleBodiesCT}
B.~Basak and M.~Binjola, ``{Minimal crystallizations of 3-manifolds with
  boundary},'' 2020.
\newblock \href{https://arxiv.org/abs/2001.10214}{arXiv:2001.10214}.

\bibitem{ToAppearTorus}
C.~Goeller, D.~Oriti, and G.~Schmid, ``{Transition Amplitudes and Large N limit
  for a Toroidal Boundary in the Coloured Boulatov Model},''
\newblock \textit{To appear in 2022}.

\bibitem{ReducedDegree}
R.~Gurau and G.~Schaeffer, ``{Regular colored graphs of positive degree},''
  {\em {Annales de l’Institut Henri Poincaré D. Combinatorics, Physics and
  their Interactions}}, vol.~3, no.~3, p.~257–320, 2016.
\newblock \href{https://arxiv.org/abs/1307.5279v3}{arXiv:1307.5279}.

\bibitem{Pezzana}
M.~Pezzana, ``{Sulla struttura topologica delle varietà compatte},'' {\em
  {Atti del Seminario Matematico e Fisico dell'Università di Modena}},
  vol.~23, pp.~269--277, 1974.

\bibitem{Pezzana2}
M.~Pezzana, ``{Diagrammi di Heegaard e triangolazione contratta},'' {\em
  {Bollettino della Unione Matematica Italiana}}, vol.~12, Suppl. Fasc. 3,
  pp.~98--105, 1975.

\bibitem{GagliardiConnectedSum}
P.~Bandieri, C.~Gagliardi, and G.~Volzone, ``{Combinatorial Handles and
  Manifolds with Boundary},'' {\em {Journal of Geometry}}, vol.~46, no.~1,
  pp.~10--19, 1993.

\bibitem{CasaliGrasselli89}
M.~R. Casali and L.~Grasselli, ``{Representing Branched Coverings by
  Edge-Coloured Graphs},'' {\em {Topology and its Applications}}, vol.~33,
  pp.~197--207, 1989.

\bibitem{DehnTwist}
M.~Dehn, ``{Die Gruppe der Abbildungsklassen. Das arithmetische Feld auf
  Flächen},'' {\em Acta Mathematica}, vol.~69, pp.~135--206, 1938.

\end{thebibliography}
\end{document}